\documentclass[%
aps,
twocolumn,
superscriptaddress,
amsmath,amssymb,
pra,
floatfix,
]{revtex4-2}
\usepackage{url}
\usepackage{graphicx}% Include figure files
\usepackage{placeins}
\usepackage{dcolumn}% Align table columns on decimal point
\usepackage{bm}% bold math
\usepackage[colorlinks, linkcolor=red, anchorcolor=green, citecolor=blue]{hyperref}
\usepackage{physics}
\usepackage{dsfont}
\usepackage{tikz,pgfplots}
\usepackage{enumerate}
\usepackage{subfigure}
\usepackage{bbold}
\usepackage{makecell}
\usepackage{array}
\makeatletter
\newcommand{\thickhline}{%
    \noalign {\ifnum 0=`}\fi \hrule height 1pt
    \futurelet \reserved@a \@xhline
}
\newcolumntype{"}{@{\hskip\tabcolsep\vrule width 1pt\hskip\tabcolsep}}
\makeatother
\usepackage{multirow}
\usepackage{tabu}
\usepackage{textcomp,booktabs}
\usepackage{colortbl}
\usepackage[T1]{fontenc}
\definecolor{mygray}{gray}{0.9}
\definecolor{mypink}{rgb}{0.99,0.91,0.95}
\definecolor{mycyan}{cmyk}{0.3,0,0,0}

\newcommand{\id}{\text{id}}

\newcommand{\mE}{\mathcal{E}}
\newcommand{\mF}{\mathcal{F}}
\newcommand{\mG}{\mathcal{G}}

\newcommand{\mQ}{\mathcal{Q}}
\newcommand{\mH}{\mathcal{H}}

\newcommand{\T}{\mathbf{T}}

\newcommand{\1}{\mathbb{1}}

\newcommand{\tvec}{\text{vec}}
\newcommand{\tmat}{\text{mat}}
\newcommand{\swap}{\text{SWAP}}

\usepackage{amsthm}
\newtheorem*{thm*}{Theorem}
\newtheorem{thm}{Theorem}
\newtheorem{cor}{Corollary}
\newtheorem{lem}{Lemma}

\begin{document}

%\preprint{APS/123-QED}

\title{
Superchannel without Tears:
A Generalized Occam's Razor for Quantum Processes
}

\author{Yunlong Xiao}
\email{mathxiao123@gmail.com}
\affiliation{Institute of High Performance Computing (IHPC), Agency for Science, Technology and Research (A*STAR), 1 Fusionopolis Way, \#16-16 Connexis, Singapore 138632, Republic of Singapore}
%\affiliation{Quantum Innovation Centre (Q.InC), Agency for Science Technology and Research (A*STAR), 2 Fusionopolis Way, Innovis \#08-03, Singapore 138634, Republic of Singapore}

\newcommand{\hilb}{\mathcal{H}}

\date{\today}
             
%%%%%%%%%%%%%%%%%%%%%%%%%%%%%%%%%%%%%%%%%%%%%%%%%%%%
\begin{abstract} 
Quantum channels function as the operational primitives of quantum theory, while superchannels describe the most general transformations acting upon them.
Yet the prevailing framework for superchannels is both internally inconsistent, owing to the coexistence of distinct Choi operator constructions, and structurally incomplete, lacking the analogue of representations that ground channel theory.
We resolve these issues by combining tensor-network methods with a generalized Occam's razor introduced here, establishing a unified foundation for superchannels.
Our framework establishes the connections between competing Choi formulations, develops the Kraus, Stinespring, and Liouville representations for superchannels, and provides a simplified derivation of the realization theorem that identifies the minimal memory required to implement a given transformation. 
These structural tools also enable characterizations of superchannels that destroy quantum correlations or causal structure, opening a systematic route to non-Markovian quantum dynamics.
\end{abstract}
%%%%%%%%%%%%%%%%%%%%%%%%%%%%%%%%%%%%%%%%%%%%%%%%%%%%

\maketitle
%\tableofcontents

%\newpage

%%%%%%%%%%%%%%%%%%%%%%%%%%%%%%%%%%%%%%%%%%%%%%%%%%%%%%%%%%%%%%%%%%%%%%%
%%%%%%%%%%%%%%%%%%%%%%%%%%%%%%%%%%%%%%%%%%%%%%%%%%%%%%%%%%%%%%%%%%%%%%%

\section{Introduction}\label{sec:Intro}

Quantum channels provide not only the elementary ``Lego pieces'' of quantum theory and emerging quantum technologies, but also the fundamental carriers of quantum resources that enable phenomena with no classical analogue. 
In quantum communication~\cite{Gisin2007}, the noise channel itself is the resource we must contend with~\cite{holevo1973bounds,PhysRevA.55.1613,PhysRevA.56.131,PhysRevLett.83.3081,4626055,7115934,PhysRevLett.124.120502}; 
in fault-tolerant quantum computing~\cite{548464,Steane1999,PRXQuantum.5.020101}, non-Clifford gates such as the T gate supply the universality~\cite{gottesman1997stabilizercodesquantumerror,gottesman1998heisenbergrepresentationquantumcomputers,Nielsen_Chuang_2010}; 
and in quantum algorithms~\cite{RevModPhys.68.733,RevModPhys.82.1,Montanaro2016,childs2017lecture,Dalzell2025}, oracles act as coherent encoding of functions that can be queried in superposition~\cite{Grover1996,Shor1997,Simon1997}. 
What unifies all these ingredients is that they are, at their core, quantum channels -- completely positive trace-preserving (CPTP) maps whose behavior sets the objective limits of what quantum systems can achieve when we use them only as given. 
Decades of work have endowed the explorations of quantum channels with a mature and versatile toolkit, spanning the Choi–Jamio\l kowski isomorphism, Kraus decompositions, Stinespring dilation, and the Liouville superoperator, each offering a distinct window into the structure and behavior of the underlying quantum dynamics.

However, when we begin to manipulate these channels as dynamical resources in their own right, we move beyond the objective limit and toward a more subjective one, where higher-order transformations allow us to reshape, refine, or even repurpose the dynamics themselves. 
This transition brings us to the framework of superchannels, which formalize such transformations and capture, in full generality, how one quantum process can convert a quantum channel into another.
Physically, a quantum superchannel can be realized as two sequential channels linked by a quantum memory, enabling interactions across multiple time points. 
Mathematically, however, the theory is anchored almost entirely in its Choi operators. Unlike channels, superchannels admit two different Choi operator constructions; 
one obtained by feeding unnormalized maximally entangled states into all inputs~\cite{PhysRevLett.101.060401}, the other by applying the superchannel to a basis of linear maps~\cite{8678741}. 
Both have been widely used, yet their relationship has remained unclear, leaving the current framework conceptually inconsistent.
Beyond these Choi formulations, other structural representations remain comparatively underdeveloped, highlighting that the current framework for superchannels remains incomplete at a fundamental level.

\begin{figure*}[t]
    \centering   
    \includegraphics[width=1\textwidth]{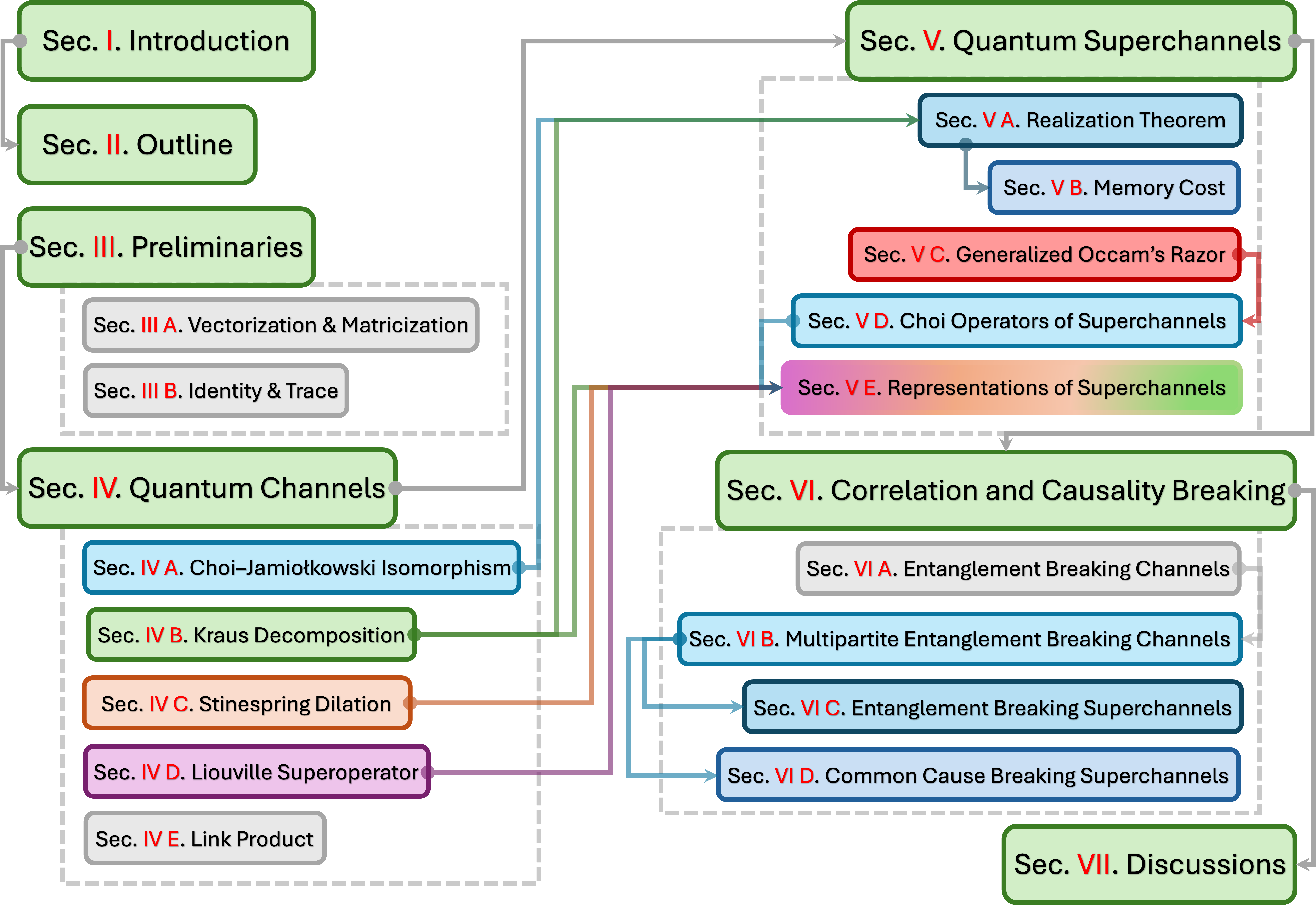}
    \caption{(Color online) \textbf{Schematic Overview of the Work}.
    Sec.~\ref{sec:Intro} sets the introduction, and Sec.~\ref{sec:Outline} outlines the structure of the work. 
    Secs.~\ref{sec:Pre} and~\ref{sec:QChannels} establish the foundational tools: tensor-network notation and the unified representations of quantum channels. 
    Sec.~\ref{sec:QSuperchannels} develops the superchannel framework -- resolving the Choi-level inconsistency, completing the structural representations, and determining memory cost through the realization theorem. 
    Sec.~\ref{sec:EB} applies this framework to superchannels that break quantum correlations or causal structure, while Sec.~\ref{sec:Discussions} provides the concluding perspective.
   }
    \label{fig:Outline}
\end{figure*}

In this work, we resolve both the inconsistency and the incompleteness of the superchannel theory. 
Two conceptual ``razors'' guide our approach. 
The first is the use of tensor-network representations as a structural Occam's razor: 
they provide a transparent language that simplifies the realization theorem for superchannels and offers clear physical insight into the memory cost required to simulate them. The second is a generalized Occam's razor introduced here, which asserts that when two competing formulations exhibit comparable complexity, preference should be given to the one that aligns most naturally with established theories, for no framework stands entirely independent of the broader theoretical landscape.
We show that the two Choi constructions of superchannels, proposed in Ref.~\cite{PhysRevLett.101.060401} and~\cite{8678741}, are equivalent up to system permutations, and this equivalence, together with the generalized Occam's razor, selects the formulation obtained by feeding unnormalized maximally entangled states into all input systems, as it is rooted directly in well-established channel theory. 
This resolves the inconsistency of the superchannel theory. With this foundation in place, we construct the Kraus decomposition, Stinespring dilation, and Liouville superoperator for superchannels, completing the structural toolkit. 
As an application, we introduce superchannels that break quantum correlations or causality, and provide complete characterizations for them. 

%%%%%%%%%%%%%%%%%%%%%%%%%%%%%%%%%%%%%%%%%%%%%%%%%%%%%%%%%%%%%%%%%%%%%%%
%%%%%%%%%%%%%%%%%%%%%%%%%%%%%%%%%%%%%%%%%%%%%%%%%%%%%%%%%%%%%%%%%%%%%%%

\section{Outline}\label{sec:Outline}

This section offers a concise overview of the paper's structure and central results, serving as a guide for navigating the development of our framework.
Fig.~\ref{fig:Outline} presents a schematic summary.
The work is organized into three parts, each with its own thematic focus:

\begin{itemize}
  \item Foundational Preparations. 
  In the first part, we lay the structural foundations of our framework, introducing the essential preparations and conventions on which the later developments rely.
  Sec.~\ref{sec:Pre} introduces the tensor-network notation and conventions (see Sec.~\ref{subsec:Vec}) that underpin the technical development of our framework, together with a set of diagrammatic rules (see Sec.~\ref{subsec:Id&Tr}).
  Building on this foundation,
  Sec.~\ref{sec:QChannels} present a unified graphical treatment of the four equivalent representations of quantum channels -- the Choi–Jamio\l kowski isomorphism (see Sec.~\ref{subsec:Choi}), Kraus decomposition (see Sec.~\ref{subsec:Kraus}), Stinespring dilation (see Sec.~\ref{subsec:Stinespring}), and the Liouville superoperator (see Sec.~\ref{subsec:Liouville}) -- as well as the composition of channels via the link product (see Sec.~\ref{subsec:Link_Product}).
  \item Superchannel Framework.
  In the second part, we address both the inconsistency and the incompleteness of the theory of superchannels. 
  Sec.~\ref{subsec:Realization} employs our first ``razor'', tensor-network methods, to provide an alternative and more transparent derivation of the realization theorem, and to furnish an intuitive route for determining the memory cost required to simulate a given superchannel, further developed in Sec.~\ref{subsec:Memory_Cost}. 
  Before proceeding, Sec.~\ref{subsec:Bi_Channels} introduces our second ``razor'', the generalized Occam's razor, which guides the choice of a canonical Choi operator on which the complete representational framework is built. 
  In Sec.~\ref{subsec:SC_Choi}, we show that the two commonly used Choi constructions for superchannels are equivalent up to system permutations, and the generalized Occam's razor singles out the formulation that takes a bipartite-channel perspective as the natural canonical choice. 
  With this foundation in place, Sec.~\ref{subsec:SC_Representations} develops the remaining three structural representations, Kraus decomposition (see Thm.~\ref{thm:Superchannel_Krasu}), Stinespring dilation (see Thm.~\ref{thm:Superchannel_Stinespring}), and the Liouville superoperator (see Thm.~\ref{thm:Superchannel_Liouville}), thereby completing the representational toolkit for superchannels.
  \item Correlation and Causality Breaking.
  In the final part, we apply the framework developed in this work to investigate superchannels that break quantum correlations and causality. 
  Sec.~\ref{subsec:EBC} revisits the notion of entanglement breaking channels, after which Sec.~\ref{subsec:MEBC} extends the discussion to multipartite settings. 
  Building on these ideas, Sec.~\ref{subsec:EBS} introduces a new form of entanglement breaking superchannel motivated by practical quantum memory architectures and establishes their structural characterizations. 
  Finally, Sec.~\ref{subsec:CCBS} turns from correlations to causality, where we introduce and analyze the concept of common cause breaking superchannels.
\end{itemize}

The framework developed here places superchannel theory on firm foundations, and for readers new to the subject, this work may also serve as an invitation to the wonderland of superchannels and higher-order quantum information processing.

%%%%%%%%%%%%%%%%%%%%%%%%%%%%%%%%%%%%%%%%%%%%%%%%%%%%%%%%%%%%%%%%%%%%%%%
%%%%%%%%%%%%%%%%%%%%%%%%%%%%%%%%%%%%%%%%%%%%%%%%%%%%%%%%%%%%%%%%%%%%%%%

\section{Preliminaries}\label{sec:Pre}

In this work, we adopt an intuitive yet rigorous perspective on quantum processes, with a focus on quantum channels and superchannels. 
To keep the underlying structure transparent, we use the tensor-network formalism, whose diagrammatic language allows us to express the key constructions in a clean and compact way. 
This section introduces the essential concepts and conventions that will serve as the foundation for the developments that follow.
Readers already comfortable with these notions may wish to skip this section and proceed directly to the next.

%%%%%%%%%%%%%%%%%%%%%%%%%%%%%%%%%%%%%%%%%%%%%%%%%%%%%%%%%%%%%%%%%%%%%%%
%%%%%%%%%%%%%%%%%%%%%%%%%%%%%%%%%%%%%%%%%%%%%%%%%%%%%%%%%%%%%%%%%%%%%%%

\subsection{Vectorization and Matricization}\label{subsec:Vec}

Central to our framework is the use of {\it vectorization} and {\it matricization}, expressed in the language of tensor networks, which offers a transparent and visually intuitive way to represent and manipulate quantum processes. 
This viewpoint is particularly useful when establishing the realization theorem for superchannels (see Thm.~\ref{thm:RT}) and analyzing structural features of quantum channels (see Sec.~\ref{sec:QChannels}). 
Alongside familiar operations such as the partial trace and partial transpose, we introduce their higher-order counterparts, {\it partial vectorization} (see Eq.~\eqref{TN:Partial_Vectorization}) and {\it partial matricization} (see Eq.~\eqref{TN:Partial_Matricization}), which serve as technical building blocks throughout the remainder of the work. 

Consider an operator $M$, taken here simply as a matrix, that maps a vector in the Hilbert space $\mH_{A}$ to one in $\mH_{B}$. 
In the tensor-network representation, such an operator is drawn as a box with two legs, an input leg on the left and an output leg on the right, following the left-to-right flow of time used throughout this work.
This is illustrated on the left-hand side of Eq.~\eqref{TN:Vectorization} below. 
\begin{align}\label{TN:Vectorization}
    \raisebox{0ex}{\includegraphics[height=7em]{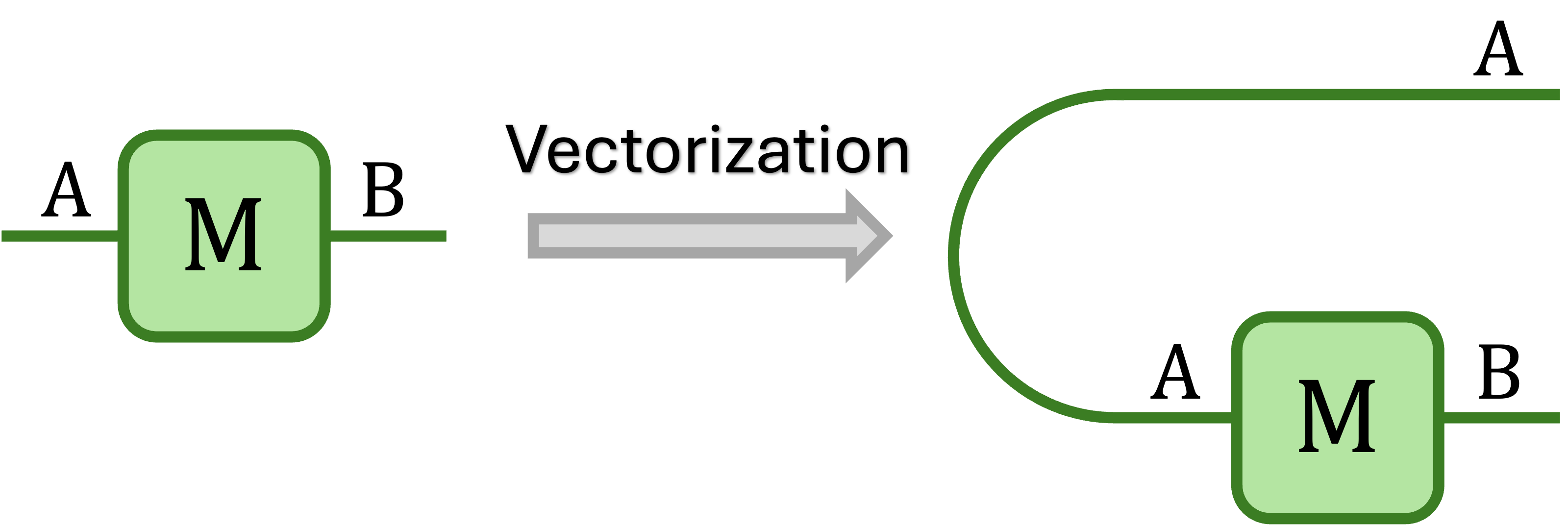}}
\end{align}
The curved arc on the right-hand side of the diagram depicts the unnormalized maximally entangled state
\begin{align}\label{eq:UMES}
    \ket{\Gamma}_{AA}:=\sum_{i}\ket{ii}_{AA}.
\end{align}
For simplicity, we will henceforth use $A$ and $B$ to denote the corresponding Hilbert spaces. 
The vectorization of $M$, denoted $\tvec(M)$, is defined as
\begin{align}
    \tvec(M):=
    \left(\1_{A}\otimes M\right) \ket{\Gamma}_{AA},
\end{align}
where $\1$ is the identity matrix. 
The corresponding tensor-network representation is shown on the right-hand side of Eq.~\eqref{TN:Vectorization}.
Note that, in some papers, this construction is alternatively written as $|M\rangle\!\rangle$.
A basic identity we will use repeatedly is the following
\begin{align}\label{TN:T}
    \raisebox{0ex}{\includegraphics[height=16em]{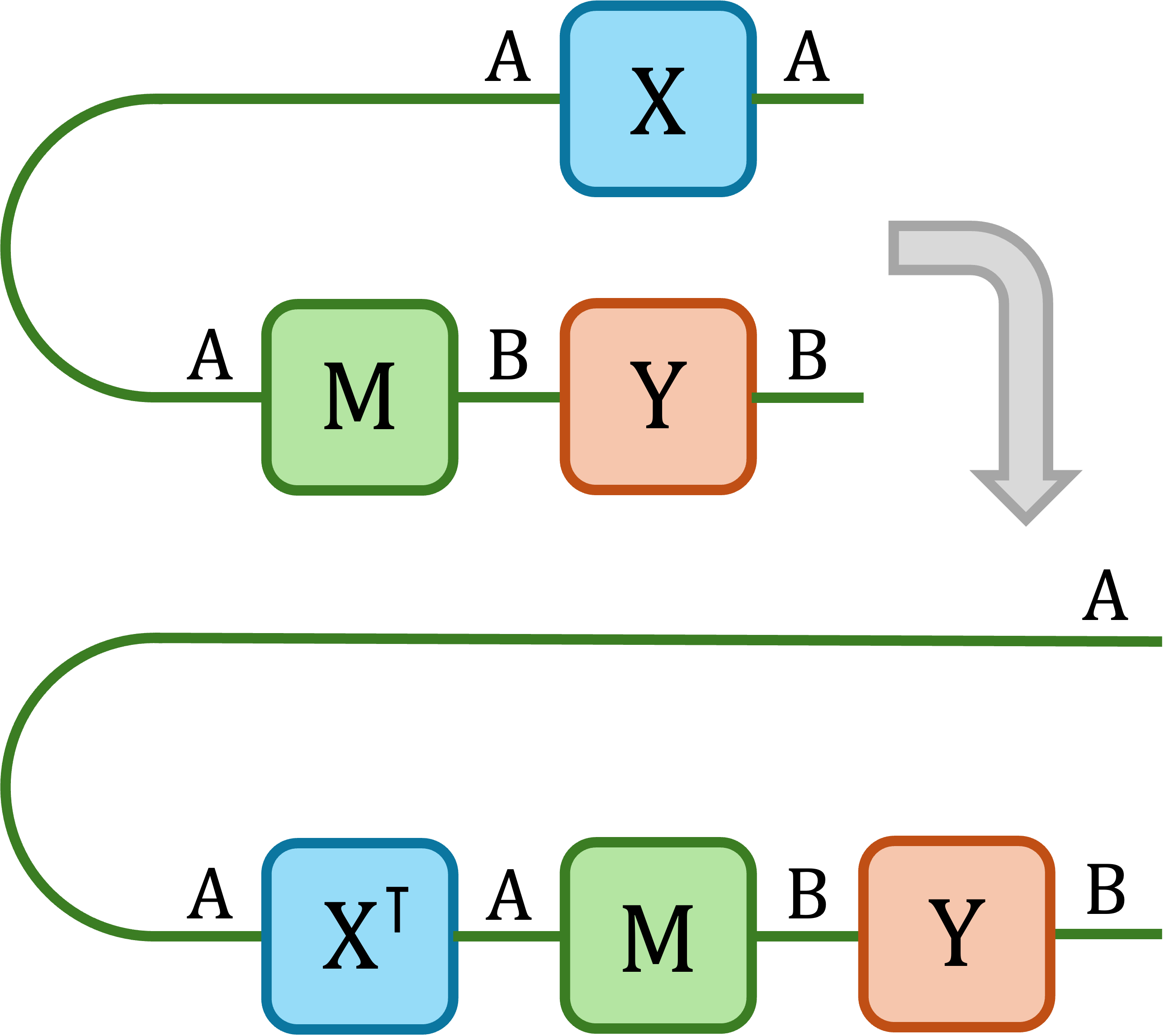}}
\end{align}
when a matrix is moved across a curved arc in the diagrammatic calculus, it is converted into its transpose. 
In algebraic form, this reads
\begin{align}\label{eq:T}
    (X\otimes Y)\tvec(M)=
    \tvec(YMX^{\T}).
\end{align}

More structure emerges when $M$ acts on multiple subsystems. 
A representative case is shown on the left-hand side of Eq.~\eqref{TN:Partial_Vectorization}, where the map sends vectors from the systems $A_1 A_2$ to $B_1 B_2$. 
Tensor product symbol $\otimes$ between subsystems are omitted.
\begin{align}\label{TN:Partial_Vectorization}
    \raisebox{0ex}{\includegraphics[height=9em]{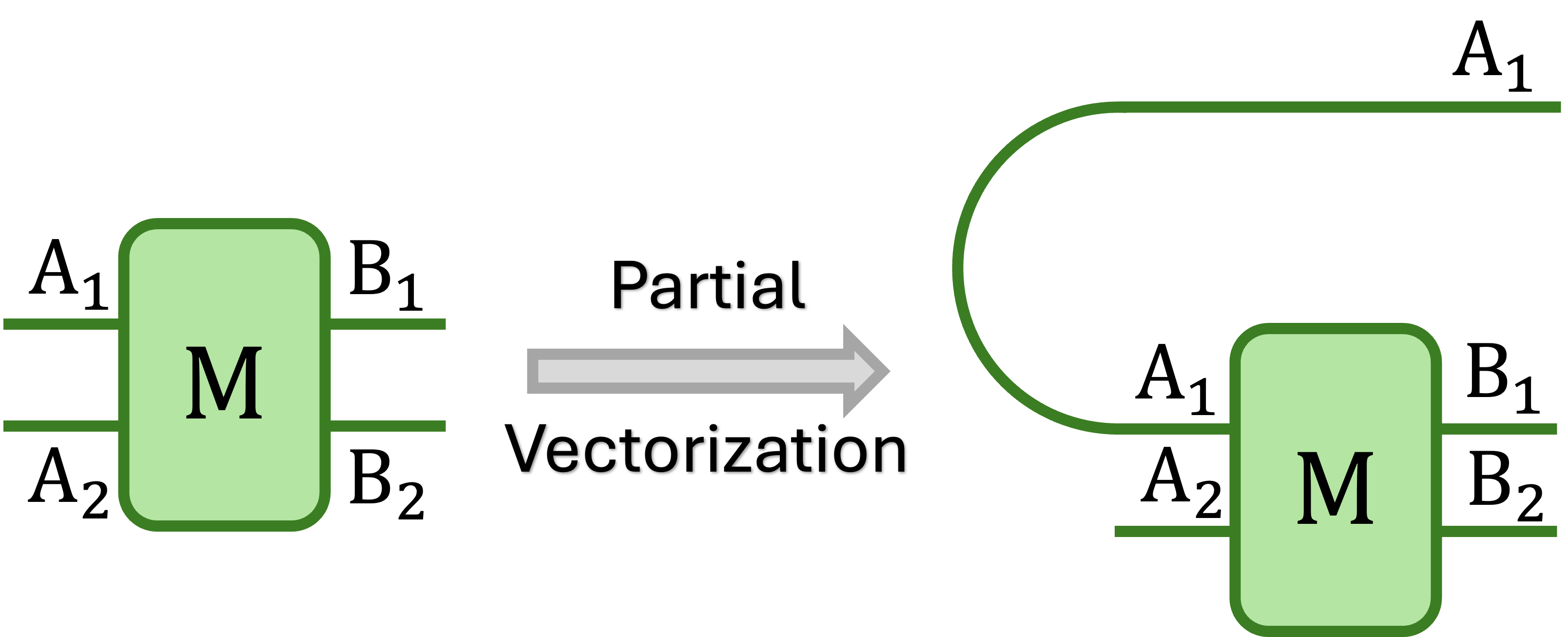}}
\end{align}
In this multiple subsystems scenario, we may also choose to vectorize only a selected subsystem, say $A_1$, yielding the partially vectorized operator $\tvec_{A_1}(M)$. 
\begin{align}
    \tvec_{A_1}(M):=
    \left(\1_{A_1}\otimes M \right) \ket{\Gamma}_{A_1A_1},
\end{align}
Its tensor-network representation is illustrated on the right-hand side of Eq.~\eqref{TN:Partial_Vectorization}, highlighting how subsystem-specific vectorization naturally extends the construction.
Remark that, here $\tvec_{A_1}(M)$ remains an operator, now acting on vectors in $A_2$ and producing outputs on the systems $A_1B_1B_2$.
Applying the same vectorization procedure to subsystem $A_2$ leads to the {\it full vectorization} of $M$, 
\begin{align}
    \tvec(M)
    :=
    \left(\1_{A_1A_2}\otimes M\right) \ket{\Gamma}_{A_1A_1}\otimes\ket{\Gamma}_{A_2A_2},
\end{align}
as demonstrated in Eq.~\eqref{TN:Full_Vectorization}.
\begin{align}\label{TN:Full_Vectorization}
    \raisebox{0ex}{\includegraphics[height=9em]{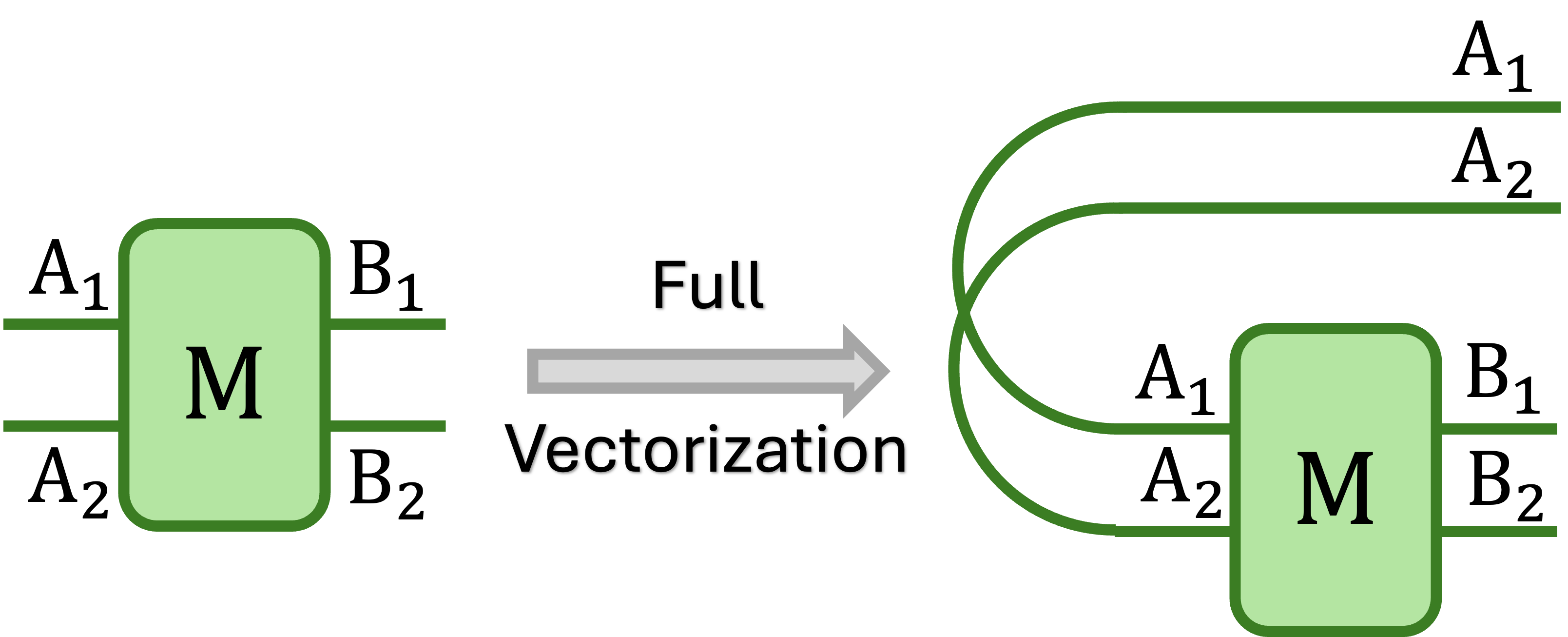}}
\end{align}

Up to this point, all forms of vectorization introduced in Eqs.~\eqref{TN:Vectorization},~\eqref{TN:Partial_Vectorization},~\eqref{TN:Full_Vectorization}, whether partial or full, should be understood not only as mappings from a matrix to a vector or another matrix, but as bijective constructions. 
Each vectorized object uniquely corresponds to the original expression on the left-hand side of Eqs.~\eqref{TN:Vectorization},~\eqref{TN:Partial_Vectorization},~\eqref{TN:Full_Vectorization} and can be recovered from it without any ambiguity.

More precisely, we define the inverse of vectorization, referred to as {\it matricization} and denoted $\tmat$.
Given a vector $\tvec(M)$ defined on systems $AB$ (see Eq.~\eqref{TN:Vectorization}), we matricize it with respect to subsystem $A$ by
\begin{align}\label{eq:Matricization}
    \tmat_{A}(\tvec(M)):=\bra{\Gamma}_{AA}
    \left(\1_{A}\otimes \tvec(M)\right)=M,
\end{align}
The resulting object is an operator that maps vectors from system $A$ to system $B$, as depicted in the tensor-network representation below
\begin{align}\label{TN:Matricization}
    \raisebox{0ex}{\includegraphics[height=10em]{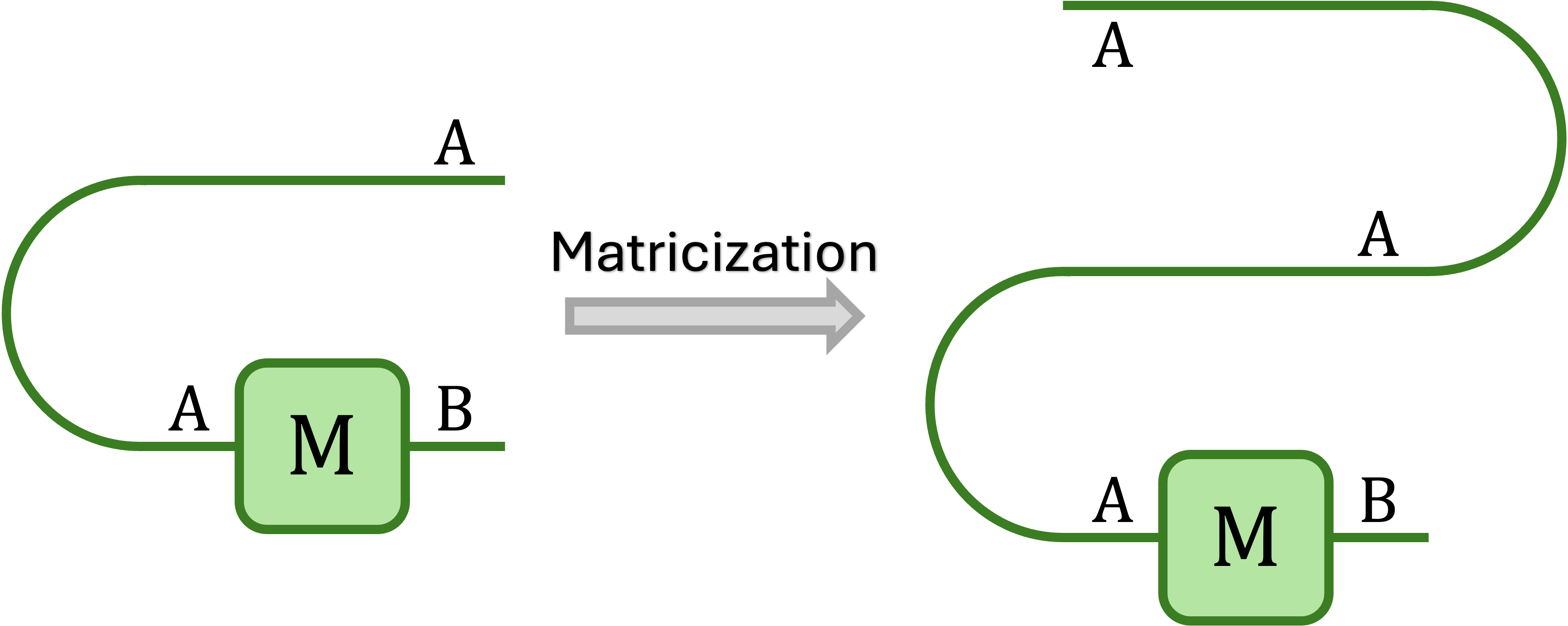}}
\end{align}
For the identity in Eq.~\eqref{eq:Matricization}, we have used the snake equation, which we will set out explicitly in Sec.~\ref{subsec:Id&Tr}.

Following the same idea, we can extend this notion to {\it partial matricization} for operators acting on multiple subsystems.
For the operator $M$ in Eq.~\eqref{TN:Partial_Vectorization}, we define its partial matricization over system $B_1$ as
\begin{align}\label{eq:Partial_Matricization}
    \tmat_{B_1}(M):=\bra{\Gamma}_{B_1B_1}
    \left(\1_{B_1}\otimes M\right),
\end{align}
its tensor-network representation is shown in the following diagram
\begin{align}\label{TN:Partial_Matricization}
    \raisebox{0ex}{\includegraphics[height=9em]{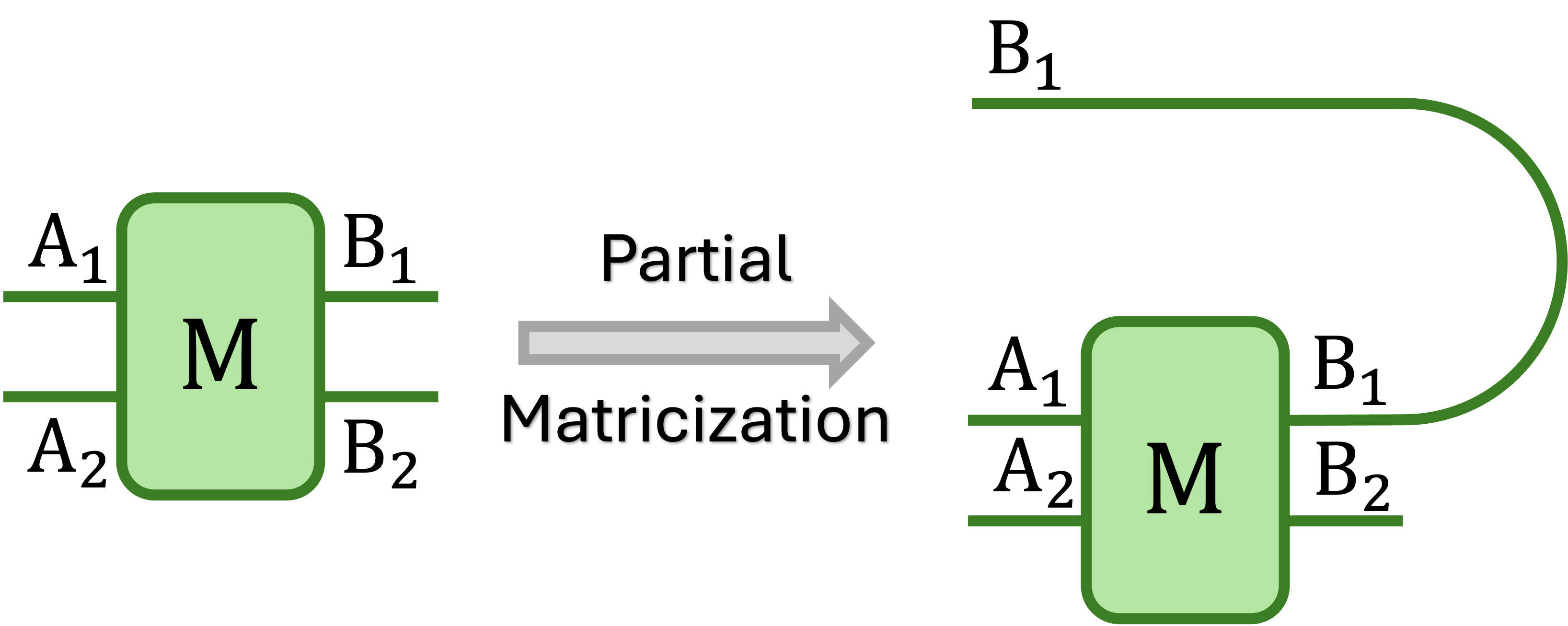}}
\end{align}
In an analogous manner, one can also define the {\it full matricization}. 
These notations will be essential in later sections, especially in our investigation of the Kraus operators of quantum superchannels;
see Sec.~\ref{subsec:SC_Representations} for further details.

%%%%%%%%%%%%%%%%%%%%%%%%%%%%%%%%%%%%%%%%%%%%%%%%%%%%%%%%%%%%%%%%%%%%%%%
%%%%%%%%%%%%%%%%%%%%%%%%%%%%%%%%%%%%%%%%%%%%%%%%%%%%%%%%%%%%%%%%%%%%%%%

\subsection{Identity and Trace}\label{subsec:Id&Tr}

Beyond vectorization, we introduce two additional notions and their representations in tensor-networks: the identity operator, and the trace operation. 
The identity operator is captured by the well-known snake (or zig-zag) equation. 
\begin{align}\label{TN:Snake}
    \raisebox{0ex}{\includegraphics[height=7em]{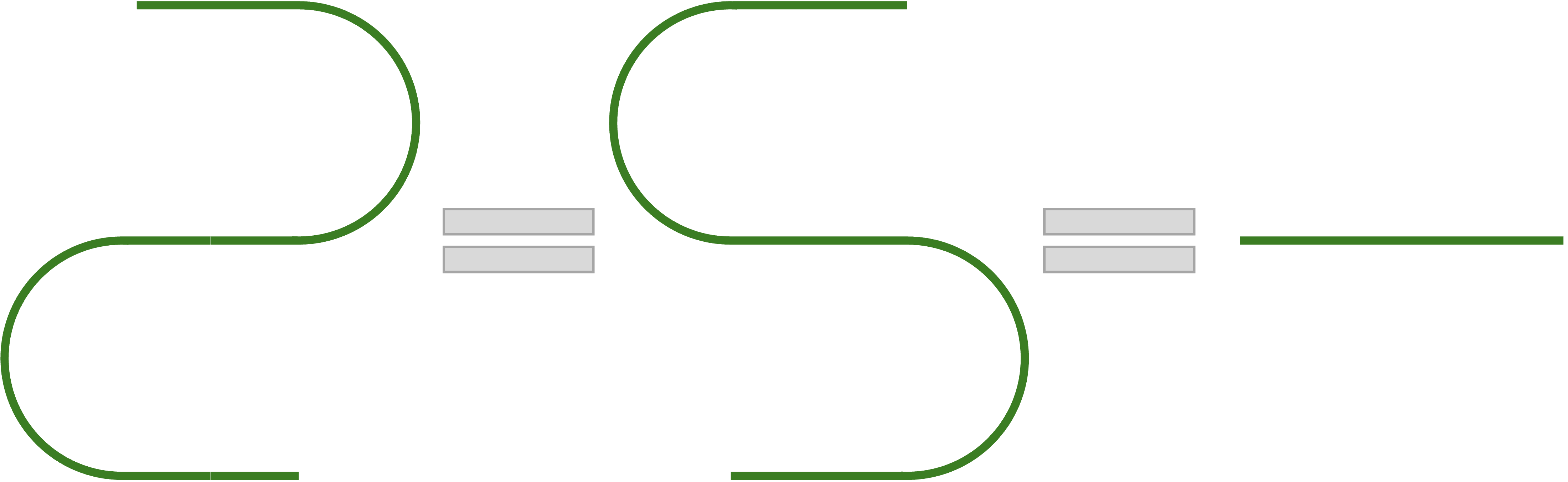}}
\end{align}
In this diagrammatic identity, the left arc represents the unnormalized maximally entangled state $\ket{\Gamma}$, whereas the right arc depicts its adjoint $\bra{\Gamma}$.
Using the algebraic decomposition $I=\sum_i\ketbra{i}{i}$, the identity operator can be recast in a tensor-network form where repeated indices $i$ are implicitly summed. 
\begin{align}\label{TN:Identity}
    \raisebox{0ex}{\includegraphics[height=9.5em]{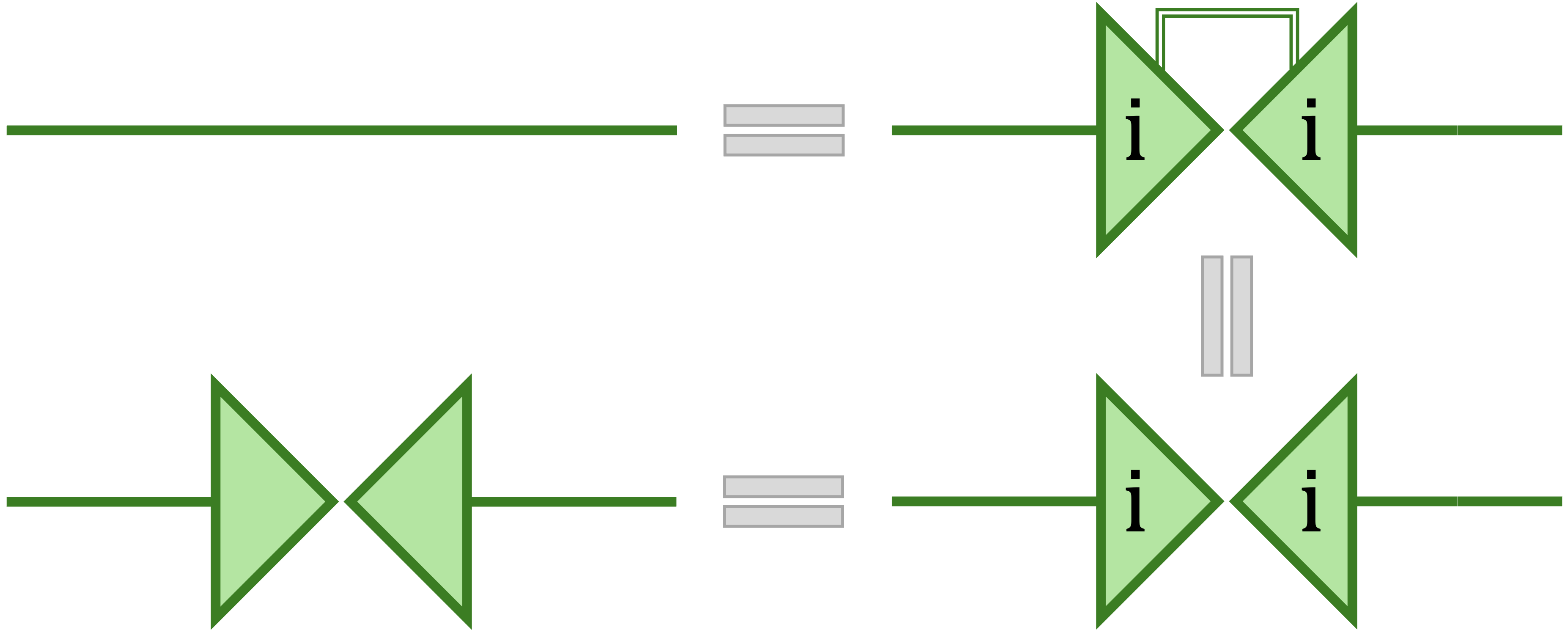}}
\end{align}
Here, A triangle oriented with its leg to the right stands for $\ket{i}$; reversing the orientation yields its dual $\bra{i}$.

For visual simplicity, in Eq.~\eqref{TN:Identity}, we omit the double line that would normally denote this contraction.
Likewise, we suppress the explicit index labels inside the triangles and use color to indicate that they share the same summed index. 
As an illustration, when two identity operators act on different systems, their tensor-network representations take the following form.
\begin{align}\label{TN:2Identities}
    \raisebox{0ex}{\includegraphics[height=7.5em]{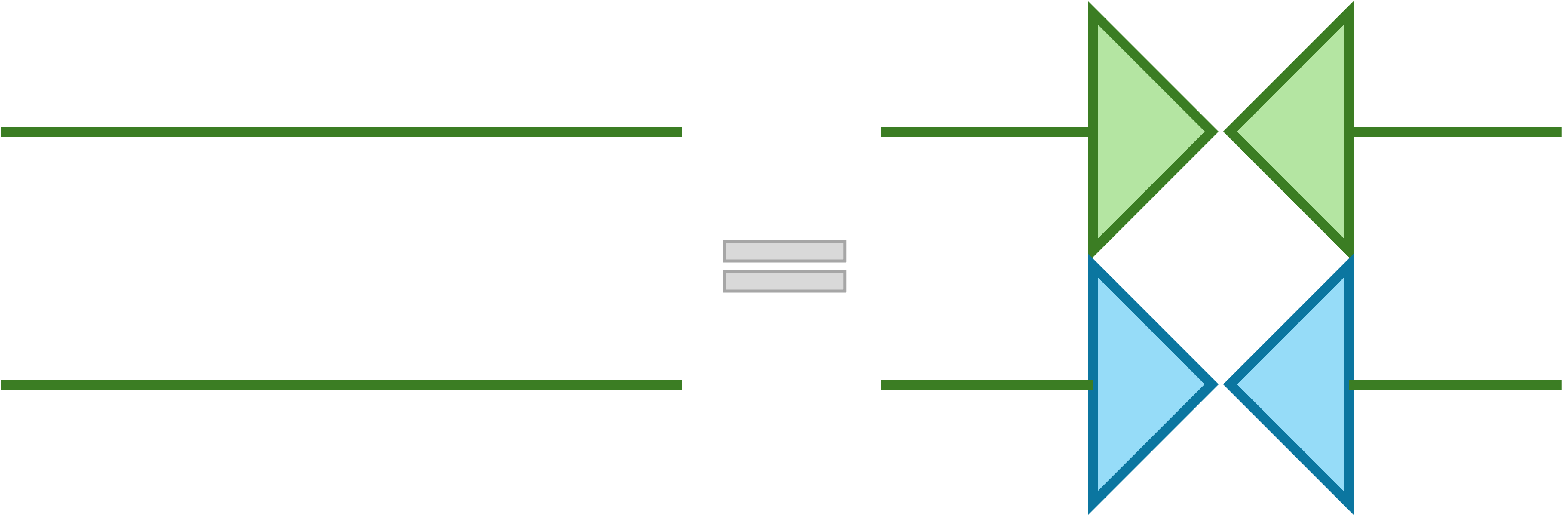}}
\end{align}

Another useful notion is the trace operation, expressible as $\Tr[\cdot]=\sum_i\bra{i}\cdot\ket{i}$, whose diagrammatic form is shown below.
\begin{align}\label{TN:Trace}
    \raisebox{0ex}{\includegraphics[height=4em]{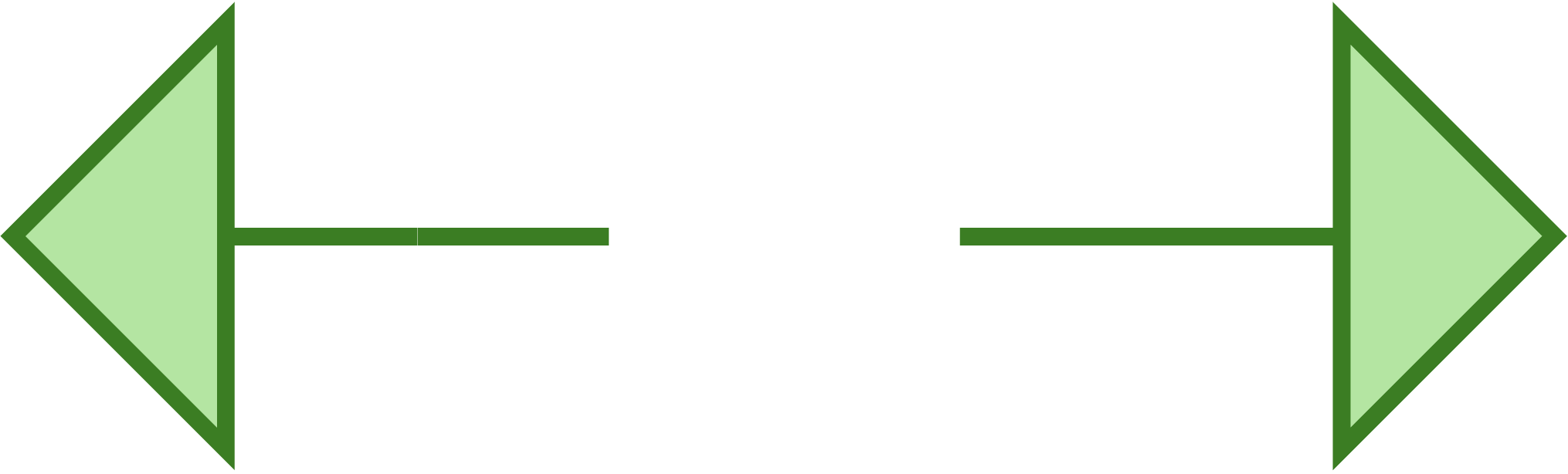}}
\end{align}

For readers seeking a more detailed background on tensor-networks, Refs.~\cite{wood2015tensornetworksgraphicalcalculus,Coecke_Kissinger_2017,Bridgeman_2017,biamonte2020lecturesquantumtensornetworks,Collura_2024} provide comprehensive and accessible introductions. 
We note, however, that the orientation conventions for vectors and the placement used in vectorization differ slightly across the literature. 
In particular, some works attach the operator being vectorized to the upper leg of $\ket{\Gamma}$ (see Eq.~\eqref{TN:Vectorization}), whereas in our convention it is placed on the lower leg.
Likewise, certain graphical calculi depict bras using triangles whose outgoing leg points to the right -- opposite to the convention adopted here. 
These differences are purely notational and should be clear from context.

%%%%%%%%%%%%%%%%%%%%%%%%%%%%%%%%%%%%%%%%%%%%%%%%%%%%%%%%%%%%%%%%%%%%%%%
%%%%%%%%%%%%%%%%%%%%%%%%%%%%%%%%%%%%%%%%%%%%%%%%%%%%%%%%%%%%%%%%%%%%%%%

\section{Quantum Channels}\label{sec:QChannels}

In this section, we establish the theoretical foundation of our work by revisiting the concept of quantum channels. 
As the fundamental carriers of physical transformations, quantum channels provide a unified description of state preparation, coherent evolution, and quantum measurement -- capturing every admissible process allowed by quantum mechanics.
Mathematically, a quantum channel is a linear map that is completely positive (CP) and trace-preserving (TP). 
Depending on the perspective taken, such maps admit several equivalent representations. 
The most commonly used are the Choi–Jamio\l kowski isomorphism, the Kraus decomposition, and the Stinespring dilation, and the Liouville superoperator (or natural representation, see Ref.~\cite{Watrous_2018} for details). 
Each highlights a different structural aspect of the underlying transformation.
In addition to these representations, we also introduce the composition of quantum channels via the link product, which serves as a fundamental tool in our subsequent analysis.

%%%%%%%%%%%%%%%%%%%%%%%%%%%%%%%%%%%%%%%%%%%%%%%%%%%%%%%%%%%%%%%%%%%%%%%
%%%%%%%%%%%%%%%%%%%%%%%%%%%%%%%%%%%%%%%%%%%%%%%%%%%%%%%%%%%%%%%%%%%%%%%

\subsection{Choi–Jamio\l kowski Isomorphism}\label{subsec:Choi}

A powerful and widely used way in quantum information theory is to represent a quantum channel by a matrix via an isomorphism. 
This correspondence, known as the Choi–Jamio\l kowski isomorphism~\cite{JAMIOLKOWSKI1972275,CHOI1975285}, establishes a one-to-one link between CPTP maps and bipartite quantum states, allowing structural properties of the channel to be read directly from its associated operator.
For a given linear map $\mE: A\to B$, the isomorphism assigns it to its corresponding Choi operator $J^{\mE}_{AB}$, defined as
\begin{align}\label{eq:Choi}
    J^{\mE}_{AB}:=\id_{A\to A}\otimes\mE_{A\to B}(\Gamma_{AA}),
\end{align}
where $\Gamma:=\ketbra{\Gamma}{\Gamma}$ denotes the unnormalized maximally entangled state. 
When the underlying systems are clear from context, we omit the subscripts to streamline notation.

A linear map $\mE$ is completely positive precisely when its Choi operator $J^{\mE}$ is positive semidefinite, i.e., $J^{\mE}\geqslant0$. 
In turn, the map is trace preserving when its Choi operator satisfies $\Tr_{B}[J^{\mE}]=\1_{A}$.
The positivity of $J^{\mE}$ ensures that it admits a decomposition into a sum of vectors multiplied by their duals. 
Using the standard correspondence (see Eq.~\eqref{TN:Vectorization}) between vectors and operators, this decomposition can, without loss of generality, be written as
\begin{align}\label{eq:Choi-decom}
    J^{\mE}=\sum_i\tvec{(K_i)}\tvec{(K_i)}^{\dag}.
\end{align}
Let $r$ denote the rank of the Choi operator $J^{\mE}$ associated with $\mE$. 
This rank sets the minimal number of vectors that can appear in the decomposition of Eq.~\eqref{eq:Choi-decom}. 
In tensor-network form, the matrix $J^{\mE}$ takes the structure shown below.
\begin{align}\label{TN:Choi}
    \raisebox{0ex}{\includegraphics[height=7.5em]{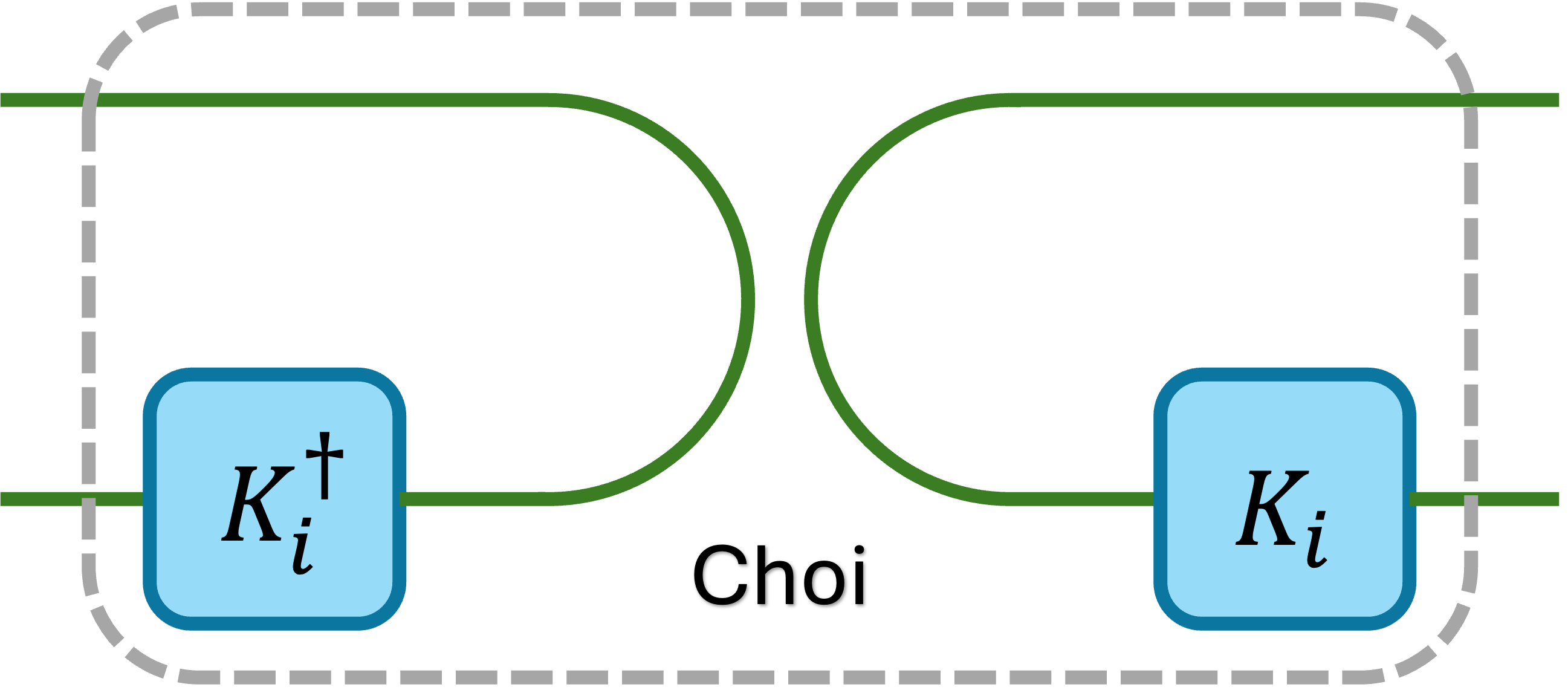}}
\end{align}
The symmetry between $\tvec{(K_i)}$ and its dual $\tvec{(K_i)}^{\dag}$ in Eq.~\eqref{TN:Choi} guarantees complete positivity of $\mE$, while trace preservation is given by
\begin{widetext}
\begin{align}\label{TN:TP}
    \raisebox{0ex}{\includegraphics[height=7.5em]{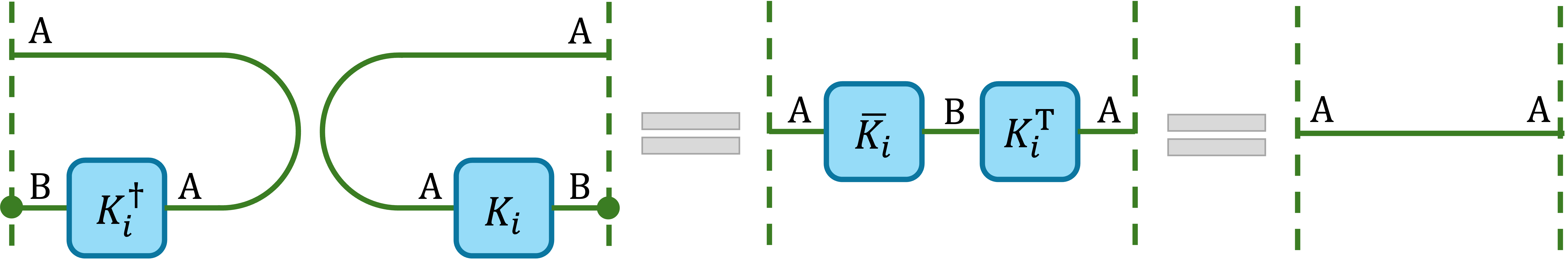}}
\end{align}
\end{widetext}
In our convention, the dashed vertical line marks the boundary, while the filled dots indicate subsystems that have been traced out. 
This dotted notation provides a compact way to represent Eq.~\eqref{TN:Trace}. 
In what follows, we freely switch between the dotted and the fully expanded diagrams without further remark, as Eq.~\eqref{TN:Trace} serves as a key ingredient in our realization theorem for superchannels (see Thm.~\ref{thm:RT}).
Here, the notation follows the same conventions adopted in our recent works~\cite{z2pr-zbwl,8g6j-w7ld}.

%%%%%%%%%%%%%%%%%%%%%%%%%%%%%%%%%%%%%%%%%%%%%%%%%%%%%%%%%%%%%%%%%%%%%%%
%%%%%%%%%%%%%%%%%%%%%%%%%%%%%%%%%%%%%%%%%%%%%%%%%%%%%%%%%%%%%%%%%%%%%%%

\subsection{Kraus Decomposition}\label{subsec:Kraus}

The second representation we will introduce is the Kraus decomposition.
Every completely positive map $\mE$ can be cast in an operator-sum form, in which its action on an input state takes the form
\begin{align}\label{eq:Kraus}
    \mE(\rho)=\sum_{i}K_i\rho K_i^{\dag}.
\end{align}
The associated tensor-network diagram reads
\begin{align}\label{TN:Kraus}
    \raisebox{0ex}{\includegraphics[height=6em]{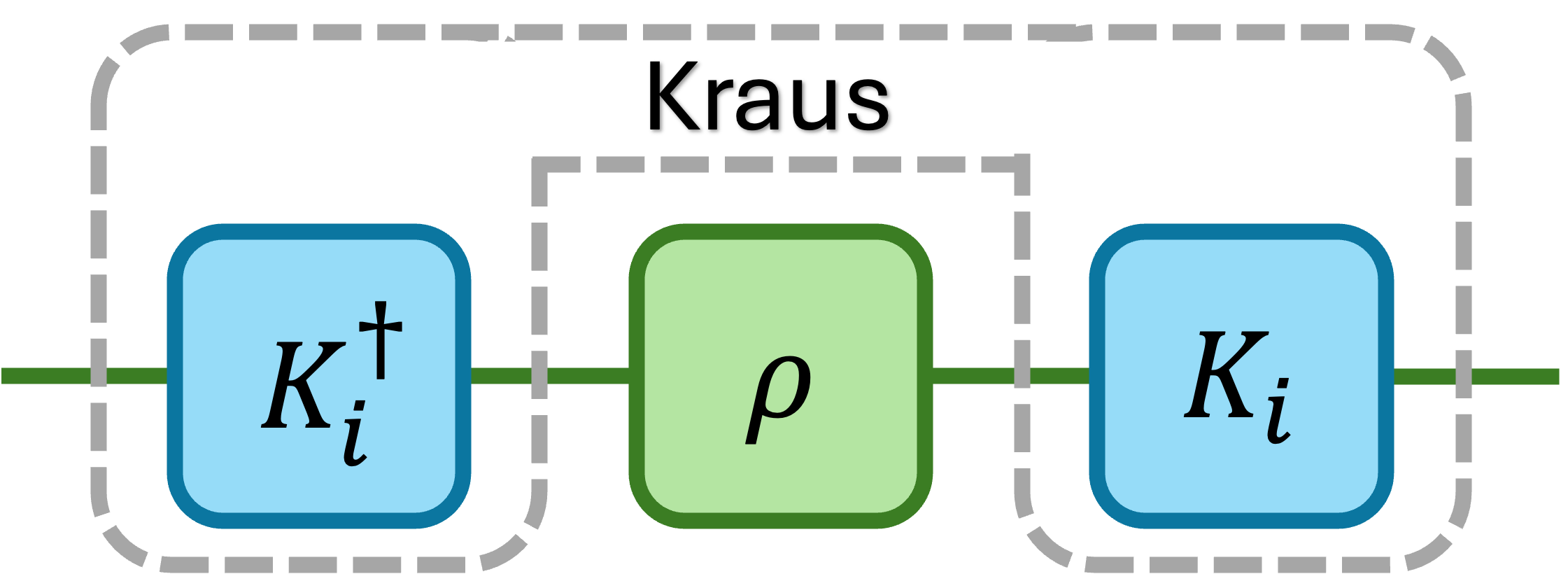}}
\end{align}
The requirement that the map $\mE$ be trace-preserving reduces to $\sum_iK_i^{\dag}K_i=\1_{A}$, together with its accompanying tensor-network depiction.
\begin{align}\label{TN:TP_Kraus}
    \raisebox{0ex}{\includegraphics[height=3.3em]{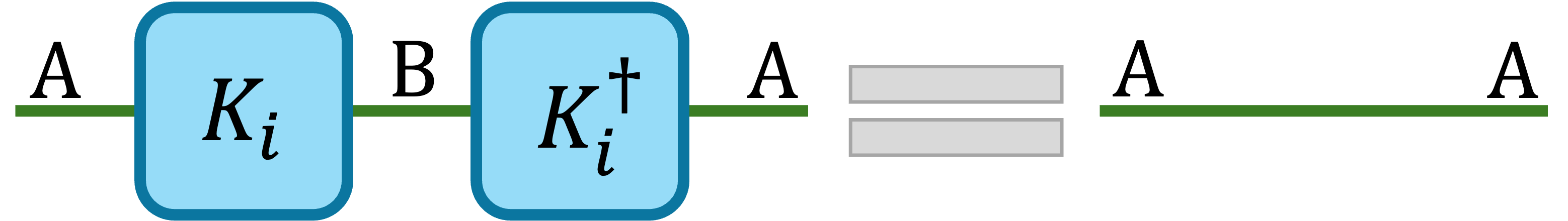}}
\end{align}
Note that, unlike the Choi operator $J^{\mE}$, which is uniquely determined by the channel $\mE$, the Kraus decomposition in Eq.~\eqref{eq:Kraus} is generally not unique. 
Different sets of Kraus operators $\{K_i\}$ are related by an isometry~\cite{Watrous_2018}.

%%%%%%%%%%%%%%%%%%%%%%%%%%%%%%%%%%%%%%%%%%%%%%%%%%%%%%%%%%%%%%%%%%%%%%%
%%%%%%%%%%%%%%%%%%%%%%%%%%%%%%%%%%%%%%%%%%%%%%%%%%%%%%%%%%%%%%%%%%%%%%%

\subsection{Stinespring Dilation}\label{subsec:Stinespring}

One may broaden the picture by adjoining an auxiliary environment $E$ to the systems $AB$. 
In this extended setting, any quantum channel admits a Stinespring dilation: it can be implemented by an isometry $V$ acting on the input $A$, producing a joint state on $BE$, followed by tracing out the environment. 
In mathematical terms, this means
\begin{align}\label{eq:Stinespring}
    \mE(\rho)=\Tr_{E}[V\rho V^{\dag}].
\end{align}
An explicit isometry realizing this dilation can be constructed as
\begin{align}\label{TN:Stinespring}
    \raisebox{0ex}{\includegraphics[height=9 em]{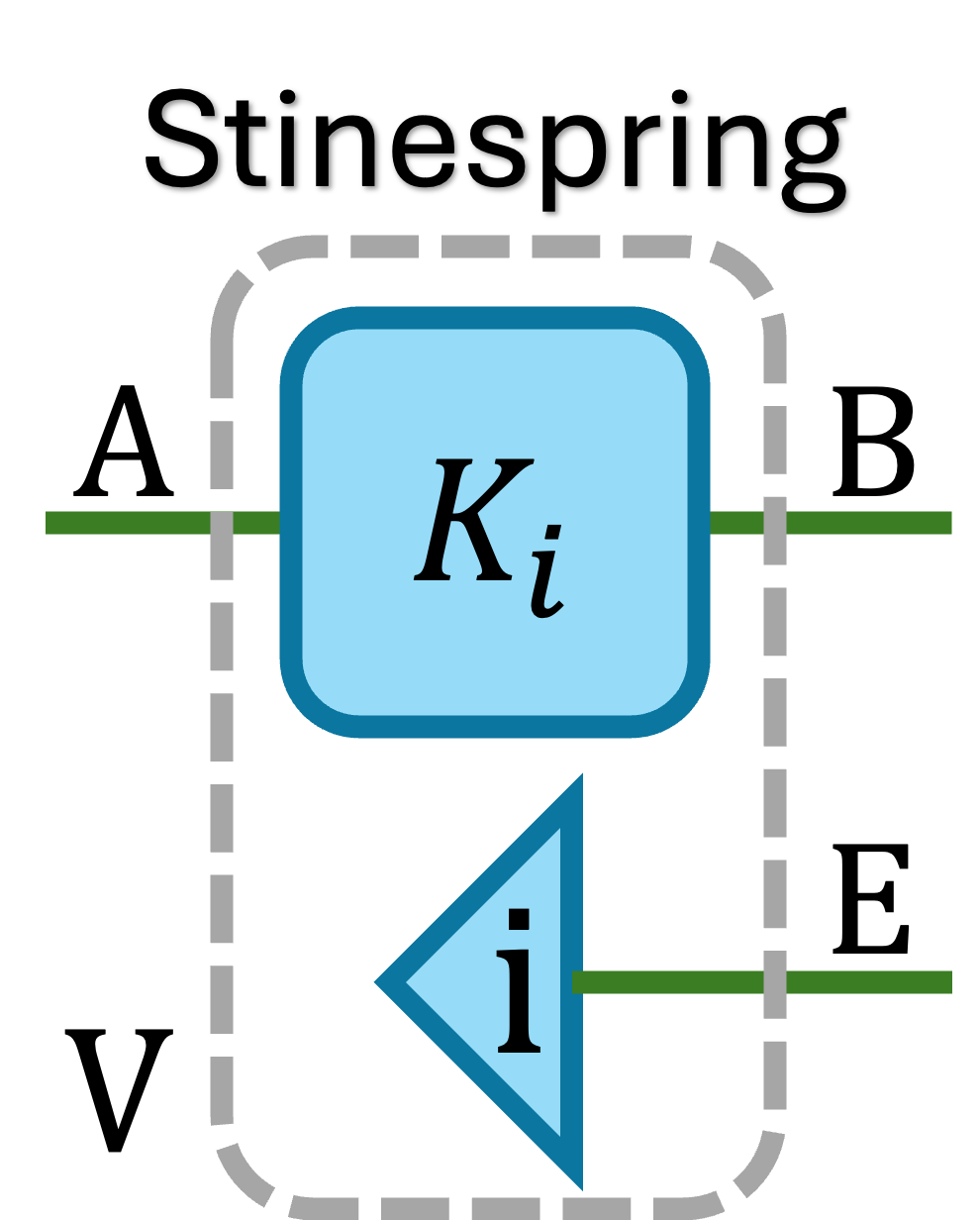}}
\end{align}
It is then straightforward to check that the matrix $V$ satisfies
\begin{align}\label{TN:Stinespring_Iso}
    \raisebox{0ex}{\includegraphics[height=7.5em]{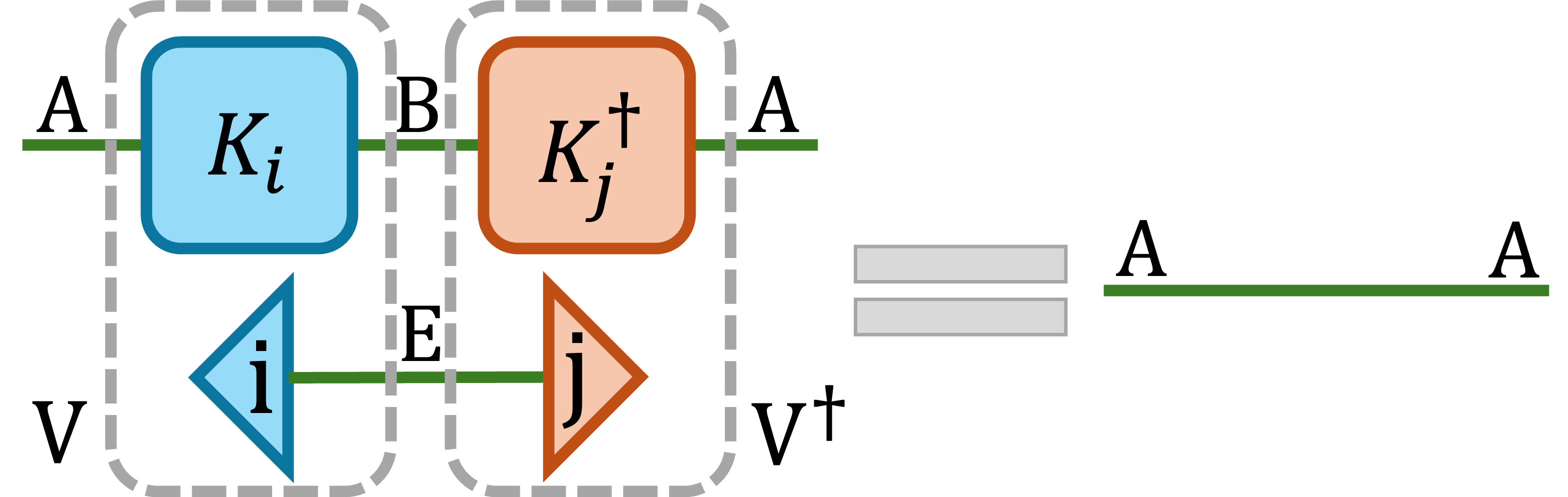}}
\end{align}
This perspective makes clear that any noisy quantum channel can be viewed as a coherent, noiseless isometry acting jointly on the system and an environment, followed by the loss of information encoded in an inaccessible environment.

%%%%%%%%%%%%%%%%%%%%%%%%%%%%%%%%%%%%%%%%%%%%%%%%%%%%%%%%%%%%%%%%%%%%%%%
%%%%%%%%%%%%%%%%%%%%%%%%%%%%%%%%%%%%%%%%%%%%%%%%%%%%%%%%%%%%%%%%%%%%%%%

\subsection{Liouville Superoperator}\label{subsec:Liouville}

The Liouville superoperator arises naturally from the vectorization (see Sec.~\ref{subsec:Vec}) of quantum states. 
Given a state $\rho$ on system $A$ and a quantum channel $\mE:A\to B$, both $\rho$ and its output $\mE(\rho)$ admit vectorized forms. 
The Liouville representation is precisely the linear map that relates these two vectors: it specifies how $\tvec{(\rho)}$ is transformed into $\tvec{(\mE(\rho))}$. 
This correspondence allows the action of a quantum channel to be visualized as a matrix acting on a vector, as illustrated below.
\begin{widetext}
\begin{align}\label{TN:Liouville}
    \raisebox{0ex}{\includegraphics[height=8em]{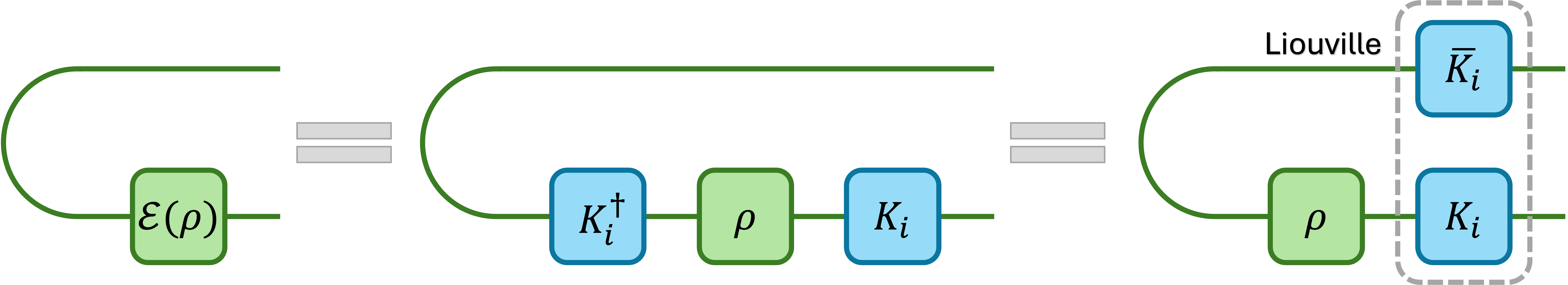}}
\end{align}
\end{widetext}
The four representations introduced above offer equivalent mathematical descriptions of quantum channels, as illustrated in Fig.~\ref{fig:Channel_Rep}. 
A detailed overview is provided in Chapter 2 of Ref.~\cite{Watrous_2018}, and the equivalence of these formulations is established in Proposition 2.20 of Ref.~\cite{Watrous_2018}.

\begin{figure}[t]
\centering   
\includegraphics[width=0.4\textwidth]{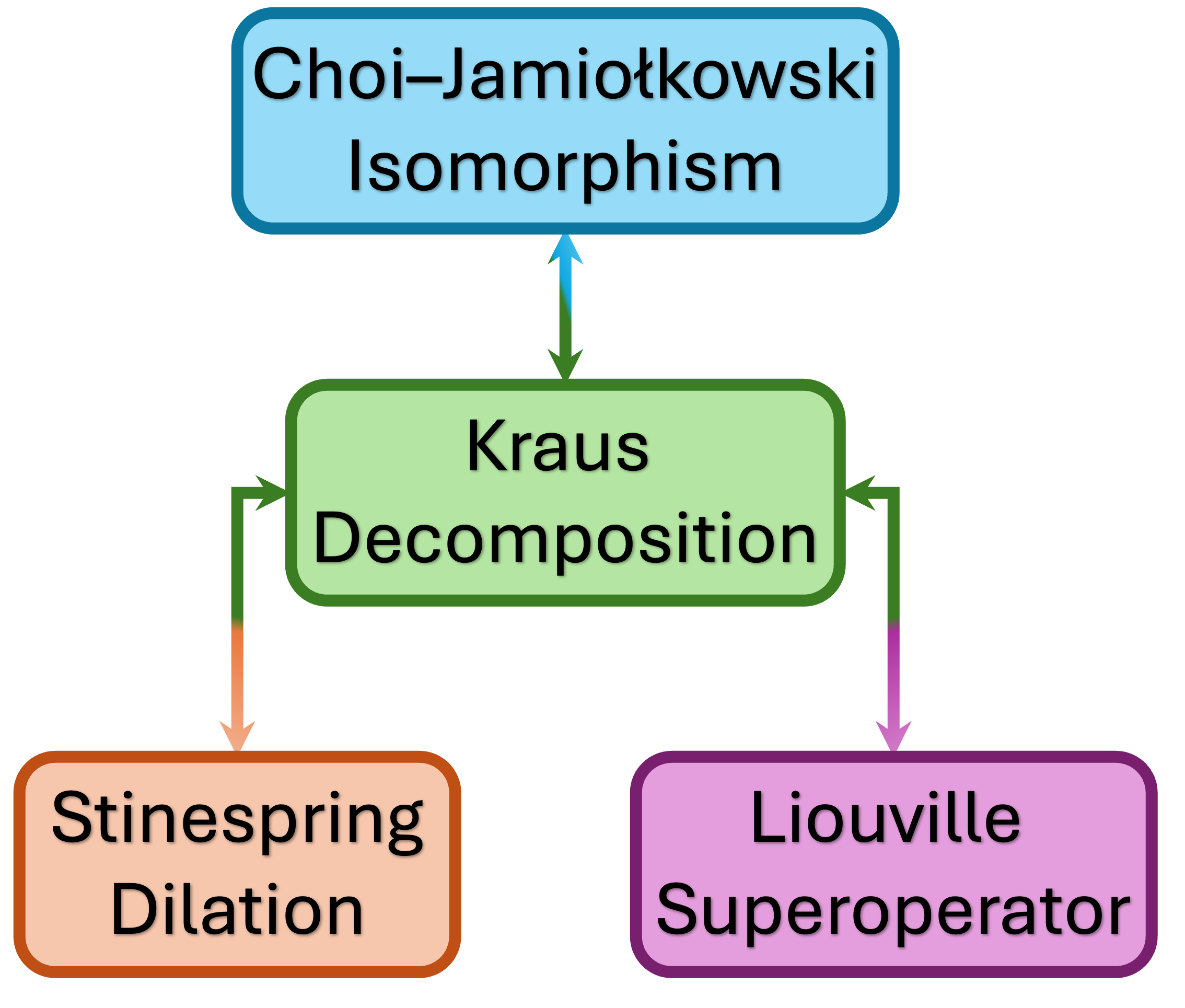}
\caption{(Color online) 
\textbf{Equivalent Representations of Quantum Channels}.  
Every quantum channel admits four representations, with each representation offering an equivalent view of the same underlying quantum dynamics.
}
\label{fig:Channel_Rep}
\end{figure}

%%%%%%%%%%%%%%%%%%%%%%%%%%%%%%%%%%%%%%%%%%%%%%%%%%%%%%%%%%%%%%%%%%%%%%%
%%%%%%%%%%%%%%%%%%%%%%%%%%%%%%%%%%%%%%%%%%%%%%%%%%%%%%%%%%%%%%%%%%%%%%%

\subsection{Link Product}\label{subsec:Link_Product}

In the closing part of our introduction to quantum channels, we turn to the link product, a tool that plays a pivotal role in capturing how quantum processes compose~\cite{959270}. 
Consider two matrices $M$ and $N$ that overlap on a common subsystem $C$. 
Their link product $\star$ is defined as
\begin{align}\label{eq:Link_Product}
    M\star N:=\Tr_{C}[M^{\T_{C}}\cdot N].
\end{align}
Here $^{\T_{C}}$ denotes the partial transpose on subsystem $C$. 
With the link product $\star$ in hand, the composition of two quantum channels 
$\mE_1:A\to B$ and $\mE_2: B\to C$, namely $\mE:=\mE_2\star\mE_1$, admits a compact expression: its Choi operator is given by
\begin{align}\label{eq:Channels_Composition}
    J^{\mE}=J^{\mE_2}\star J^{\mE_1}.
\end{align}

Although the term link product is relatively recent, the underlying idea has been in use since the earliest days of quantum theory. 
A simple example already appears in the basic interplay between state preparation and quantum measurement: 
the Choi operator of preparing a state is simply its density matrix $\rho$, while implementing a measurement $M_i$  corresponds to the Choi operator $M_i^{\T}$. 
Composing these two processes yields a probability computed exactly via their link product, namely
\begin{align}\label{eq:Born_Rule}
    p_i=\rho\star M_i^{\T}=\Tr[\rho M_i],
\end{align}
which is nothing but the Born rule~\cite{Born1926} introduced at the dawn of quantum mechanics.
Another useful example is the Choi operator associated with tracing out a system: 
it is simply the identity on that system, e.g., $J^{\Tr_A}=\1_A$.

Before concluding this section, we outline the labeling conventions adopted in this work. 
Peres once reminded us that 
``{\it Quantum phenomena do not occur in a Hilbert space. 
They occur in a laboratory.}''~\cite{peres2002quantum}.
A quantum system carries no inherent label -- these identifiers are assigned by us to keep track of its operational role. 
With this viewpoint in mind, we impose a consistent ordering of subsystems whenever quantum processes are composed.
Choi operators will always be written with input systems placed first and output systems last. 
When several subsystems are present, we follow their temporal order as they appear in the laboratory, arranging earlier systems to the left and later ones to the right.
As an illustration, consider a bipartite channel $\mE$ with inputs $A_1$, $A_2$ and outputs $B_1$, $B_2$, and suppose their laboratory times satisfy
\begin{align}
    t_{A_1}<t_{B_1}<t_{A_2}<t_{B_2}.
\end{align}
Following our convention, its Choi operator is written as $J^{\mE}_{A_1A_2B_1B_2}$.
With this ordering in place, the link product becomes insensitive to the specific permutation of subsystems; 
equivalently, it is commutative up to a relabeling of the quantum systems. 
In particular, for the examples considered in Eq.~\eqref{eq:Channels_Composition}, we have $J^{\mE_2}\star J^{\mE_1}=J^{\mE_1}\star J^{\mE_2}$.

For a broader and more detailed treatment of quantum channels, Refs.~\cite{wolf2012quantum,RevModPhys.86.1203,Wilde_2017,khatri2024principlesquantumcommunicationtheory} offer outstanding expositions that complement and extend the synthesis presented in this section.

%%%%%%%%%%%%%%%%%%%%%%%%%%%%%%%%%%%%%%%%%%%%%%%%%%%%%%%%%%%%%%%%%%%%%%%
%%%%%%%%%%%%%%%%%%%%%%%%%%%%%%%%%%%%%%%%%%%%%%%%%%%%%%%%%%%%%%%%%%%%%%%

\section{Quantum Superchannels}\label{sec:QSuperchannels}

Quantum channels form the basic operational units of quantum theory -- the analogue of individual ``Lego bricks''. Superchannels provide the most general framework for manipulating these building blocks, and they surface across nearly every corner of quantum information science.
For instance, when multiple uses of a noisy channel are available, any coding scheme that transmits quantum information, its encoding and decoding included, constitutes a memoryless superchannel~\cite{PhysRevA.51.2738}. 
In superdense coding, once the Pauli gates are fixed, the process that extracts classical information from the Pauli actions on the sender's side is itself a superchannel~\cite{PhysRevLett.69.2881}. 
In teleportation~\cite{PhysRevLett.70.1895}, given a fixed entangling operation, the protocol that processes and consumes the resource similarly takes the form of a superchannel~\cite{xing2023fundamentallimitationscommunicationquantum,glc7-xy8t}. 
Even programming a quantum processor~\cite{PhysRevLett.79.321,PhysRevLett.122.080505,PhysRevLett.125.210501}, and more broadly controlling its operation, can be captured within the superchannel formalism.

Understanding the structure and limitations of superchannels is therefore essential for identifying what is fundamentally achievable, and what is fundamentally prohibited, in quantum theory and emerging quantum technologies. 
However, in stark contrast to the mature development of quantum channel theory, the systematic study of superchannels remains comparatively underexplored.
Even at the level of basic formalism, two main issues have impeded a coherent development.

The first is a matter of inconsistency. 
For quantum channels, the Choi operator is uniquely defined. 
Superchannels, however, admit two competing Choi-like representations that have both been used extensively across the literature. 
One is obtained by viewing a superchannel as a bipartite channel and feeding in unnormalized maximally entangled states~\cite{PhysRevLett.101.060401}; historically the earliest and still the most common approach. 
The other arises by treating a superchannel as a higher-order transformation acting on quantum channels, yielding a second representation constructed from its action on a linear basis of inputs~\cite{8678741}. 
Whether these two constructions coincide, differ by a fixed transformation, or encode genuinely different operational content has never been clarified. 
As a result, even the fundamental notion of a ``Choi operator'' for superchannels has remained conceptually unsettled.
If the two frameworks diverge, results established in one may not translate faithfully to the other; 
if they are equivalent up to a systematic correspondence, it becomes essential to identify which representation provides the more natural lens for analyzing quantum dynamics across multiple time points.

The second issue is one of completeness.
Quantum channels enjoy a mature structural theory built upon four fully equivalent representations, Kraus operators, Stinespring dilation, Choi operators, and Liouville superoperator, each revealing a different aspect of their behavior (see Sec.~\ref{sec:QChannels}).
Superchannels, by contrast, have been understood almost entirely through their Choi-like descriptions.
Although these forms are convenient for convex analysis and optimization~\cite{Chiribella_2016}, they provide neither a structural theory nor a coherent unifying viewpoint.
Superchannel analogues of the Kraus representation, the Stinespring dilation, and the Liouville formalism are still missing; just as absent is a framework that relates these perspectives or reconciles the two competing Choi constructions.
A comprehensive theoretical foundation that ties these representations together has yet to be developed.

This gap is more than a technical inconvenience.
A conceptually grounded framework would not only resolve these inconsistencies but would also catalyze progress across quantum information science: 
from the design of multi-time quantum protocols~\cite{PhysRevLett.54.857,Emary_2014,Zych2019,liu2025spatialincompatibilitywitnessesquantum} and programmable quantum processors~\cite{PhysRevLett.79.321,PhysRevLett.122.080505,PhysRevLett.125.210501} to the characterization of memory effects~\cite{PhysRevA.72.062323,Taranto2024characterising}, adaptive strategies~\cite{xing2023fundamentallimitationscommunicationquantum,glc7-xy8t}, and higher-order quantum resources~\cite{taranto2025higherorderquantumoperations}.
In this section, we address these challenges and develop a consistent and complete formalism that places superchannels on the same structural footing as quantum channels.

%%%%%%%%%%%%%%%%%%%%%%%%%%%%%%%%%%%%%%%%%%%%%%%%%%%%%%%%%%%%%%%%%%%%%%%
%%%%%%%%%%%%%%%%%%%%%%%%%%%%%%%%%%%%%%%%%%%%%%%%%%%%%%%%%%%%%%%%%%%%%%%

\subsection{Realization Theorem}\label{subsec:Realization}

The first step toward a consistent and complete framework is to establish the technical foundations of superchannels and clarify their physical origin. 
In both classical~\cite{6773024} and quantum Shannon theory~\cite{PhysRevA.51.2738}, a noisy channel can be transformed into a more reliable one by surrounding it with an encoding map and a decoding map. 
In the quantum setting, these become two quantum channels, pre-processing and post-processing, which may, in general, be linked through an internal quantum memory. 
Any such construction defines a higher-order dynamical object~\cite{Bisio2019}: 
a pair of quantum processes connected by a memory system that acts on an input channel and produces a new channel. This is precisely what we call a {\it superchannel}.

What is truly striking is that the converse also holds. 
Chiribella {\it et al.} showed that any admissible transformation that maps quantum channels to quantum channels, while respecting the basic physical constraints of quantum mechanics, must arise from such a pre-/post-processing structure with quantum memory~\cite{Chiribella_2008}. 
In other words, superchannels are not merely a convenient abstraction: they capture all physically allowed higher-order dynamics on quantum channels.
Gour later provided an alternative proof of this structural characterization in~\cite{8678741}.

To keep the present development self-contained, we offer yet another derivation of this realization.
Our approach relies on tensor-network methods, which serve as our first Occam's razor: 
the diagrammatic calculus strips away redundant structure and exposes the minimal architecture underlying any superchannel.
In this representation, the memory cost associated with implementing a given superchannel becomes immediate and transparent.
The tensor-network viewpoint thus renders the internal architecture of superchannels almost self-evident, clarifying both their operational meaning and the necessity of their sequential form.
We begin by recalling a key lemma that sets the foundation for the realization.

\begin{lem}
[\bf{Prepare–and–Trace}~\cite{Chiribella_2008}]
\label{lem:P&T}
    Any matrix $M$ for which $\Tr[M J^{\mE}]=1$ holds for all quantum channels $\mE$ must be of the form
    \begin{align}
        M=\rho\otimes\1,
    \end{align}
    with $\Tr[\rho]=1$. If, in addition, $M\geqslant0$, then $\rho$ is a valid quantum state.
\end{lem}

\begin{figure}[t]
\centering   
\includegraphics[width=0.48\textwidth]{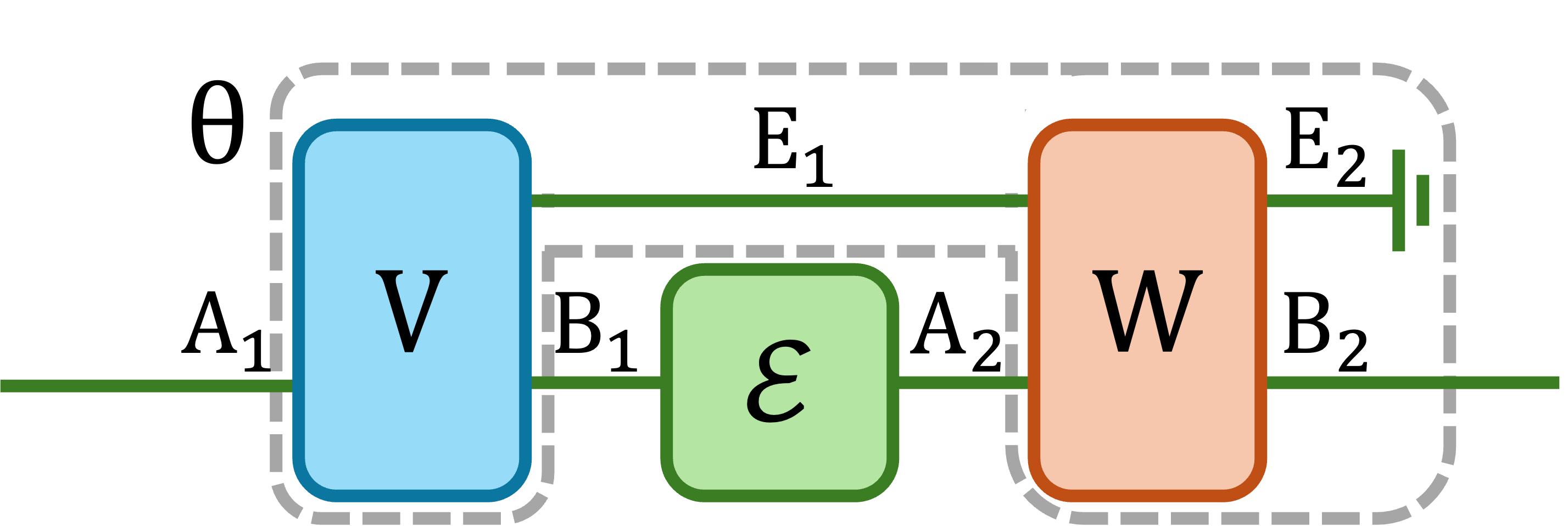}
\caption{(Color online) \textbf{Quantum Superchannel}.  
Given a quantum channel $\mE:B_1\to A_2$, any deterministic transformation $\theta$ acting on it can be physically realized by inserting two isometries in sequence: a pre-processing isometry $V:A_1\to E_1B_1$ applied before $\mE$, and a post-processing isometry $W:E_1A_2\to E_2B_2$ applied afterwards. 
Part of the output of $V$ is routed through a quantum memory $E_1$, which interfaces with the subsequent isometry $W$. 
After $W$ acts, the ancillary environment $E_2$ is traced out. 
The resulting overall dynamics $\theta=\Tr_{E_2}\circ W\circ V$ defines a superchannel.
}
\label{fig:Superchannel}
\end{figure}

The mathematical proof for Lem.~\ref{lem:P&T} is given in Ref.~\cite{Chiribella_2008}, but the accompanying physical picture is rarely spelled out. 
We provide that intuition here.
Consider feeding the channel $\mE$ with the state $\rho^{\T}$. 
Since transposition on all subsystems preserves positivity, $\rho^{\T}$ remains a legitimate quantum state. 
If we subsequently trace out the channel's output, the resulting prepare-and-trace procedure is represented by a Choi operator of the form $\rho^{\T}\otimes\1$.
Viewed this way, the whole process becomes 
\begin{align}
    (\rho^{\T}\otimes\1)\star J^{\mE}=\Tr[\rho\otimes\1 \cdot J^{\mE}]=1,
\end{align}
for any $\mE$.
Lem.~\ref{lem:P&T} shows that the converse also holds.
Any such operator $M$ necessarily arises as the Choi operator of a process that prepares a quantum state, feeds it through the channel, and then traces out the output.
In tensor-network terms, this means that
\begin{align}\label{TN:T&P_1}
    \raisebox{0ex}{\includegraphics[height=5em]{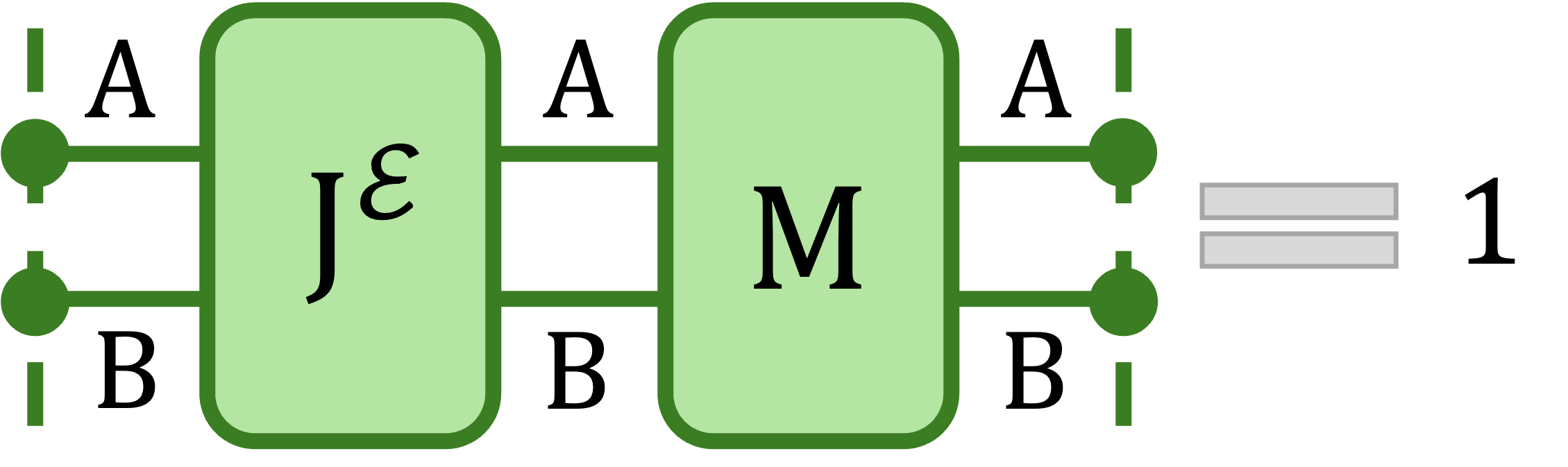}}
\end{align}
implies
\begin{align}\label{TN:T&P_2}
    \raisebox{0ex}{\includegraphics[height=5em]{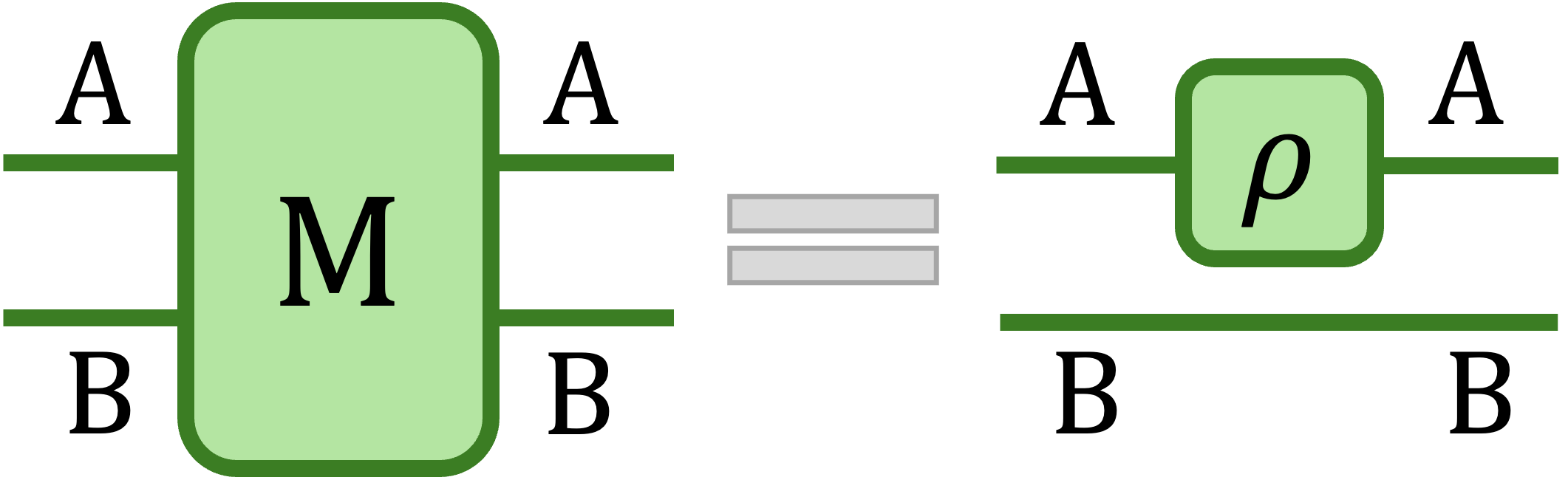}}
\end{align}

We now consider deterministic transformations of quantum channels.
The forward direction is immediate: applying an isometry before a channel $\mE$ and another isometry afterwards, as illustrated in Fig.~\ref{fig:Superchannel}, yields a new quantum channel.
What makes this framework powerful is that the converse is also true.
Every superchannel admits such a sequential realization, with all allowed higher-order dynamics captured by a pre-processing isometry, a quantum memory, and a post-processing isometry acting in that order.
Formally, the realization theorem for superchannels states that

\begin{figure}[t]
\centering   
\includegraphics[width=0.48\textwidth]{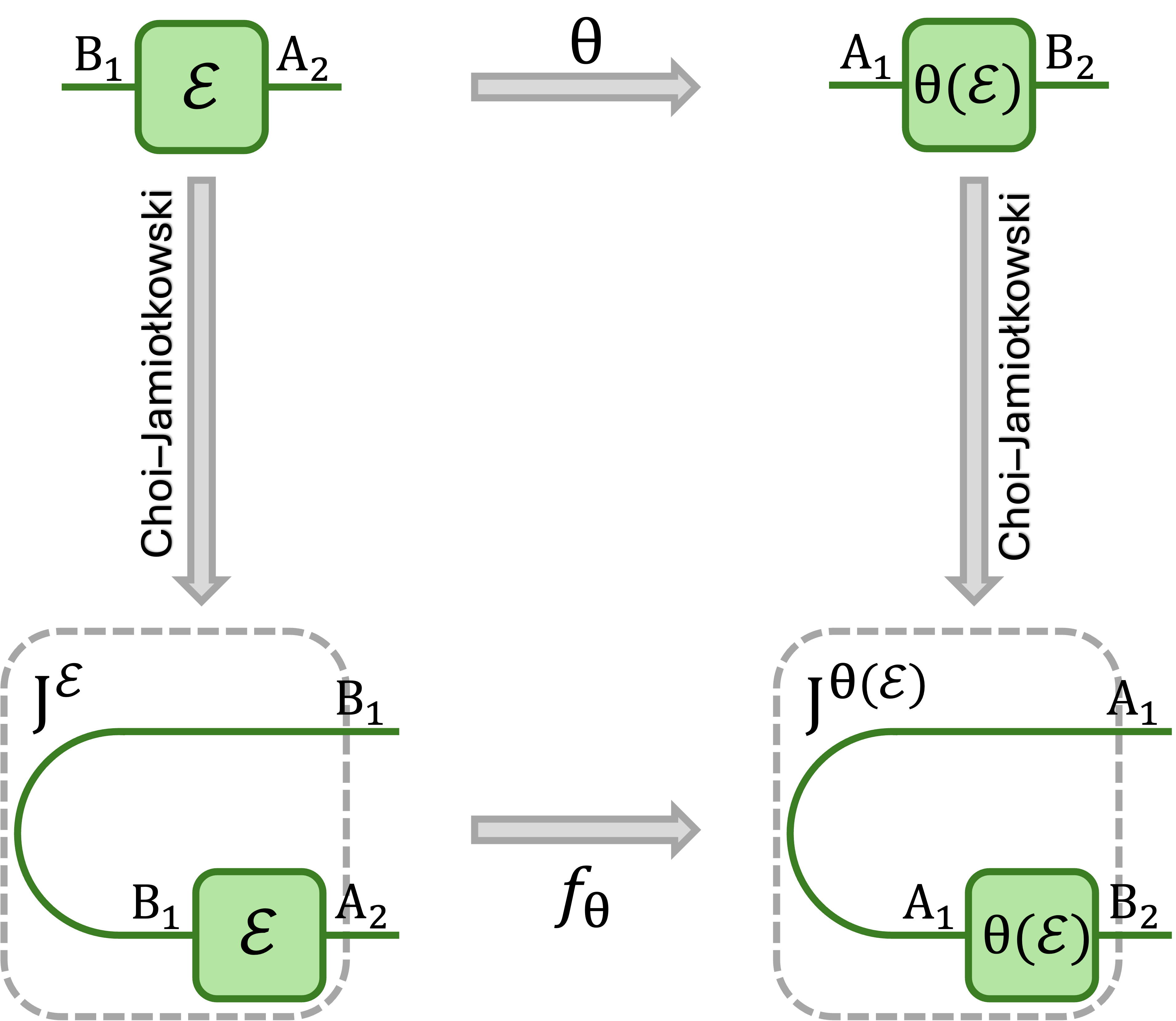}
\caption{(Color online) \textbf{Commutative Diagram for Superchannels}. 
Each superchannel $\theta$ admits a unique map $f_{\theta}$ that carries the Choi operator of a channel $\mE$ to the Choi operator of $\theta(\mE)$, i.e., $f_{\theta}(J^{\mE})=J^{\theta(\mE)}$. 
This correspondence is precisely what guarantees the commutativity of the diagram above.
}
\label{fig:Commutative_Diagram}
\end{figure}

\begin{thm}
[\bf{Realization Theorem}~\cite{Chiribella_2008}]
\label{thm:RT}
If a quantum process $\theta$ maps every channel from $B_1$ to $A_2$ into a channel from $A_1$ to $B_2$, then it admits a sequential physical realization.
Concretely, there exists an isometry $V:A_1\to E_1B_1$ followed by another isometry $W: E_1A_2\to B_2$ after which the environment $E_2$ is traced out.
The two isometries are connected via the memory system $E_1$, which mediates all allowed correlations between the pre- and post-processing stages, as illustrated in Fig.~\ref{fig:Superchannel}.
\end{thm}

Readers interested in the original proofs may consult Ref.~\cite{Chiribella_2008} and the alternative formulation discussed in Ref.~\cite{8678741}.
For completeness, and to illustrate how tensor-network calculus streamlines the entire construction, we present a concise and intuitive derivation below.

\begin{proof}

\begin{widetext}

Given a transformation $\theta$, we associate to it a map $f_{\theta}$ that acts directly on the Choi operator of a channel.
For a quantum channel $\mE(\cdot)=\sum_k M_{k}\cdot M_{k}^{\dag}$, this induced action is written as $f_{\theta}(J^\mE):=J^{\theta(\mE)}$ (see Fig.~\ref{fig:Commutative_Diagram}).
A useful observation is that $\theta$ is completely positive if and only if its induced map $f_{\theta}$ is completely positive.
This equivalence allows us to analyze $\theta$ within the framework of quantum channels developed in Sec.~\ref{sec:QChannels}.
In particular, we will make use of its Kraus decomposition $f_{\theta}(\cdot)=\sum_iK_i\cdot K_i^{\dag}$.
The proof begins by noting that for any channel $\mE$, its Choi operator satisfies $J^{\mE}\star(\rho^{\T}\otimes\1)=1$.
If $\theta$ is a superchannel, then the induced map $f_{\theta}$ also obeys $f_{\theta}(J^{\mE})\star(\rho^{\T}\otimes\1)=1$.
Evaluating the expression within the tensor-network, we arrive at

\begin{align}\label{TN:RT_1}
    \raisebox{0ex}{\includegraphics[height=7.6em]{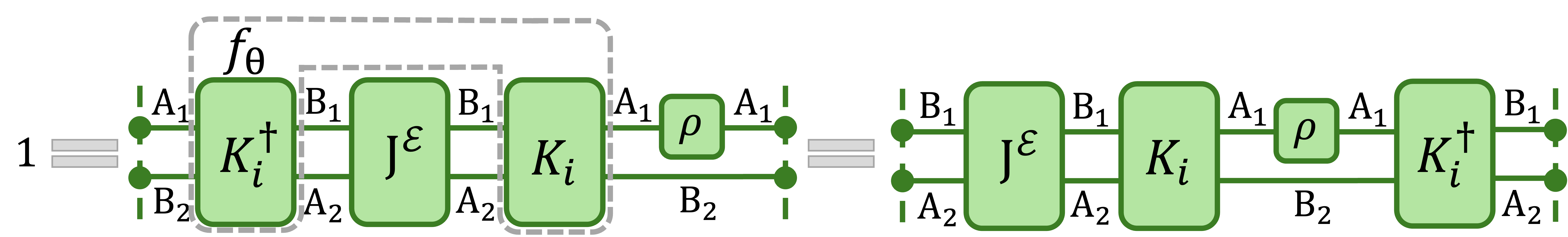}}
\end{align}

Thanks to Lem.~\ref{lem:P&T}, we obtain

\begin{align}\label{TN:RT_2}
    \raisebox{0ex}{\includegraphics[height=8em]{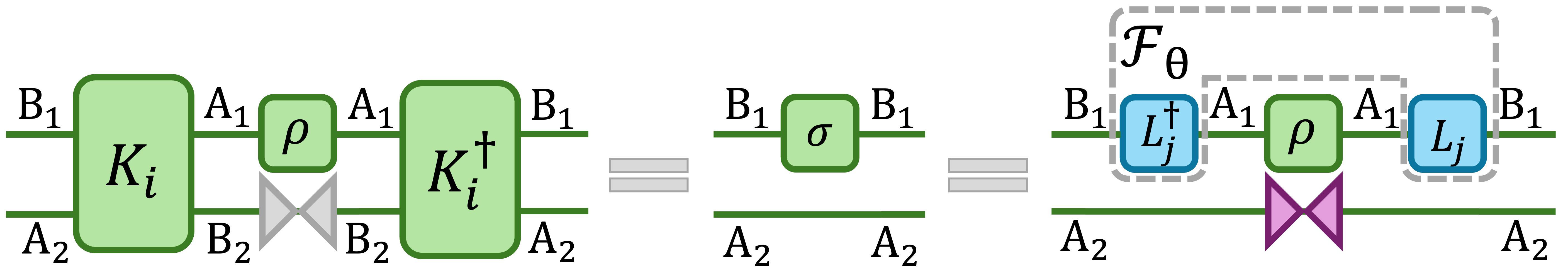}}
\end{align}
where the second equality follows from noting that the procedure generating $\sigma$ can be viewed as the action of a quantum channel $\mF_{\theta}$, namely $\sigma=\mF_{\theta}(\rho)=\sum_j L_j\rho L_j^{\dag}$. 
More precisely, 
\begin{align}\label{eq:mF_theta}
    \mF_{\theta}(\rho):= 
    \frac{1}{d_{A_2}}
    \Tr_{A_2}[f_{\theta}^{\dag}(\rho_{A_1}\otimes\1_{B_2})],
\end{align}
is completely positive and the procedure $\rho\to\sigma$ is trace-preserving, so $\mF_{\theta}$ is a CPTP map.
Here, $d_{A_2}$ denotes the dimension of system $A_2$.
Since the leftmost and rightmost Kraus representations in Eq.~\eqref{TN:RT_2} describe the same quantum channel, their Kraus operators must be related by an isometry $W$, that is
\begin{align}\label{TN:RT_3}
    \raisebox{0ex}{\includegraphics[height=9em]{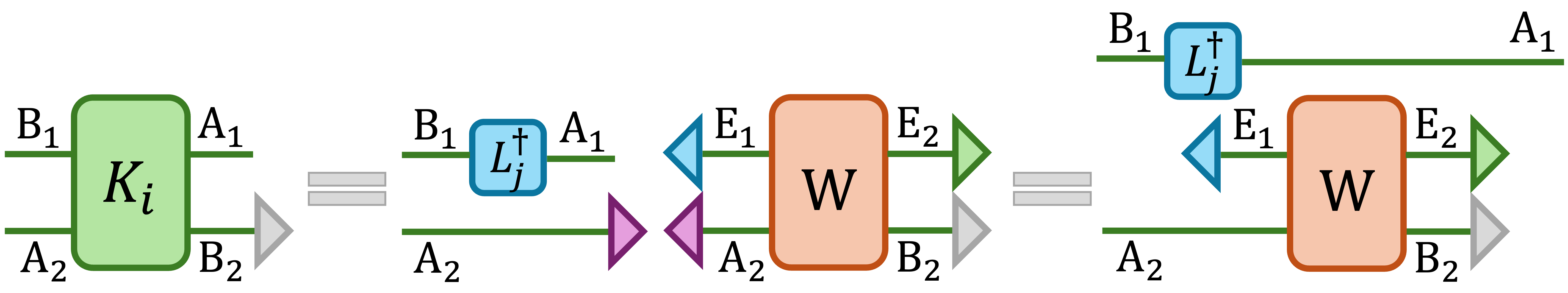}}
\end{align}
Now the state $\theta(\mE)(\rho)=\Tr_{A_1}[f_{\theta}(J^{\mE})\cdot\rho^{\T}]$ can be expressed as
\begin{align}\label{TN:RT_4}
    \raisebox{0ex}{\includegraphics[height=8.6em]{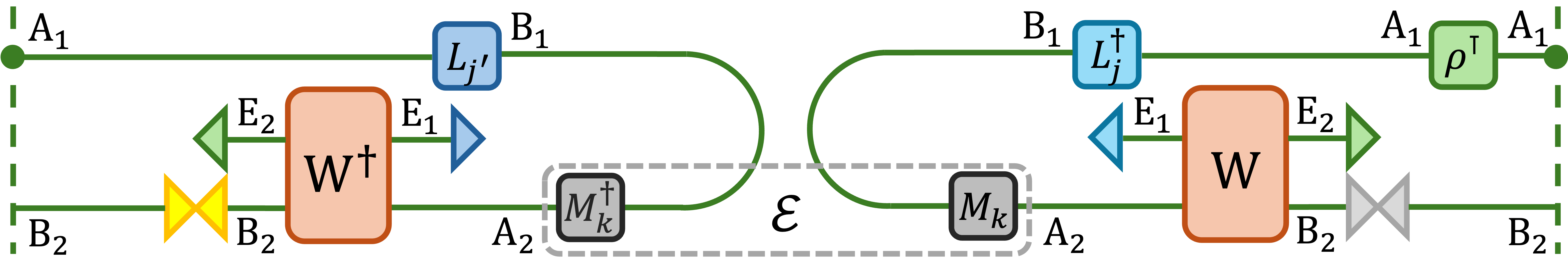}}
\end{align}
or, equivalently,
\begin{align}\label{TN:RT_5}
    \raisebox{0ex}{\includegraphics[height=5.8em]{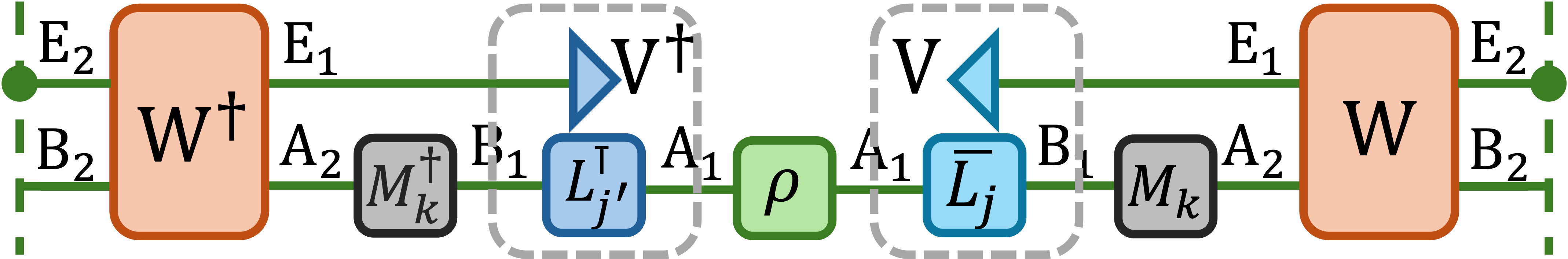}}
\end{align}
thereby completing the proof.

\end{widetext}

\end{proof}

Throughout the proof, we repeatedly invoke the standard trick of expanding the identity operator in an orthonormal basis (see Sec.~\ref{subsec:Id&Tr}). 
For example, in Eq.~\eqref{TN:RT_2} we write $\1_{B_2}=\sum_x\ketbra{x}{x}$ and $\1_{A_2}=\sum_y\ketbra{y}{y}$, where the corresponding basis indices $x$ and $y$ differ. 
To keep the tensor-network diagrams visually uncluttered, we suppress these indices and use colors to encode them instead.

Following the same convention, in Eq.~\eqref{TN:RT_3} the blue triangle denotes the ket carrying the index $j$ associated with the Kraus operator $L_{j}^{\dag}$; 
the summation implied by the blue box–triangle pair is omitted for brevity. 
Finally, in Eq.~\eqref{TN:RT_4}, although the grey and yellow triangles each represent a basis decomposition of the same Hilbert space $B_2$, their associated indices need not coincide. 
Distinct colors are therefore used to emphasize that the two decompositions are independent.

Thm.~\ref{thm:RT} establishes that every superchannel admits a physical realization as a sequential process comprising a pre-processing map, and a post-processing map. 
The converse is immediate: inserting any pre-processing channel before the input channel and any post-processing channel afterward always yields a valid quantum channel. 
Taken together, these yield

\begin{cor}
[\bf{Superchannel}]
\label{cor:Superchannel}
A linear map is a superchannel if and only if it admits a sequential realization by quantum channels, as depicted in Fig.~\ref{fig:Superchannel}.
\end{cor}

%%%%%%%%%%%%%%%%%%%%%%%%%%%%%%%%%%%%%%%%%%%%%%%%%%%%%%%%%%%%%%%%%%%%%%%
%%%%%%%%%%%%%%%%%%%%%%%%%%%%%%%%%%%%%%%%%%%%%%%%%%%%%%%%%%%%%%%%%%%%%%%

\subsection{Memory Cost}\label{subsec:Memory_Cost}

In contrast to the case of quantum state transformations, which never call for a quantum memory, the transformation of quantum channels typically does. 
The realization theorem makes this requirement explicit: any superchannel acting on channels must, in general, be implemented with an intermediate memory system $E_1$ (see Fig.~\ref{fig:Superchannel}). 
Our tensor-network derivation further identifies the minimal memory cost, showing that it is fixed by the Choi rank of the channel $\mF_{\theta}$.
This leads directly to the following corollary.

\begin{cor}
[\bf{Memory Cost}]
\label{cor:MC}
In the architecture of a superchannel $\theta$ (see Fig.~\ref{fig:Superchannel}), the minimal memory cost $d_{\theta}$, namely the smallest admissible dimension of the intermediate system $E_1$, is set by the rank of $\mF_{\theta}$ (see Eq.~\eqref{eq:mF_theta}).
Mathematically, we have 
\begin{align}\label{eq:memory_cost}
    d_{\theta}=\rank(\Tr_{A_2B_2}[J^{f_{\theta}^{\dag}}]).
\end{align}
\end{cor}

\noindent
Eq.~\eqref{eq:memory_cost} is obtained by noting that

\begin{widetext}
\begin{align}\label{TN:Memory_Cost}
    \raisebox{0ex}{\includegraphics[height=11em]{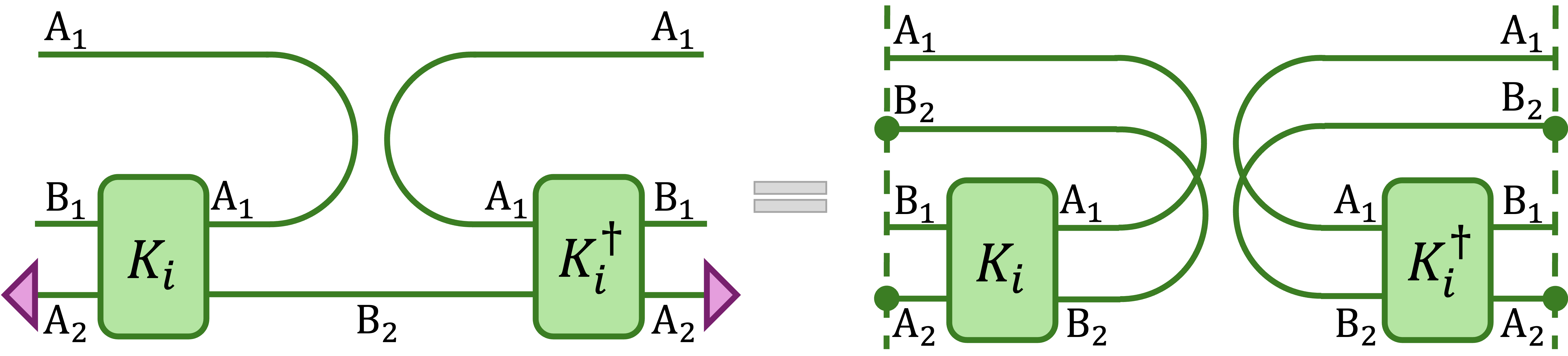}}
\end{align}
\end{widetext}

\noindent
Here, we omit the overall scalar factor $1/d_{A_2}$, since it has no effect on the rank of $\mF_{\theta}$.

%%%%%%%%%%%%%%%%%%%%%%%%%%%%%%%%%%%%%%%%%%%%%%%%%%%%%%%%%%%%%%%%%%%%%%%
%%%%%%%%%%%%%%%%%%%%%%%%%%%%%%%%%%%%%%%%%%%%%%%%%%%%%%%%%%%%%%%%%%%%%%%

\subsection{Generalized Occam's Razor}\label{subsec:Bi_Channels}

{\it Occam's Razor} is often invoked to favor the simpler of two competing descriptions~\cite{Sober_2015}. 
For quantum superchannels, however, this familiar idea offers little guidance. 
The two prevailing viewpoints, 
\begin{enumerate}[i]
  \item ({\bf Higher-Order Transformation}) One treating a superchannel as a higher order transformation that maps a quantum channel to another channel;
  \item ({\bf Bipartite Channel}) The other interpreting it as a bipartite channel involving multiple time steps and constrained by no signaling from post-processing to pre-processing,
\end{enumerate}
share essentially the same structural and complexity.
Neither viewpoint is appreciably simpler, and the usual spirit of Occam's Razor cannot select between them.

\begin{figure}[t]
\centering   
\includegraphics[width=0.48\textwidth]{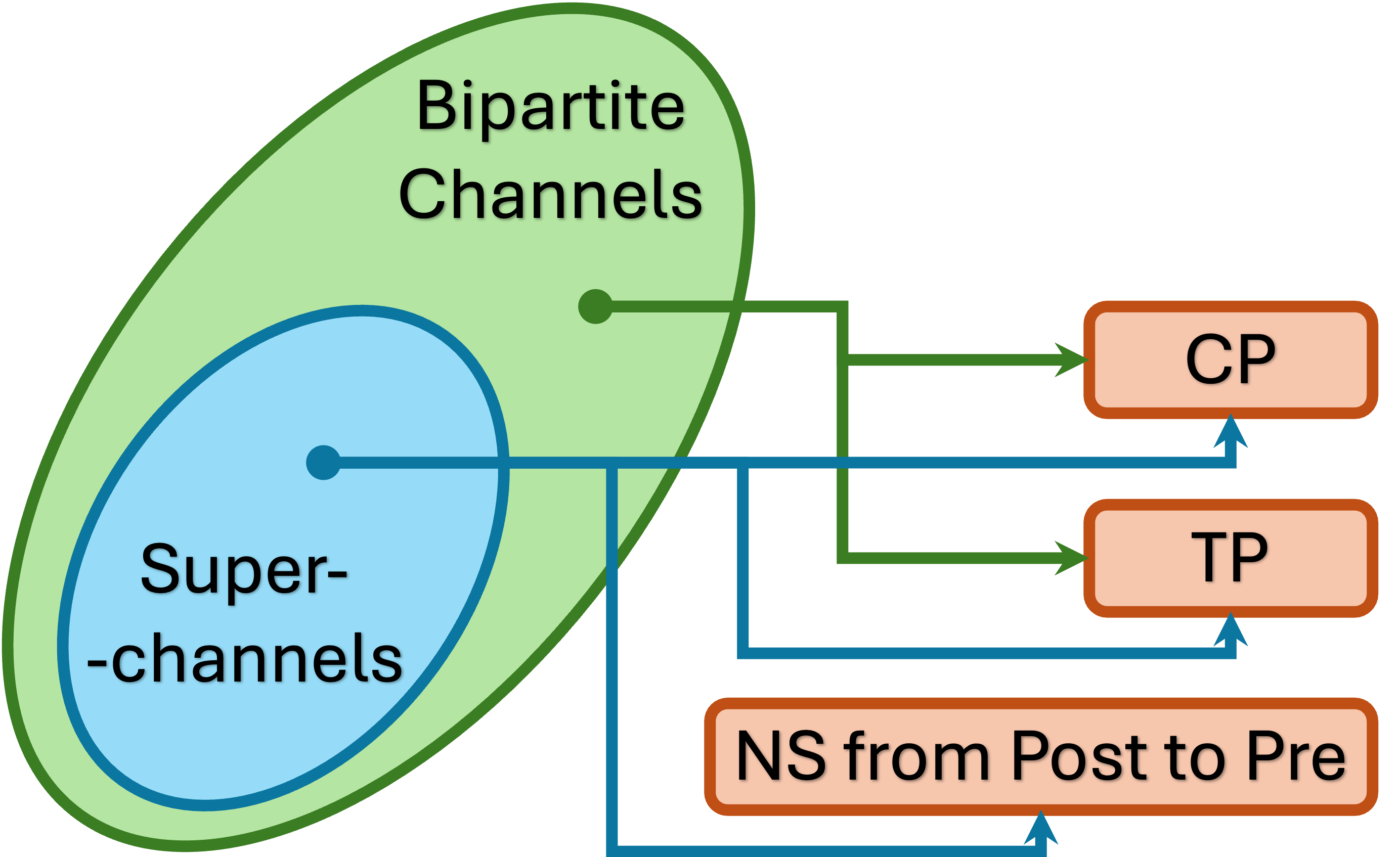}
\caption{(Color online) \textbf{Superchannels as a Subset of Bipartite Quantum Channels}.  
Any superchannel satisfies CP, TP, and the no-signaling (NS) condition from post-processing to pre-processing, and can therefore be regarded as a bipartite channel. 
For instance, the superchannel $\theta$ illustrated in Fig.~\ref{fig:Superchannel} can be viewed as a bipartite channel from $A_1A_2$ to $B_1B_2$.
}
\label{fig:Bipartite_Channel}
\end{figure}

This motivates a generalized version of the principle. 
When two formulations exhibit comparable complexity, the more compelling choice is the one that aligns more naturally with established theoretical frameworks, since no theory stands entirely on its own. 
The theory of quantum channels is already rich and well developed, and regarding superchannels as a special subclass of bipartite channels embeds higher-order transformations within this familiar setting (see Fig.~\ref{fig:Bipartite_Channel}). 
This perspective maintains continuity with existing results and allows us to draw directly on the full set of tools from channel theory~\cite{PhysRevA.80.022339}.

Adopting this viewpoint turns out to be highly valuable. 
It resolves inconsistencies that persist in current formulations, provides a unified account of the four standard representations of superchannels, and leads naturally to refined notions such as entanglement breaking superchannels, which in turn offer new structural and operational insights.
These developments will be detailed in the sections that follow.
In this sense, viewing a superchannel as a bipartite channel embodies the generalized Occam's razor -- the second razor guiding this work.

%%%%%%%%%%%%%%%%%%%%%%%%%%%%%%%%%%%%%%%%%%%%%%%%%%%%%%%%%%%%%%%%%%%%%%%
%%%%%%%%%%%%%%%%%%%%%%%%%%%%%%%%%%%%%%%%%%%%%%%%%%%%%%%%%%%%%%%%%%%%%%%

\subsection{Choi Operators of Superchannels}\label{subsec:SC_Choi}

Two constructions of the Choi operator for a superchannel are commonly used, each reflecting one of the viewpoints introduced in Sec.~\ref{subsec:Bi_Channels}. 
Guided by our Generalized Occam's Razor, we adopt the bipartite channel perspective and take its corresponding Choi operator as the canonical choice. 
This representation provides a clean point of departure for our analysis and helps illuminate how it relates to the alternative formulation.
In this way, our framework reconciles the two formulations and resolves the inconsistency in how Choi operators for superchannels are defined.

Formally, one may simply redraw the superchannel as the bipartite diagram shown in Fig.~\ref{fig:Superchannel_Choi_1}, treating $\theta$ as a channel from $A_1A_2$ to $B_1B_2$.
Its Choi operator is then defined in the usual way via full vectorization of the input systems (see Eq.~\eqref{TN:Full_Vectorization}).
Explicitly, it is given by
\begin{align}\label{eq:Superchannel_Choi_1}
    J^{\theta}_{A_1A_2B_1B_2}
    :=\id\otimes\theta_{A_1A_2\to B_1B_2}(\Gamma_{A_1A_1}\otimes\Gamma_{A_2A_2}).
\end{align}
This is precisely the construction used by Chiribella {\it et al.} in Ref.~\cite{PhysRevLett.101.060401} to define the Choi operator of a quantum comb from its causal network.

\begin{figure}[t]
\centering   
\includegraphics[width=0.42\textwidth]{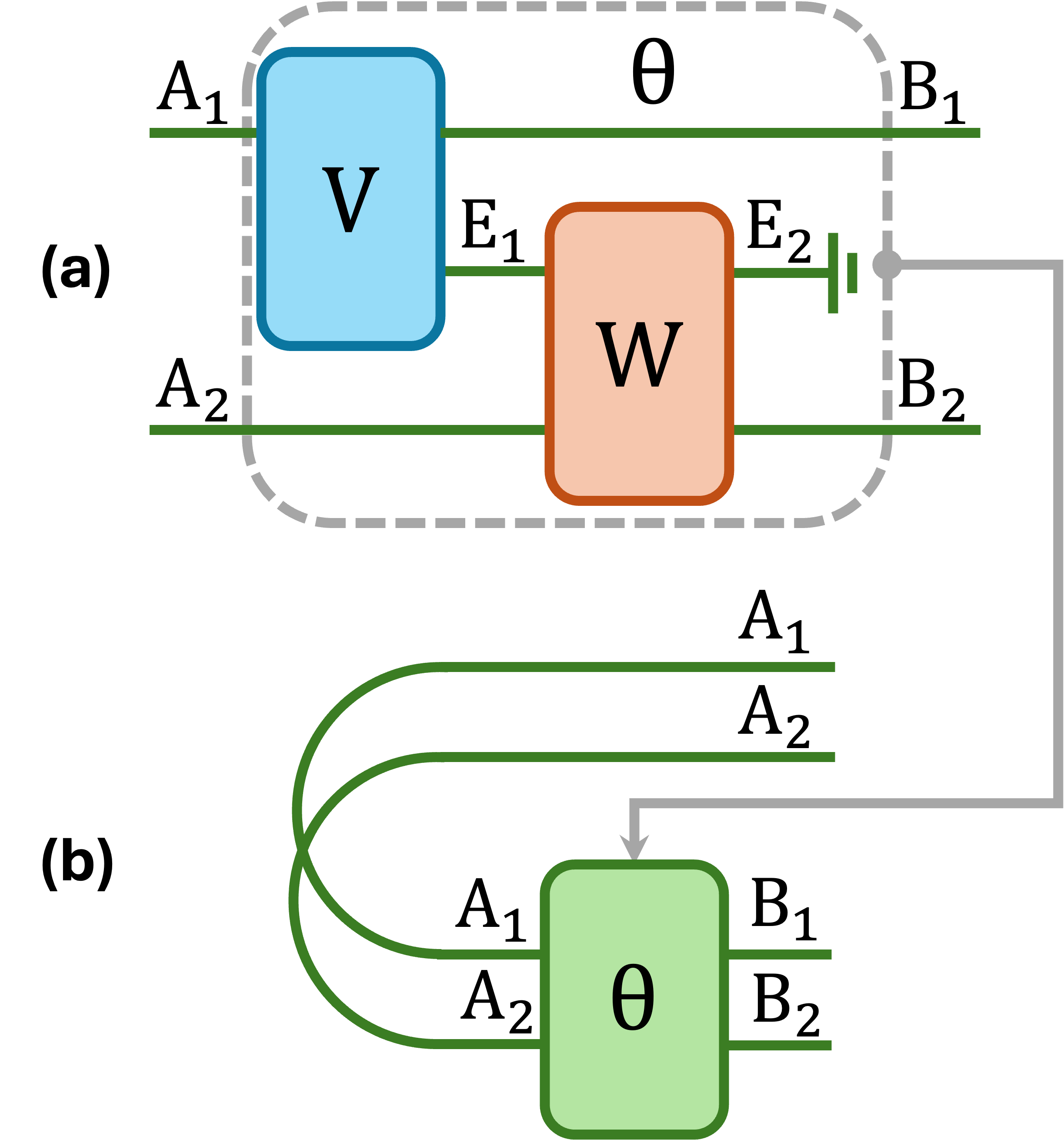}
\caption{(Color online) \textbf{Superchannel Choi Operator}.  
(a) Bipartite representation of the superchannel $\theta$ (see Fig.~\ref{fig:Superchannel}).
(b) The corresponding Choi operator.
Exploiting the symmetry of the tensor network, we depict only half of the conventional representation; the Kraus operators are omitted and replaced by the symbol $\theta$, giving a compact circuit-style diagram.
}
\label{fig:Superchannel_Choi_1}
\end{figure}

In Ref.~\cite{8678741}, Gour introduced a Choi-type operator for a superchannel from the higher-order transformation perspective.
To distinguish this construction from the conventional Choi operator defined in Eq.~\eqref{eq:Superchannel_Choi_1}, we refer to it as the {\it Gour operator}.
Once both representations are on the table, a natural question emerges: how are they related?
Surprisingly, this connection has received little attention in the existing literature on superchannels and has remained unresolved.
%Before establishing the relation, we set up the necessary notation.

Before relating the two constructions, we briefly recall how the Gour operator is defined.
It is convenient to return first to the conventional Choi operator (see Eq.~\eqref{eq:Choi}) for channels, whose underlying object is a quantum state. 
For an ordinary channel $\mE:A\to B$, one identifies the image of each matrix $\ketbra{i}{j}_A$ spanning the input space, records $\mE(\ketbra{i}{j})_B$, and sums these contributions to obtain the channel's Choi operator 
\begin{align}\label{eq:Channel_Choi}
    J^{\mE}=\sum_{ij}\ketbra{i}{j}\otimes\mE(\ketbra{i}{j}).
\end{align}
The same logic extends naturally to superchannels $\theta:A_1A_2\to B_1B_2$ (see Fig.~\ref{fig:Superchannel}). 
Here, the ``input space'' is the vector space of all linear maps from $B_1$ to $A_2$. Its canonical basis is given by the maps $e_{ijkl}:B_1\to A_2$
\begin{align}\label{eq:e_ijkl}
    e_{ijkl}(\cdot):=
    \langle \ketbra{i}{j}_{B_1}, \cdot \rangle 
    \ketbra{k}{l}_{A_2},
\end{align}
whose Choi operator takes the form
\begin{align}
    J^{e_{ijkl}}_{B_1A_2}
    =
    \ketbra{i}{j}_{B_1}\otimes\ketbra{k}{l}_{A_2}.
\end{align}
Evaluating the superchannel on each $J^{e_{ijkl}}$ and summing the corresponding Choi operators yields the Gour operator $G^{\theta}$, in direct analogy with the construction of the conventional Choi operator for quantum channels.
\begin{align}\label{eq:Superchannel_Gour}
    G^{\theta}_{B_1A_2A_1B_2}:=\sum_{ijkl}
    J^{e_{ijkl}}_{B_1A_2}\otimes J^{\theta(e_{ijkl})}_{A_1B_2}.
\end{align}

Strictly speaking, the Gour operator $G^{\theta}$ (see Eq.~\eqref{eq:Superchannel_Gour}) and the Choi operator $J^{\theta}$ (see Eq.~\eqref{eq:Superchannel_Choi_1}) of a quantum superchannel $\theta$ are not the same mathematical object. 
Even at first glance, the two differ in the ordering of their underlying systems. 
Beyond this apparent mismatch, one may ask whether they differ in any deeper sense. 
To explore this, we examine their relationship through the following reformulation,

\begin{align}
    G^{\theta}
    =
    &\sum_{ijkl}
    J^{e_{ijkl}}_{B_1A_2}\otimes 
    J^{\theta}_{A_1\tilde{A}_2\tilde{B}_1B_2}\star 
    J^{e_{ijkl}}_{\tilde{B}_1\tilde{A}_2}\notag\\
    =
    &J^{\theta}\star
    \left(\sum_{ijkl}
    \ketbra{i}{j}_{B_1}\otimes\ketbra{k}{l}_{A_2}
    \otimes
    \ketbra{i}{j}_{\tilde{B}_1}\otimes\ketbra{k}{l}_{\tilde{A}_2}
    \right)
    \notag\\
    =
    &J^{\theta}\star
    \left(\Gamma_{B_1\tilde{B}_1}\otimes\Gamma_{A_2\tilde{A}_2}\right)\notag\\
    =
    &\Tr_{\tilde{B}_1\tilde{A}_2}[J^{\theta}_{A_1\tilde{A}_2\tilde{B}_1B_2}\cdot \swap_{B_1\tilde{B}_1}\otimes\swap_{A_2\tilde{A}_2}].
    \label{eq:Gour_Choi}
\end{align}
Here the tilde notation is used to distinguish the reference system $B_1A_2$ from the actual system $\tilde{B}_1\tilde{A}_2$ on which $\theta$ acts, so as to avoid any potential confusion. Using tensor-network, Eq.~\eqref{eq:Gour_Choi} can be written as

\begin{widetext}
\begin{align}\label{TN:Gour_Choi}
    \raisebox{0ex}{\includegraphics[height=11.2em]{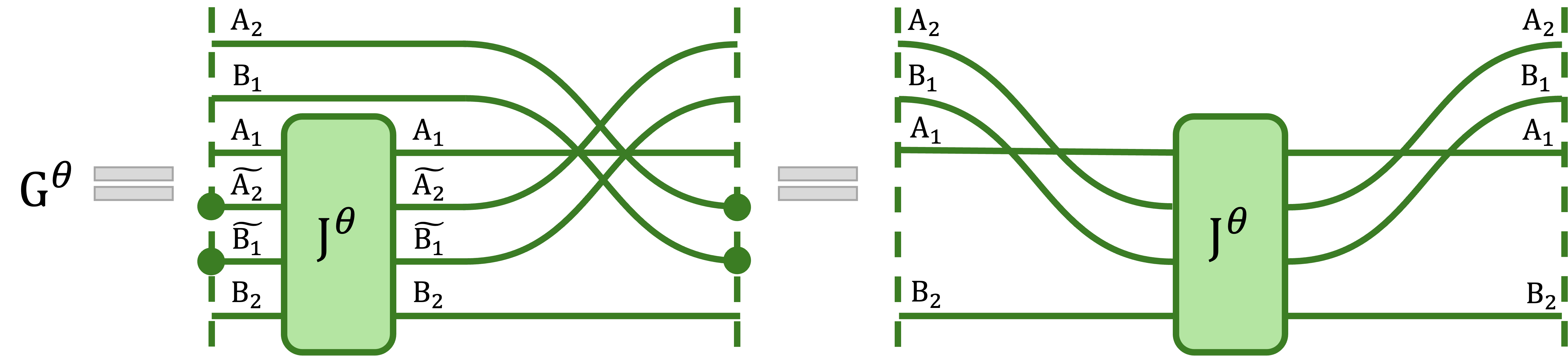}}
\end{align}
\end{widetext}

\noindent
It is now evident that the conventional Choi operator $J^{\theta}$ and the Gour operator $G^{\theta}$ of a superchannel $\theta$ differ only by system permutations. 
Since, in this work, matrices are treated as equivalent modulo permutations, we conclude that

\begin{thm}
[\bf{Equivalence Theorem}]
\label{thm:ET}
Given a superchannel $\theta$, its Choi operator (see Eq.~\eqref{eq:Superchannel_Choi_1}) and its Gour operator (see Eq.~\eqref{eq:Superchannel_Gour}) are equivalent up to a permutation of subsystems.
\end{thm}

The literature makes extensive use of both the Choi operator and the Gour operator, especially in the analysis of dynamical quantum resources. 
Yet the absence of a clear link between these two constructions has long obscured the search for a unified account of quantum superchannels. 
The theorem above resolves this obstacle. 
Although the two operators are mathematically distinct, largely because their underlying system orderings are arranged differently, they become equivalent once permutations of subsystems are treated as immaterial. 
This insight brings the two viewpoints into alignment and provides the conceptual foundation needed for a unified treatment of quantum superchannels.

Building on this foundation, we take the Choi operator defined in Eq.~\eqref{eq:Superchannel_Choi_1} as the canonical representation. 
It serves as the starting point from which the other formulations of a superchannel naturally emerge, allowing us to develop a complete and internally coherent framework for their structure and behavior.
It is worth emphasizing, however, that if one disregards physical constraints and focuses purely on producing matrix representations, then infinitely many ``Choi-like'' constructions become available.
For a quantum channel $\mE$, any pair of bijective linear maps $f$ and $g$ generates an operator
\begin{align}
    J^{\mE}_{(f,g)}
    :=
    f\left(\sum_{ij}\ketbra{i}{j}\otimes\mE\left(g\left(\ketbra{i}{j}\right)\right)\right),
\end{align}
recovering the conventional Choi operator (see Eq.~\eqref{eq:Channel_Choi}) when $f=g=\id$ (See Ref.~\cite{CHOI1975285}), the Jamio\l kowski operator when $f=\id$ and $g=\T$ (See Ref.~\cite{JAMIOLKOWSKI1972275}), and, in general, producing an infinite family of representations that all contain complete information about the underlying dynamics.
This construction extends straightforwardly to superchannels and, more broadly, to non-Markovian quantum dynamics.

Finally, we turn to another fundamental question: 
given the Choi operator of a linear map $\theta$ (see Fig.~\ref{fig:Superchannel}), denoted $J^{\theta}$, how can we determine whether it represents a genuine superchannel?
If $\theta$ is indeed a superchannel, then its Choi operator must satisfy three conditions: 
the complete positivity (CP) condition in Eq.~\eqref{eq:Superchannel_Choi_CP}, the trace-preserving (TP) condition in Eq.~\eqref{eq:Superchannel_Choi_TP}, and the no-signaling (NS) constraint from post-processing to pre-processing in Eq.~\eqref{eq:Superchannel_Choi_NS}.

\begin{align}
    J^{\theta}_{A_1A_2B_1B_2}&\geqslant0,\label{eq:Superchannel_Choi_CP}\\
    \Tr_{B_1B_2}[J^{\theta}_{A_1A_2B_1B_2}]&=\1_{A_1A_2},\label{eq:Superchannel_Choi_TP}\\
    \Tr_{B_2}[J^{\theta}_{A_1A_2B_1B_2}]&=
    J^{\theta}_{A_1B_1}\otimes
    \frac{1}{d_{A_2}}\1_{A_2}.\label{eq:Superchannel_Choi_NS}
\end{align}
Here, $J^{\theta}_{A_1B_1}$ refers to the marginal of $J^{\theta}_{A_1A_2B_1B_2}$. 
In full, this is given by $J^{\theta}_{A_1B_1}:=\Tr_{A_2B_2}[J^{\theta}_{A_1A_2B_1B_2}]$. 
What makes the Choi operator formalism especially powerful is that the converse also holds. 
Whenever $J^{\theta}$ satisfies Eqs.~\eqref{eq:Superchannel_Choi_CP},~\eqref{eq:Superchannel_Choi_TP}, and~\eqref{eq:Superchannel_Choi_NS}, the underlying map $\theta$ is necessarily a superchannel. 
This characterization leads directly to the following theorem.

\begin{thm}
[\bf{Superchannel Choi Operator}~\cite{PhysRevA.80.022339}]
\label{thm:Superchannel_Choi}
A linear map $\theta:A_1A_2\to B_1B_2$ is a superchannel if and only if its Choi operator, defined in Eq.~\eqref{eq:Superchannel_Choi_1}, satisfies the three structural conditions listed in Eqs.~\eqref{eq:Superchannel_Choi_CP},~\eqref{eq:Superchannel_Choi_TP}, and~\eqref{eq:Superchannel_Choi_NS}.
\end{thm}

For completeness, we include a proof of Thm.~\ref{thm:Superchannel_Choi} here. 
Readers interested in the original derivations and further discussion may consult Theorems 3 and 5 of Ref.~\cite{PhysRevA.80.022339}.

\begin{proof}
Based on Cor.~\ref{cor:Superchannel}, the sufficiency is immediate. 
We therefore focus on the necessary direction: assuming Eqs.~\eqref{eq:Superchannel_Choi_CP}–\eqref{eq:Superchannel_Choi_NS} hold, we show that $\theta$ is a superchannel, that is, a deterministic transformation mapping channels to channels.
Eq.~\eqref{eq:Superchannel_Choi_CP} guarantees that $\theta$ sends completely positive maps to completely positive maps. 
The remaining step is to confirm that, for any channel $\mE$, the image $\theta(\mE)$ is trace-preserving. 
We check this through the Choi operator of $\theta(\mE)$, namely,
\begin{align}
    \Tr_{B_2}[J^{\theta(\mE)}]
    &=
    \Tr_{B_2}[J^{\theta}\star J^{\mE}]\notag\\
    &=
    \Tr_{B_1A_2B_2}[J^{\theta}\cdot (J^{\mE})^{\T}]\notag\\
    &=
    \Tr_{B_1A_2}[
    J^{\theta}_{A_1B_1}\otimes
    \frac{1}{d_{A_2}}\1_{A_2}\cdot (J^{\mE})^{\T}
    ]\notag\\
    &=
    \frac{1}{d_{A_2}}
    \Tr_{B_1}[J^{\theta}_{A_1B_1}]\notag\\
    &=
    \1_{A_1},
\end{align}
where the third line is obtained by applying Eq.~\eqref{eq:Superchannel_Choi_NS}, and the final equality is derived using Eq.~\eqref{eq:Superchannel_Choi_TP}.
\end{proof}

%%%%%%%%%%%%%%%%%%%%%%%%%%%%%%%%%%%%%%%%%%%%%%%%%%%%%%%%%%%%%%%%%%%%%%%
%%%%%%%%%%%%%%%%%%%%%%%%%%%%%%%%%%%%%%%%%%%%%%%%%%%%%%%%%%%%%%%%%%%%%%%

\subsection{Representations of Superchannels}\label{subsec:SC_Representations}

Up to this point, we have resolved the inconsistency surrounding the Choi representation of superchannels, establishing a unified framework that clarifies how the different formulations in the literature fit together (see Thm.~\ref{thm:ET}). 
Several key questions, however, remain open.
The first concerns the relationship between the Choi operator $J^{\theta}$ (see Fig.~\ref{fig:Superchannel_Choi_1}(b)) and the map $f_{\theta}$ (see Fig.~\ref{fig:Commutative_Diagram}). 
In particular, given the spectral decomposition of $J^{\theta}$, how can one explicitly construct $f_{\theta}$? 
This question is central, as the map $f_{\theta}$ plays a pivotal role in the proof of the realization theorem (see Thm.~\ref{thm:RT}).

A second line of inquiry arises from the fact that, so far, our analysis has focused exclusively on the Choi operator of a superchannel. 
It is therefore natural to ask: what is the Kraus representation of a superchannel, what does its Stinespring dilation look like, and how should we define its Liouville superoperator? 
These questions complete the picture by extending the standard representations of quantum channels to the higher-order setting.

With these considerations in mind, we begin this subsection by examining the second representation of superchannels: the Kraus decomposition (see Sec.~\ref{subsec:Kraus}).
For a quantum channel $\mE:B_1\to A_2$ and a superchannel $\theta:A_1A_2\to B_1B_2$ (see Fig.~\ref{fig:Superchannel}), we assume that their Choi operators take the following spectral decomposition:
\begin{align}
    J^{\mE}_{B_1A_2}=\sum_{k}\tvec{(M_{k})}\tvec{(M_{k})}^{\dag},
\end{align}
and 
\begin{align}\label{eq:Superchannel_Choi_Spectral}
    J^{\theta}_{A_1A_2B_1B_2}=\sum_{i}\tvec{(N_{i})}\tvec{(N_{i})}^{\dag}.
\end{align}
Their corresponding tensor-networks can be drawn as
\begin{align}\label{TN:Rep_Channel}
    \raisebox{0ex}{\includegraphics[height=7em]{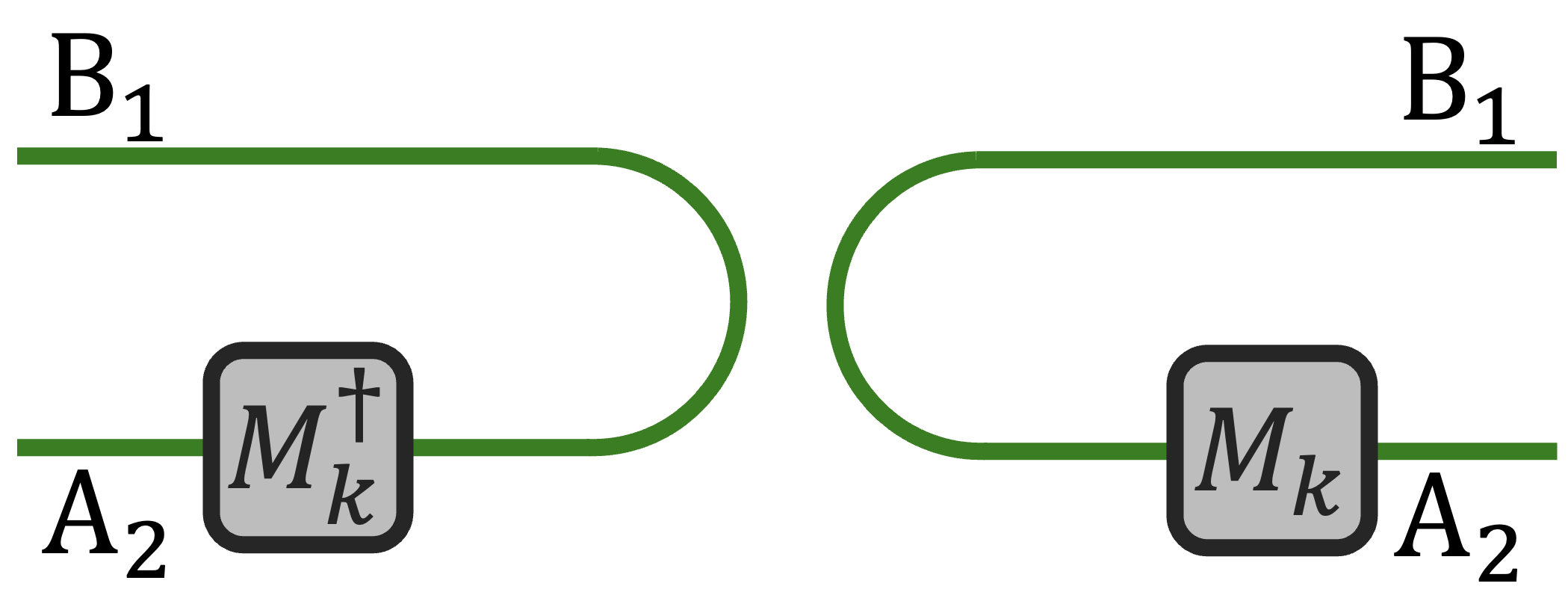}}
\end{align}
and
\begin{align}\label{TN:Rep_Superchannel}
    \raisebox{0ex}{\includegraphics[height=11.2em]{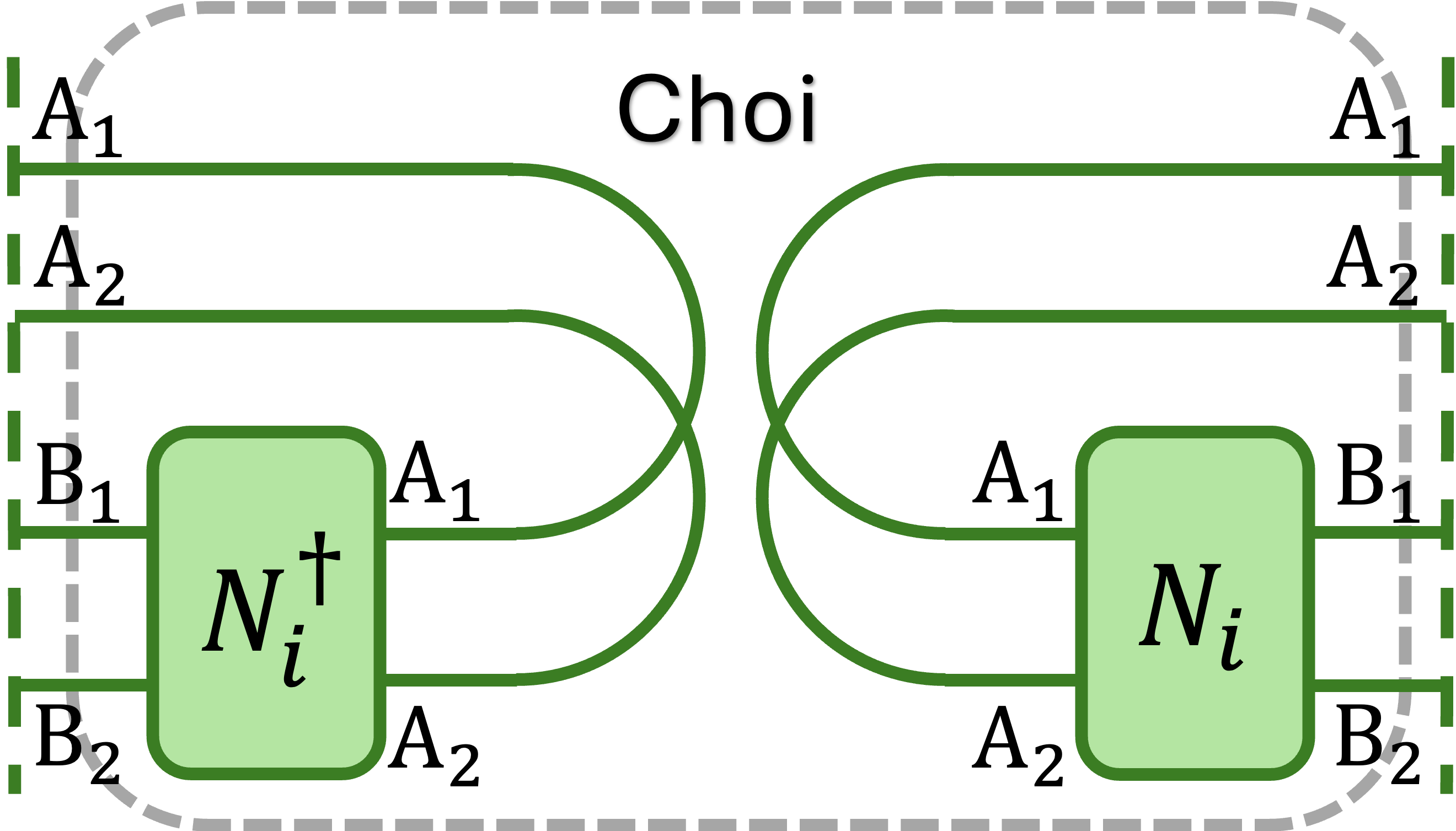}}
\end{align}
respectively. 

For the operators $\{N_i\}$, Eq.~\eqref{eq:Superchannel_Choi_CP} is fulfilled automatically whenever each $N_i$ appears in the tensor network paired symmetrically with its dual.
Eq.~\eqref{eq:Superchannel_Choi_TP}, the trace-preserving requirement, is equivalent to
\begin{align}\label{TN:Superchannel_Choi_TP}
    \raisebox{0ex}{\includegraphics[height=6.4em]{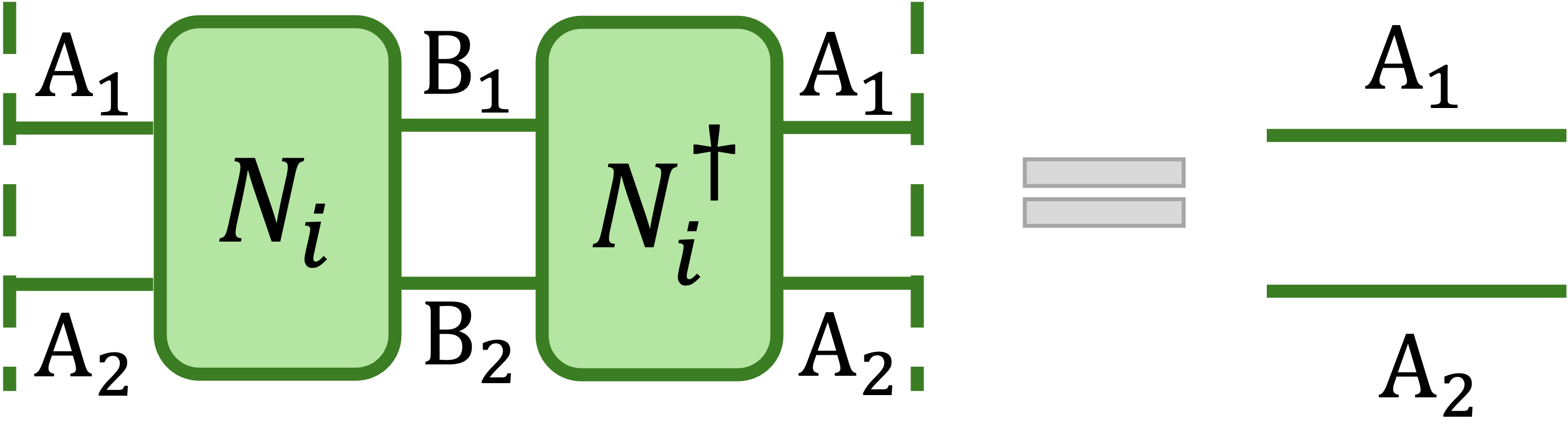}}
\end{align}
Finally, Eq.~\eqref{eq:Superchannel_Choi_NS} can be rewritten as
\begin{widetext}
\begin{align}\label{TN:Superchannel_Choi_NS}
    \raisebox{0ex}{\includegraphics[height=11em]{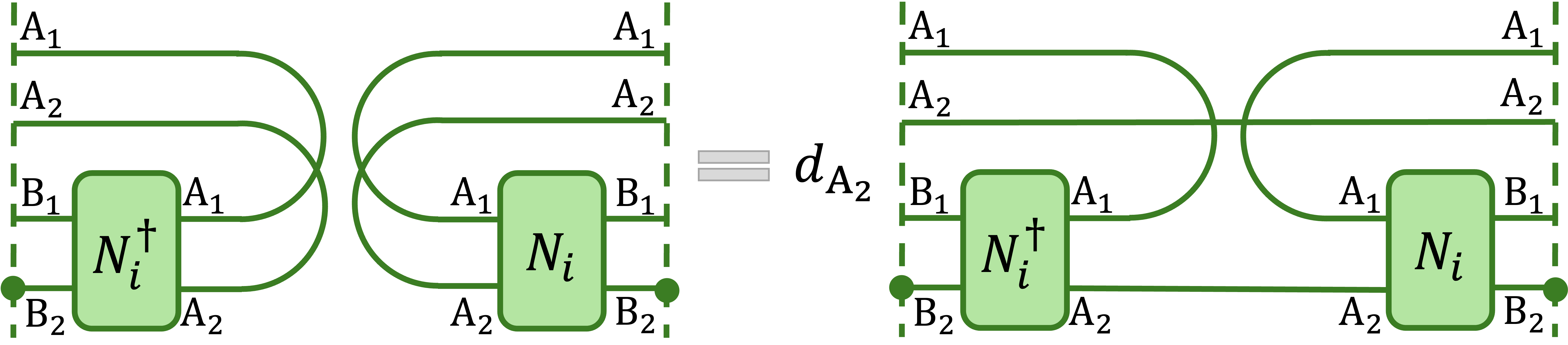}}
\end{align}
\end{widetext}
which encapsulates the NS from post-processing to pre-processing.

We follow the conventions established in the proof of the realization theorem (see Thm.~\ref{thm:RT}). 
Kraus operators of the channel $\mE$ appear in black, whereas those arising from the action of the superchannel $\theta$ are marked in green.
When specifying the systems, we place inputs above outputs, so input systems appear first from top to bottom. Under this ordering, the Choi operator $J^{\mE}$ of channel $\mE$ acts on $B_1A_2$, while the Choi operator $J^{\theta}$ of superchannel $\theta$ acts on $A_1A_2B_1B_2$.
Since both multiple inputs and multiple outputs are present, the time ordering matters, and the two Choi operators cannot be linked directly.

Under our convention that operators are identified up to permutations, we may first apply an appropriate SWAP to the Choi operator of $\mE$ to align the systems. Once this alignment is made, the link product (see Sec.~\ref{subsec:Link_Product}) can be taken, giving the resulting form of $J^{\theta(\mE)}=J^{\theta}\star J^{\mE}=\Tr_{B_1A_2}[J^{\theta}\cdot (J^{\mE})^{\T}]$:

\begin{widetext}

\begin{align}\label{TN:Theta_E_1}
    \raisebox{0ex}{\includegraphics[height=11em]{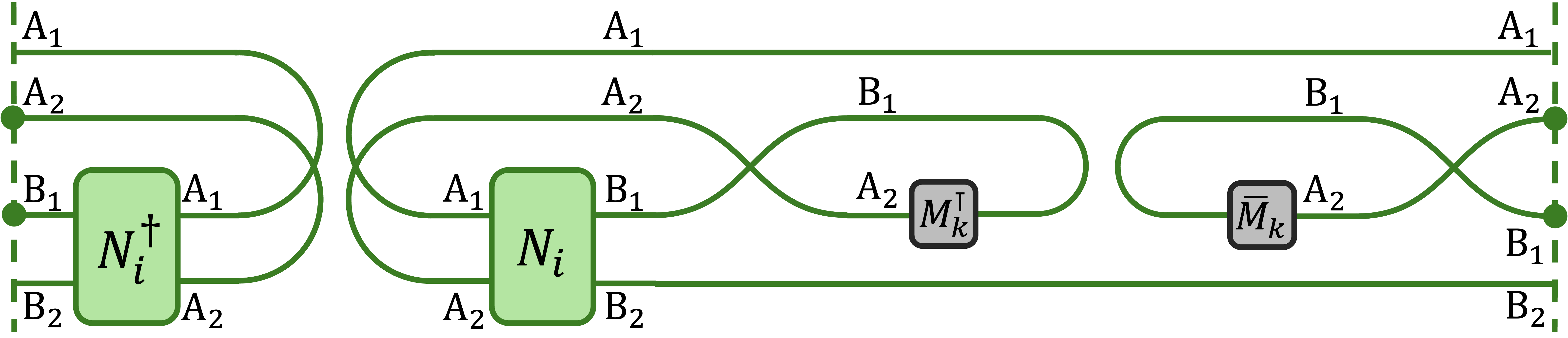}}
\end{align}
which is equivalent to
\begin{align}\label{TN:Theta_E_2}
    \raisebox{0ex}{\includegraphics[height=11em]{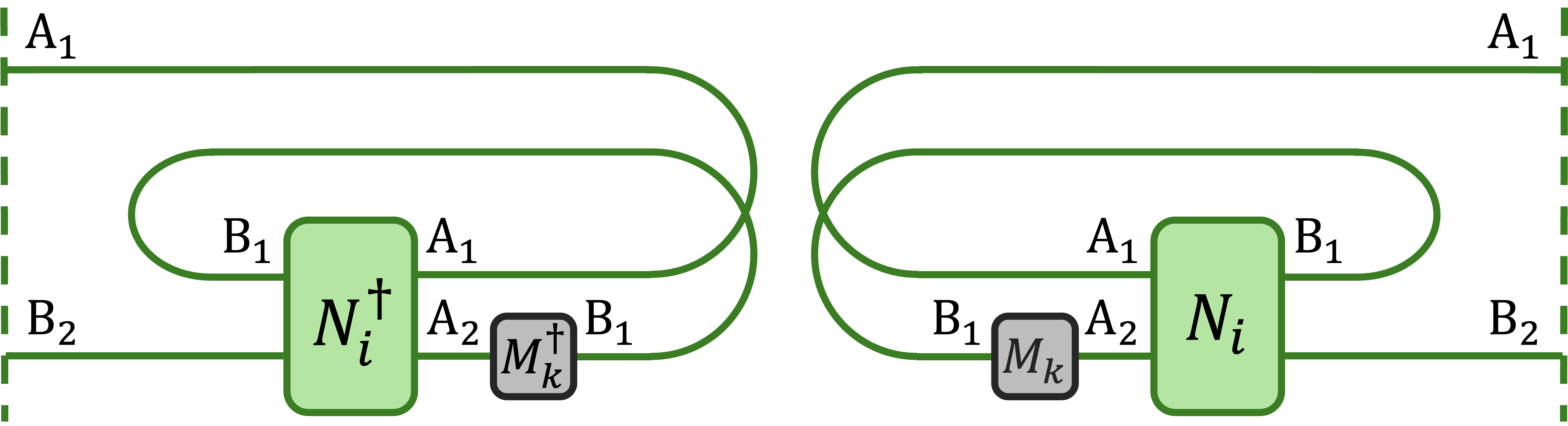}}
\end{align}

\end{widetext}

\noindent
Hence, by introducing the operators $Q_i$ as the partial matricization (see Eq.~\eqref{eq:Partial_Matricization}) over the subsystem $B_1$, defined as
\begin{align}\label{eq:Q_i}
    Q_i:= \tmat_{B_1}(N_i),
\end{align}
with the associated tensor-network
\begin{align}\label{TN:Q_i}
    \raisebox{0ex}{\includegraphics[height=11em]{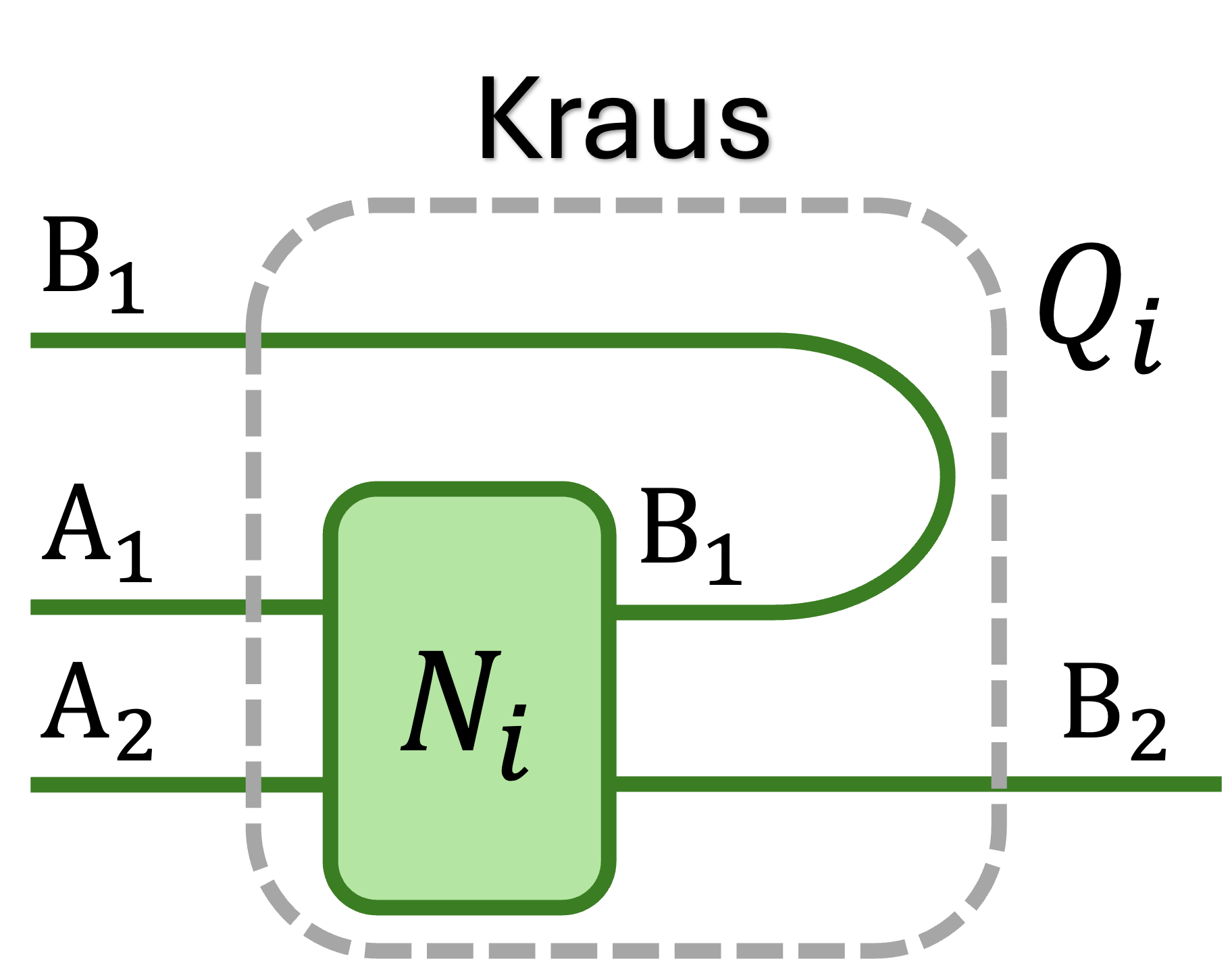}}
\end{align}
we obtain

\begin{thm}
[\bf{Superchannel Kraus Decomposition}]
\label{thm:Superchannel_Krasu}
A quantum channel $\mE:B_1\to A_2$ is transformed into a new channel $\theta(\mE):A_1\to B_2$ through the action of a superchannel $\theta:A_1A_2\to B_1B_2$ (see Fig.~\ref{fig:Superchannel}).
Once the Choi operator $J^{\theta}$ of $\theta$ is expressed through the spectral decomposition in Eq.~\eqref{eq:Superchannel_Choi_Spectral}, this transformation becomes fully explicit: the resulting channel takes the form
\begin{align}\label{eq:Superchannel_Kraus}
    \theta(\mE)
    =
    &\sum_{i}Q_i
    \left(\id_{B_1\to B_1}\otimes\mE_{B_1\to A_2}(\Gamma_{B_1B_1})\right)
    Q_i^{\dag}\notag\\
    =
    &\sum_{i}Q_i J^{\mE} Q_i^{\dag},
\end{align}
where $Q_i$ is defined in Eq.~\eqref{eq:Q_i}, and its structure is illustrated in Eq.~\eqref{TN:Q_i}.
\end{thm}

Equation~\eqref{eq:Superchannel_Kraus} generalizes the Kraus decomposition introduced in Sec.~\ref{subsec:Kraus} from quantum channels to the superchannel setting.
This becomes evident by examining the special case in which the channel $\mE$ in Fig.~\ref{fig:Superchannel} is replaced by a state preparation channel.
In this case, the input system of $\mE$ is trivial, formally, $B_1=\mathbb{C}$, and the Choi operator of $\mE$ reduces to the corresponding density operator, i.e., 
\begin{align}
    J^{\mE}_{B_1A_2}=\rho_{A_2}.
\end{align}
At the same time, the pre-processing stage of the superchannel becomes trivial as $A_1=B_1=\mathbb{C}$, leaving only the post-processing component. 
This remaining part is simply a quantum channel, whose Kraus operators will be denoted by $\{G_i\}$.
In this setting, we have $Q_i=G_i$, and 
\begin{align}
    J^{\theta}_{A_1A_2B_1B_2}=
    J^{\theta}_{A_2B_2}=
    \sum_{i}\tvec{(G_{i})}\tvec{(G_{i})}^{\dag}.
\end{align}
Meanwhile, the action of $\theta(\mE)$ collapses to the expression in Eq.~\eqref{eq:Kraus}, namely
\begin{align}
    \theta(\mE)=\sum_i G_i\rho G_i^{\dag}.
\end{align}
Hence, Eq.~\eqref{eq:Superchannel_Kraus} recovers the standard Kraus representation when the input is a state preparation channel, and thereby establishes the Kraus decomposition for general quantum superchannels.

It is worthwhile to examine the structural properties of the operators $Q_i$ (see Eq.~\eqref{TN:Q_i}), especially whether they fulfill the completeness relation
\begin{align}
    \sum_i Q_i^{\dag}Q_i=\1,
\end{align}
mirroring the Kraus operators of a channel.
However, they fail to satisfy this identity in general; 
instead, they satisfy a modified relation. 
By taking the partial trace over the subsystem $B_1$, the expression $\Tr_{B_1}[\sum_i Q_i^{\dag}Q_i]$ reduces to

\begin{widetext}
\begin{align}\label{TN:Theta_E_Q_i}
    \raisebox{0ex}{\includegraphics[height=10em]{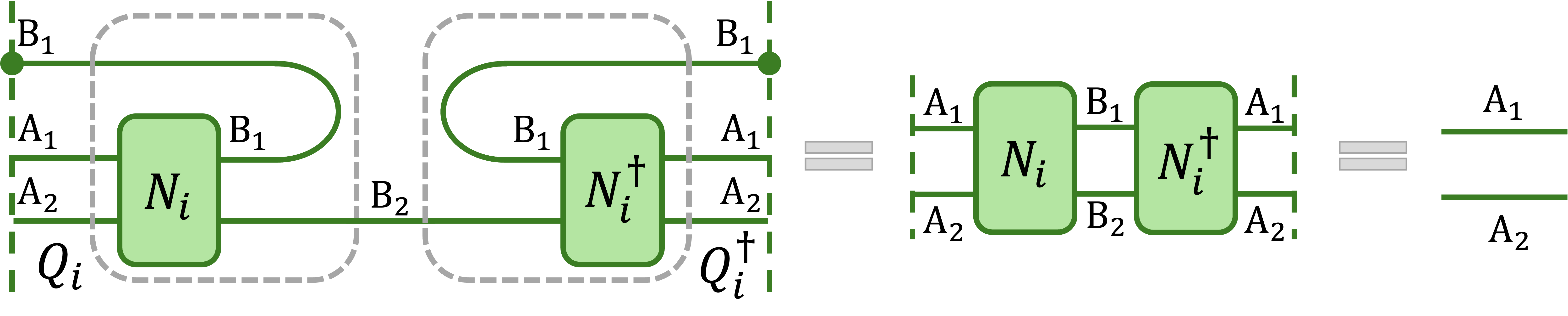}}
\end{align}
\end{widetext}

\noindent
which is equivalent to saying
\begin{align}\label{eq:Theta_E_Q_i}
    \Tr_{B_1}\left[\sum_i Q_i^{\dag}Q_i\right]=\1_{A_1A_2}.
\end{align}

Let $\mQ$ denote the completely positive map generated by the collection $\{Q_i\}$; that is
\begin{align}
    \mQ(\cdot)
    :=
    \sum_iQ_i\cdot Q_i^{\dag}.
\end{align}
With this notation, the transformed channel $\theta(\mE)$ becomes
\begin{align}\label{TN:Theta_E}
    \raisebox{0ex}{\includegraphics[height=9em]{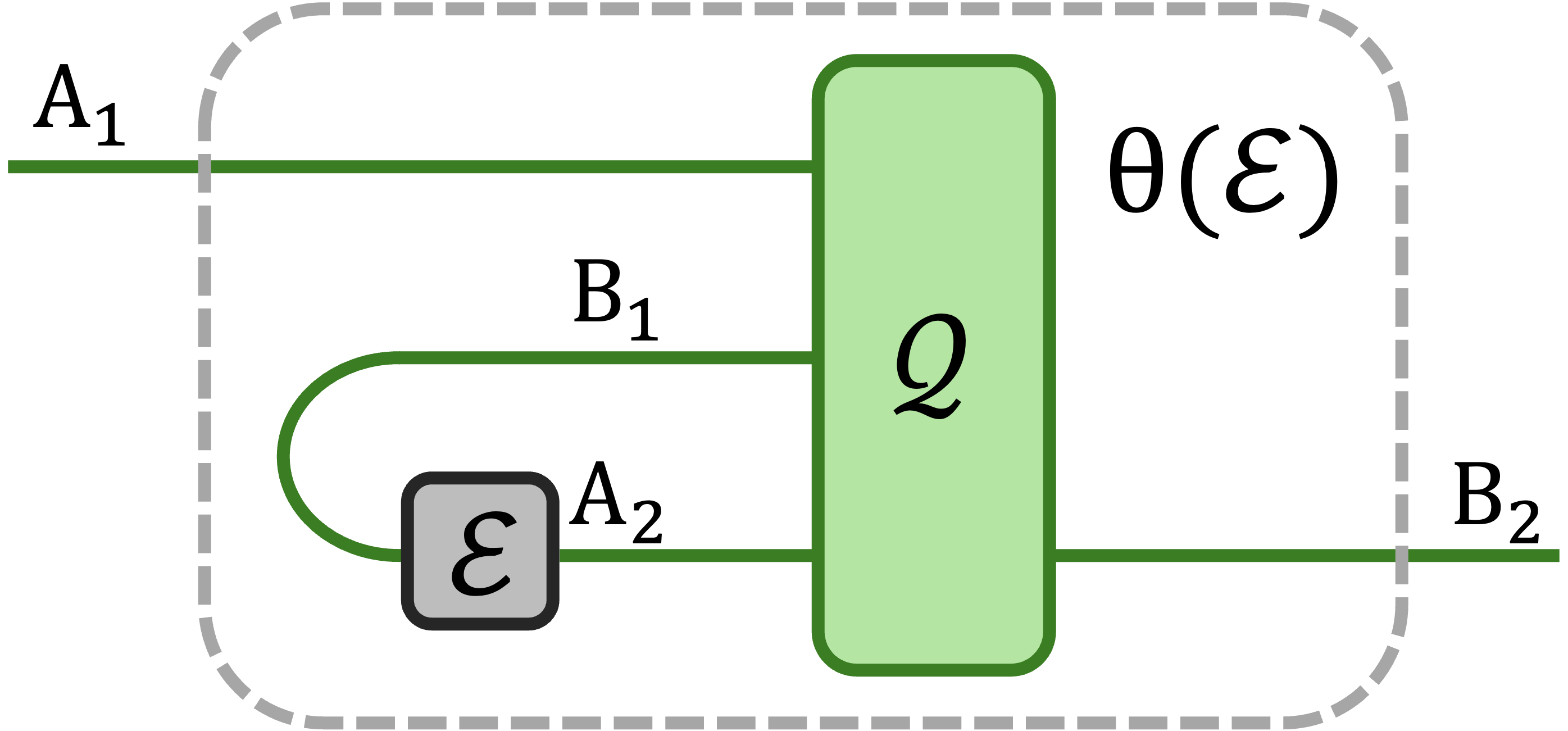}}
\end{align}
Owing to the symmetry of the construction, only half of the tensor-network diagram is shown.
Here, the ordering of the subsystems, especially $A_1$ and $B_1$, has been adjusted to render the diagram cleaner and more compact.
In tensor-network representations, only matching subsystems can be connected, and for this reason we freely use the forms of $Q_i$ shown in Eqs.~\eqref{TN:Q_i} and~\eqref{TN:Theta_E} interchangeably. 
This notational flexibility will not lead to ambiguity.

We next turn to the third representation of quantum superchannels: the Stinespring dilation (see Sec.~\ref{subsec:Stinespring}).
Building on the Kraus operators $Q_i$ from Eq.~\eqref{TN:Q_i}, we define the operator $V$ associated with the superchannel $\theta$ as
\begin{align}\label{TN:Superchannel_Stinespring_V}
    \raisebox{0ex}{\includegraphics[height=14em]{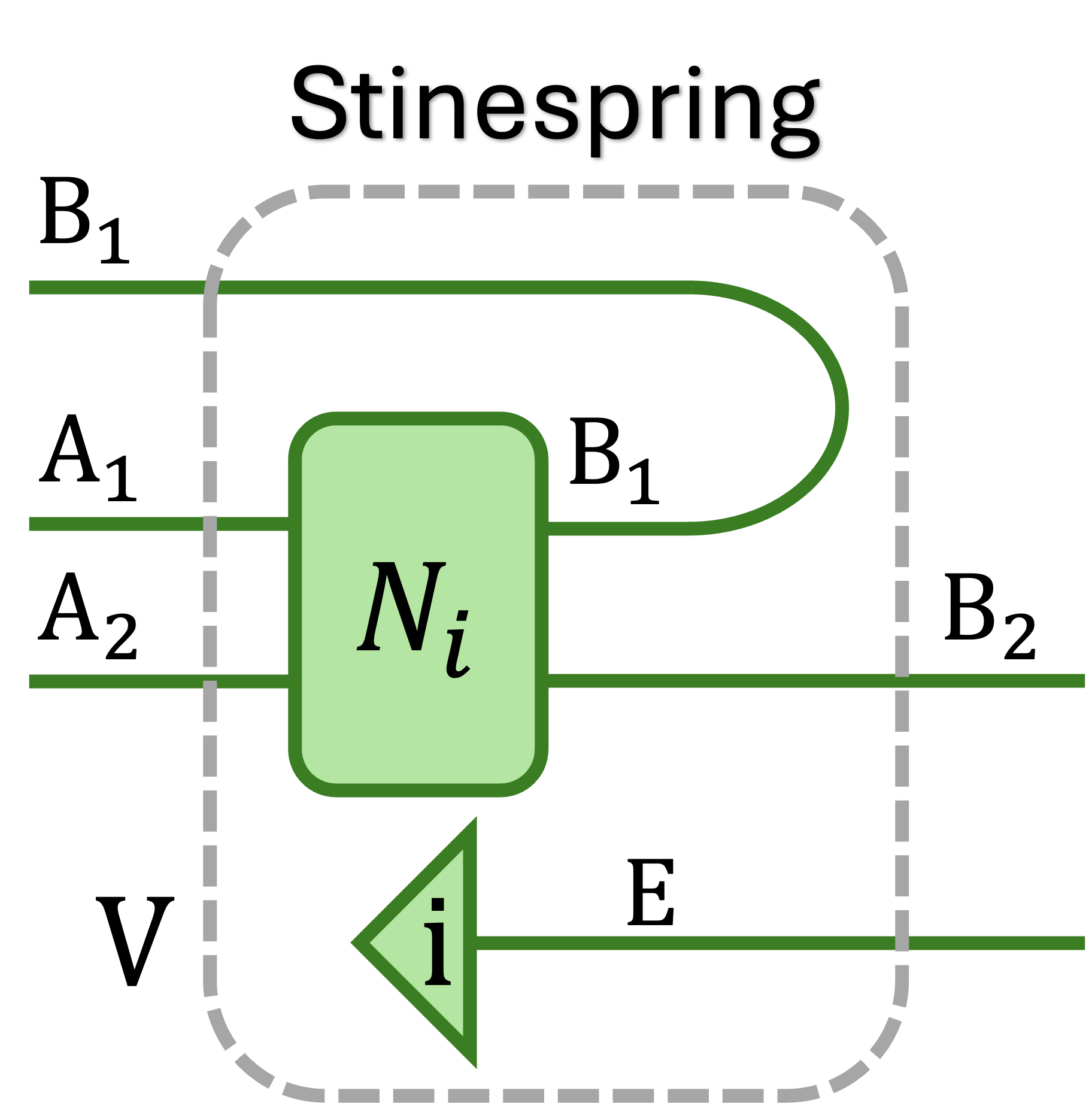}}
\end{align}
A direct calculation confirms that

\begin{thm}
[\bf{Superchannel Stinespring Dilation}]
\label{thm:Superchannel_Stinespring}
A superchannel $\theta:A_1A_2\to B_1B_2$ maps any channel $\mE:B_1\to A_2$ to a new channel $\theta(\mE):A_1\to B_2$ (see Fig.~\ref{fig:Superchannel}).
If the Choi operator $J^{\theta}$ has the spectral form of Eq.~\eqref{eq:Superchannel_Choi_Spectral}, then the output channel $\theta(\mE)$ reads
\begin{align}\label{eq:Superchannel_Stinespring}
    \theta(\mE)= \Tr_{E}[VJ^{\mE}V^{\dag}].
\end{align}
where the operator $V$ is the one introduced in Eq.~\eqref{TN:Superchannel_Stinespring_V}.
\end{thm}

In this form, the operator $V$ plays the role analogous to Eq.~\eqref{TN:Stinespring} in the Stinespring dilation of quantum channels, extending that construction to the superchannel setting.
The correspondence becomes transparent in the special case where $\mE$ acts as a state preparation channel: the superchannel collapses to a quantum channel, the pre-processing stage is trivialized with $A_1=B_1=\mathbb{C}$, and $V$ reduces precisely to the isometry introduced in Eq.~\eqref{TN:Stinespring}.
Outside this degenerate scenario, however, $V$ is no longer constrained to be an isometry.
Rather, it satisfies a relaxed normalization condition: upon tracing out the subsystem $B_1$, one obtains
\begin{align}
    \Tr_{B_1}[V^{\dag}V]=\1_{A_1A_2},
\end{align}
a relation obtained directly by inserting Eq.~\eqref{TN:Superchannel_Stinespring_V} into Eq.~\eqref{TN:Theta_E_Q_i}.

As a preparatory step toward the final representation, we make explicit the relationship between $J^{\theta}$ and $f_{\theta}$.
This is achieved by introducing a family of operators $K_i$, obtained by vectorizing $N_i$ over subsystem $A_1$ and subsequently matricizing over subsystem $B_1$.
Since these transformations act on disjoint subsystems, the order in which they are applied is irrelevant.
Mathematically, $K_i$ is given by
\begin{align}\label{eq:K_N}
    K_i:= \tmat_{B_1}(\tvec_{A_1}(N_i)),
\end{align}
with the following tensor-network diagram
\begin{align}\label{TN:K_i}
    \raisebox{0ex}{\includegraphics[height=12em]{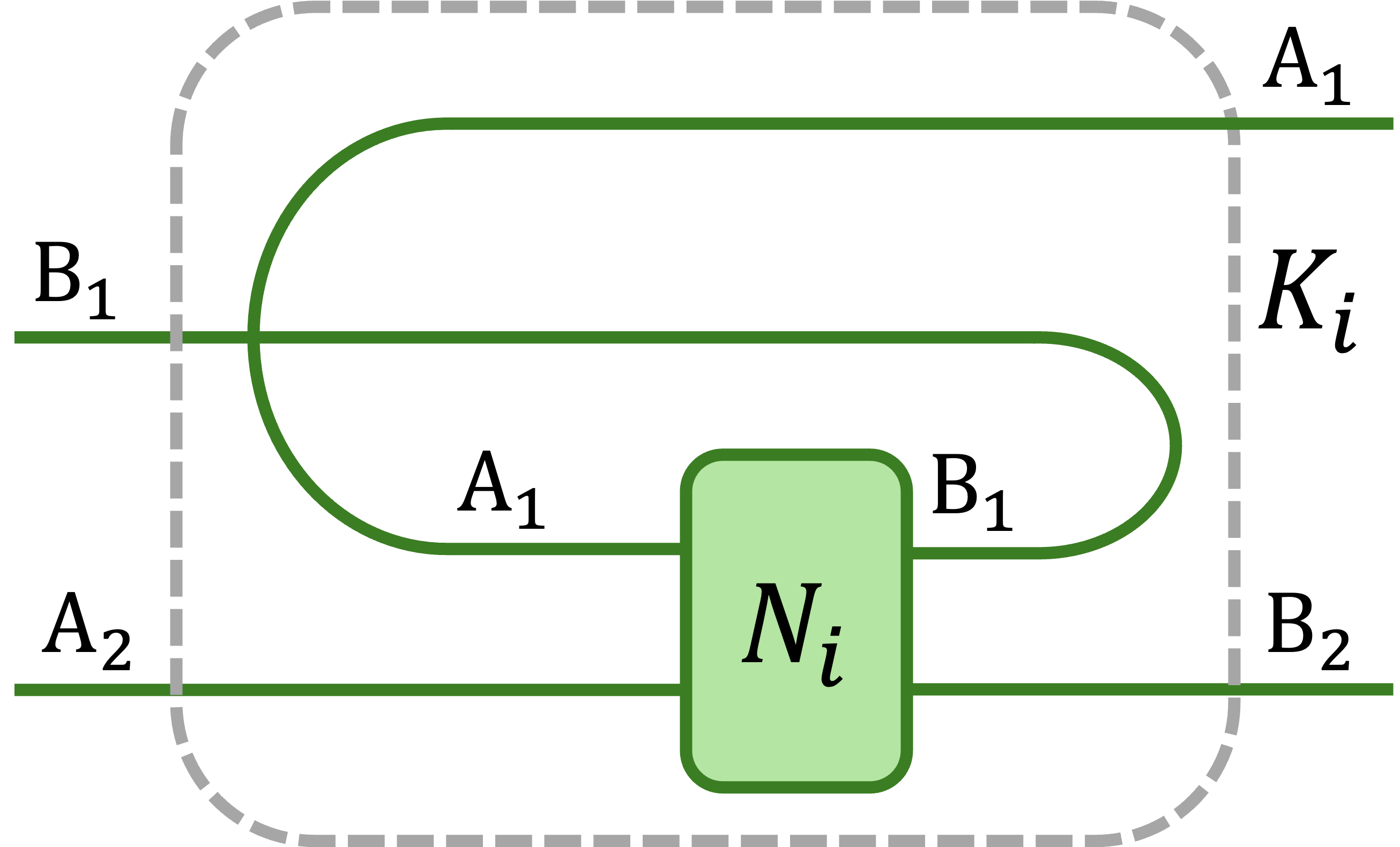}}
\end{align}

With the operators $K_i$ now at hand, we return to the central task posed at the beginning of this subsection: 
how to extract an explicit expression for the map $f_{\theta}$ directly from the Choi operator of the superchannel $\theta$.
In this regard, Eq.~\eqref{TN:Theta_E_2} shows that
\begin{align}\label{eq:f_theta_Kraus}
    f_{\theta}(J^{\mE})
    = J^{\theta(\mE)}
    = \sum_i K_i J^{\mE} K_i^{\dag}.
\end{align}
The statement is formalized in the lemma below, which also lays the groundwork for introducing the Liouville superoperator associated with superchannel $\theta$.

\begin{lem}
[\bf{Kraus Decomposition of $f_{\theta}$}]
\label{lem:Kraus_f_theta}
    For any superchannel $\theta$, there exists a completely positive map $f_{\theta}$ that takes the Choi operator of a channel $\mE$ to that of the output channel $\theta(\mE)$ (see Fig.~\ref{fig:Commutative_Diagram}).
    Let $\{K_i\}$ be a Kraus representation of $f_{\theta}$, so that $f_{\theta}(\cdot)=\sum_iK_i\cdot K_i^{\dag}$.
    Meanwhile, denote by $J^{\theta}$ the Choi operator of $\theta$, with spectral decomposition given in Eq.~\eqref{eq:Superchannel_Choi_Spectral}.
    As shown in Eqs.~\eqref{eq:K_N} and~\eqref{TN:K_i}, the operators $N_i$ arising from this spectral decomposition are in one-to-one correspondence with the Kraus operators $K_i$ of $f_{\theta}$, thereby linking the intrinsic structure of $J^{\theta}$ to the operational action of $f_{\theta}$.
\end{lem}

As a by-product, Lem.~\ref{lem:Kraus_f_theta} may be combined with Cor.~\ref{cor:MC} to obtain the memory cost of simulating the superchannel $\theta$ directly in terms of the operators $N_i$ arising from the spectral decomposition of $J^{\theta}$ (see Eq.~\eqref{eq:Superchannel_Choi_Spectral}), rather than the Kraus operators $K_i$ appearing in the decomposition of $f_{\theta}$ (see Eq.~\eqref{eq:f_theta_Kraus}), namely

\begin{cor}
[\bf{Memory Cost}]
\label{cor:MC_2}
For a superchannel $\theta$ whose spectral decomposition is provided in Eq.~\eqref{eq:Superchannel_Choi_Spectral}, the associated minimal memory cost $d_{\theta}$ is characterized by
\begin{align}\label{TN:Memory_Cost_2}
    \raisebox{0ex}{\includegraphics[height=9.2em]{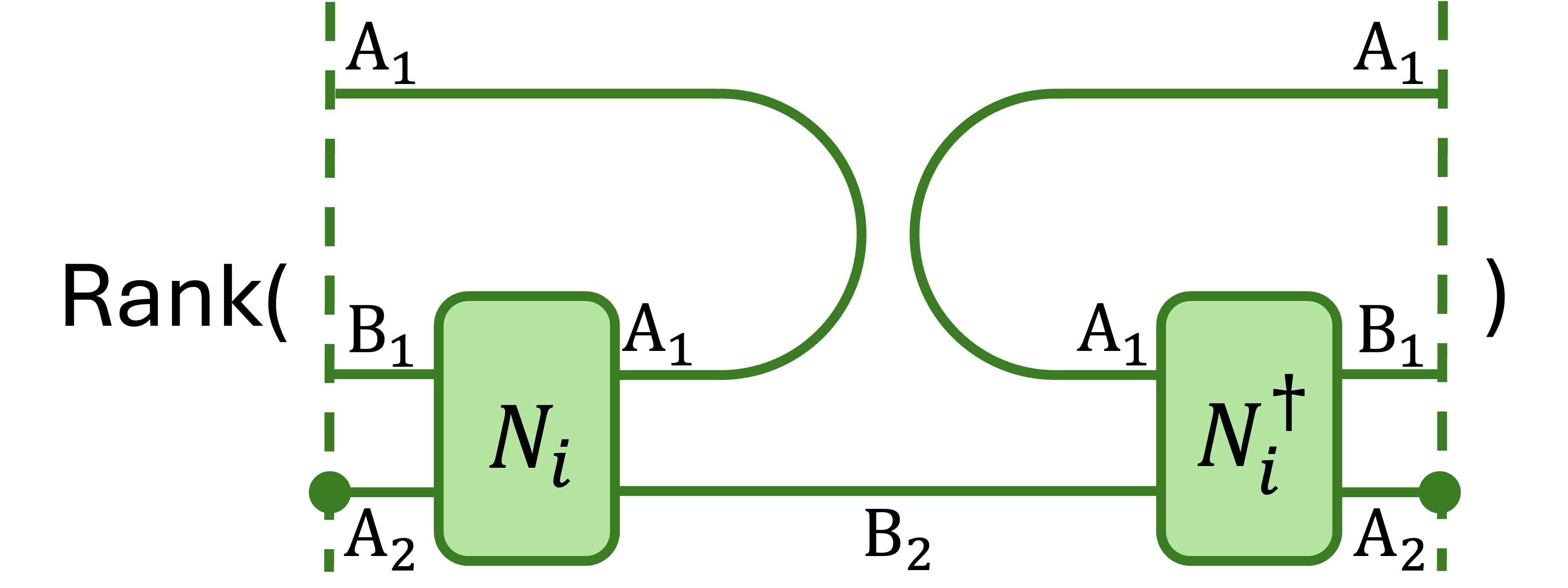}}
\end{align}
\end{cor}

We proceed to the fourth and final representation of superchannels: 
the Liouville superoperator (see Sec.~\ref{subsec:Liouville}).
Upon vectorizing the Choi operator of $\theta(\mE)$, we derive
\begin{align}
    \tvec(J^{\theta(\mE)})=
    \left(\sum_i \overline{K}_i\otimes K_i\right)\,
    \tvec(J^{\mE}).
\end{align}
Introducing
\begin{align}\label{eq:K}
    K:= \sum_i \overline{K}_i\otimes K_i,
\end{align}
we arrive at the operator that captures the superchannel in its Liouville representation, namely, the Liouville superoperator associated with superchannel $\theta$.
Formally, the construction can be summarized in the following theorem

\begin{thm}
[\bf{Superchannel Liouville Superoperator}]
\label{thm:Superchannel_Liouville}
Given a superchannel $\theta$ and a quantum channel $\mE$, the action of $\theta$ yields a new channel $\theta(\mE)$ (see Fig.~\ref{fig:Superchannel}).
In the Liouville representation, the vectorized Choi operators of these channels are connected through the Liouville superoperator $K$, defined in Eq.~\eqref{eq:K}, as
\begin{align}\label{eq:Superchannel_Liouville}
    \tvec(J^{\theta(\mE)})=K\cdot\tvec(J^{\mE}).
\end{align}
\end{thm}

Once again, the proof admits a transparent tensor-network representation, shown below

\begin{widetext}
\begin{align}\label{TN:Superchannel_Liouville}
    \raisebox{0ex}{\includegraphics[height=16em]{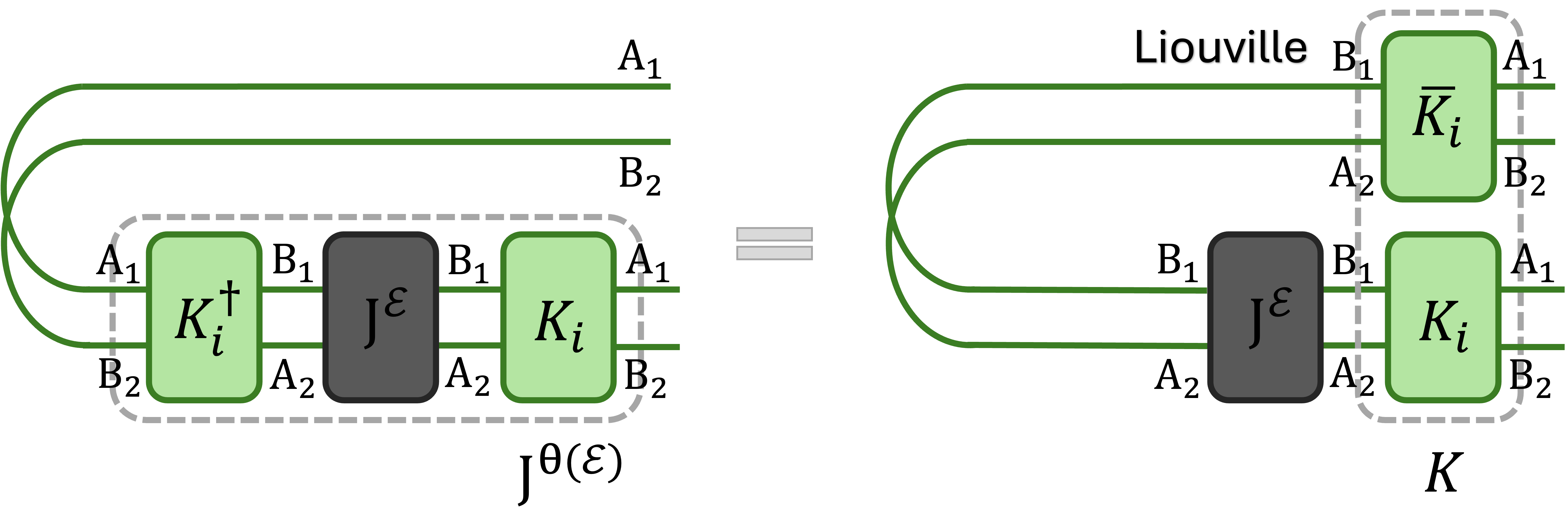}}
\end{align}
\end{widetext}

Altogether, these constructions show that the higher-order structure of a superchannel can be unpacked with the same clarity and completeness as that of quantum channels. 
Starting from its Choi operator (see Eq.~\eqref{TN:Rep_Superchannel}), each subsequent representation, the Kraus decomposition (see Eq.~\eqref{TN:Q_i}), the Stinespring dilation (see Eq.~\eqref{TN:Superchannel_Stinespring_V}), and the Liouville superoperator (see Eq.~\eqref{TN:Superchannel_Liouville}), arises naturally through the tensor-networks, revealing complementary aspects of the same underlying transformation. 
What may at first seem like distinct mathematical frameworks thus converge into a unified operational picture: 
a superchannel reshapes the Choi operator of the input channel in a manner dictated entirely by its causal architecture (see Fig.~\ref{fig:Superchannel}). 
With this, the second major gap in the existing framework of superchannel -- the absence of a complete set of structural representations -- is closed.
This synthesis positions superchannels squarely within the broader theoretical landscape of higher-order quantum theory and provides a coherent foundation for exploring their structural and operational implications.

%%%%%%%%%%%%%%%%%%%%%%%%%%%%%%%%%%%%%%%%%%%%%%%%%%%%%%%%%%%%%%%%%%%%%%%
%%%%%%%%%%%%%%%%%%%%%%%%%%%%%%%%%%%%%%%%%%%%%%%%%%%%%%%%%%%%%%%%%%%%%%%

\section{Correlation and Causality Breaking}
\label{sec:EB}

Entanglement-breaking (EB) channels offer a clear benchmark for understanding when a quantum process loses the ability to sustain non-classical correlations~\cite{RevModPhys.81.865}. 
Once a channel becomes EB, all entanglement across its input-output interface is irreversibly lost, and the resulting output states can be fully simulated using only local operations and classical communication (LOCC)~\cite{Chitambar2014}. 
Examining this boundary reveals when and why quantum correlations disappear, clarifies the mechanisms through which quantum advantages degrade, and identifies the regimes in which quantum resources can, or cannot, be preserved and leveraged~\cite{RevModPhys.91.025001}. 
In this sense, EB channels mark a fundamental limit of quantum processes and serve as indispensable tools for analyzing both the structure of entanglement and the operational significance of its loss.

From an operational standpoint, establishing that a physical process is not EB is the most basic criterion for its use as a quantum memory or for any task that relies on entanglement as a resource~\cite{PhysRevX.8.021033}. 
Yet much of the existing literature concentrates on EB behavior for quantum states -- an idealization that overlooks the temporal structure inherent in realistic quantum devices.
In practice, quantum systems interact with their environments over extended periods, are queried or driven repeatedly, and may serve as storage elements across multiple rounds of use.
In such settings, the relevant question is no longer whether a single channel breaks entanglement, but whether the entire temporal process -- with its history, internal correlations, and memory effects -- admits only classical explanations.

It is this broader and more experimentally relevant temporal perspective that motivates a closer examination of EB superchannels. 
Since superchannels admit multiple complementary viewpoints, the notion of ``entanglement breaking'' at the higher-order level is not predetermined to be unique. 
The perspective i outlined in Sec.~\ref{subsec:Bi_Channels} underpins the definition proposed in Ref.~\cite{Chen2020entanglement}, but this naturally prompts several questions: Is that construction the only coherent way to define EB for superchannels? 
Can the alternative viewpoint ii in Sec.~\ref{subsec:Bi_Channels} lead to a different, yet still meaningful, formulation of an EB superchannel? And if multiple definitions exist, do they correspond to distinct physical behaviors, or does an alternative formulation reveal aspects of temporal entanglement loss that the original one overlooks? 
In what follows, we examine these issues systematically and introduce an operational formulation of EB superchannels that aligns with the structure of multi-time quantum dynamics encountered in realistic scenarios.
Alongside quantum correlations, we also examine superchannels that destroy quantum common causes, leading naturally to the notion of {\it common cause breaking superchannels}. 
These developments underscore a unifying theme: the intrinsically multi-time nature of superchannels reveals a richer landscape of correlation and casualty breaking mechanisms than is apparent from single-step transformations alone.

%%%%%%%%%%%%%%%%%%%%%%%%%%%%%%%%%%%%%%%%%%%%%%%%%%%%%%%%%%%%%%%%%%%%%%%
%%%%%%%%%%%%%%%%%%%%%%%%%%%%%%%%%%%%%%%%%%%%%%%%%%%%%%%%%%%%%%%%%%%%%%%

\subsection{Entanglement Breaking Channels}
\label{subsec:EBC}

Before turning to the main discussion, it is helpful to briefly recall the notion of EB and the key results surrounding EB channels. 
Consider a quantum channel $\mE:A\to B$. 
We say that $\mE$ is EB if, for every ancillary system $R$ and joint input state $\rho_{RA}$, the output $\id_{R\to R}\otimes\mE_{A\to B}(\rho_{RA})$ is always separable across the bipartition $R\,|\,B$, as depicted in Fig.~\ref{fig:EBC}.
In other words, no matter what state is fed into the channel, even one that is maximally entangled, the channel invariably destroys all entanglement.
The following theorem provides a necessary and sufficient condition for a channel to be EB, formulated directly in terms of its Choi operator (see Sec.~\ref{subsec:Choi}).

\begin{figure}[t]
\centering   
\includegraphics[width=0.34\textwidth]{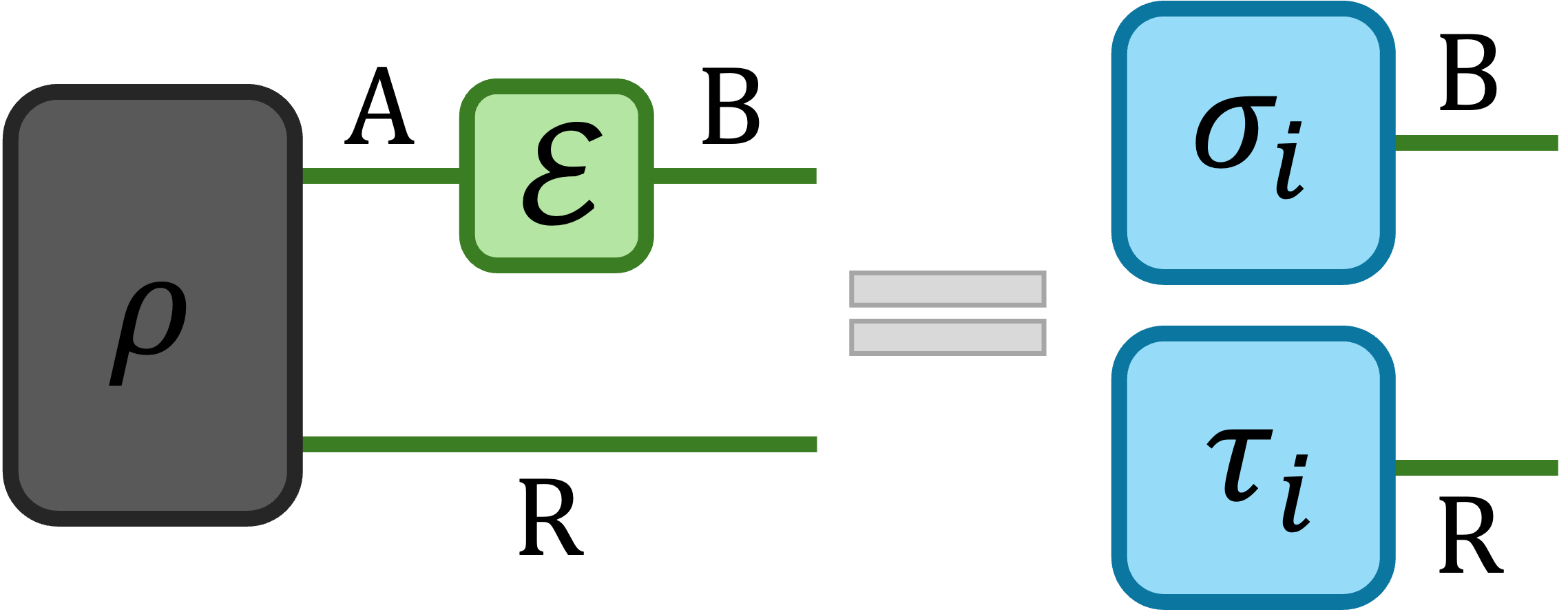}
\caption{(Color online) 
\textbf{EB Channel}.  
A channel $\mE:A\to B$ is entanglement breaking if, for every ancillary system $R$ and every joint state $\rho_{RA}$, the output $\id\otimes\mE(\rho)$ is always separable across $R\,|\,B$, i.e., $\id\otimes\mE(\rho)=\sum_i p_i\sigma_i\otimes\tau_i$ some probability distribution $p_i$, with $\sigma_i\geqslant0$ and $\tau_i\geqslant0$. 
In the diagrammatic representation, the blue shading of $\sigma_i$ and $\tau_i$ signals that they are summed over a common index $i$ in accordance with our tensor-network convention.
For simplicity, the probability distributions are omitted in the diagram.
}
\label{fig:EBC}
\end{figure}

\begin{thm}
[\bf{EB Channel}~\cite{Horodecki2023Entanglement}]
\label{thm:EB_Channel}
A channel $\mE:A\to B$ is EB if and only if its Choi operator is separable across $A\,|\,B$; that is,
\begin{align}\label{eq:EB_Channel}
    J^{\mE}_{AB}=\sum_i X_i\otimes Y_i,
\end{align}
with $X_i\geqslant0$ and $Y_i\geqslant0$ acting on $A$ and $B$, respectively.
\end{thm}

Equation~\eqref{eq:EB_Channel} makes clear that every EB channel can be implemented by a measure-and-prepare scheme: the input is measured, and the outcome determines the result state. 
Indeed,
\begin{align}
    J^{\mE}\star\rho
    =
    \sum_i\Tr[X_i^{\T}\rho]Y_i
    =
    \sum_i\Tr[\left(\Tr[Y_i]X_i^{\T}\right)\rho]
    \frac{Y_i}{\Tr[Y_i]},
\end{align}
where 
\begin{align}
    M_i:=\Tr[Y_i]X_i,
\end{align}
defines a positive operator-valued measure (POVM) and each
\begin{align}
    \sigma_i:=\frac{Y_i}{\Tr[Y_i]},
\end{align}
is a quantum state.
In this way, the action of $\mE$ can be represented diagrammatically as a measurement $\{M_i\}$ followed by the preparation of $\sigma_i$ conditioned on outcome $i$.

\begin{align}\label{TN:EBC_Realization}
    \raisebox{0ex}{\includegraphics[height=6em]{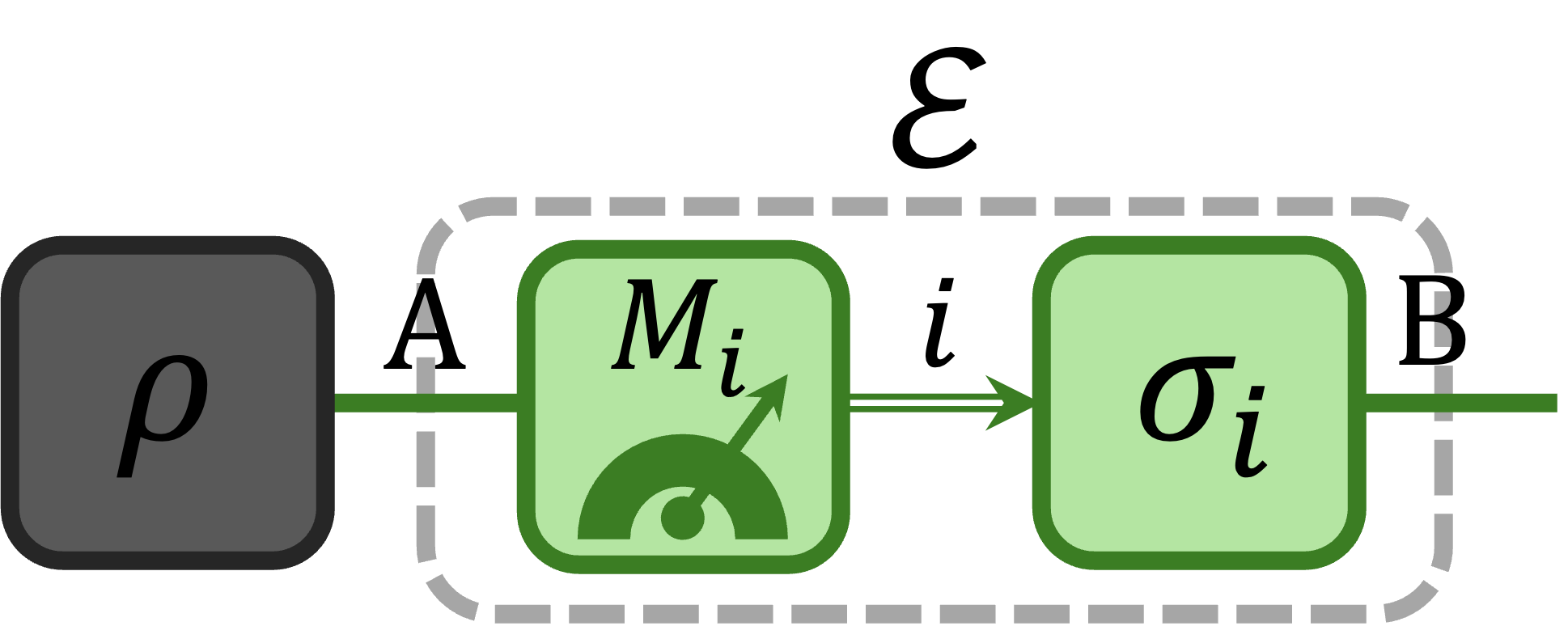}}
\end{align}

The theory of EB channels is well established, and every such channel admits a deterministic measurement-and-preparation realization, as shown in Eq.~\eqref{TN:EBC_Realization}. 
In contrast, once we move to higher-order quantum processes, particularly those exhibiting non-Markovian dynamics, such deterministic implementations are no longer guaranteed to exist.

%%%%%%%%%%%%%%%%%%%%%%%%%%%%%%%%%%%%%%%%%%%%%%%%%%%%%%%%%%%%%%%%%%%%%%%
%%%%%%%%%%%%%%%%%%%%%%%%%%%%%%%%%%%%%%%%%%%%%%%%%%%%%%%%%%%%%%%%%%%%%%%

\subsection{Multipartite Entanglement Breaking Channels}
\label{subsec:MEBC}

To motivate our notion of EB superchannels, it is helpful to examine how the standard definition of an EB channel (see Fig.~\ref{fig:EBC}) extends from the point-to-point setting to the multipartite one. 
This perspective is consistent with our guiding philosophy of treating a superchannel as a bipartite channel (see Fig.~\ref{fig:Bipartite_Channel}), in line with the generalized Occam's razor (see Sec.~\ref{subsec:Bi_Channels}).

\begin{figure}[t]
\centering   
\includegraphics[width=0.4\textwidth]{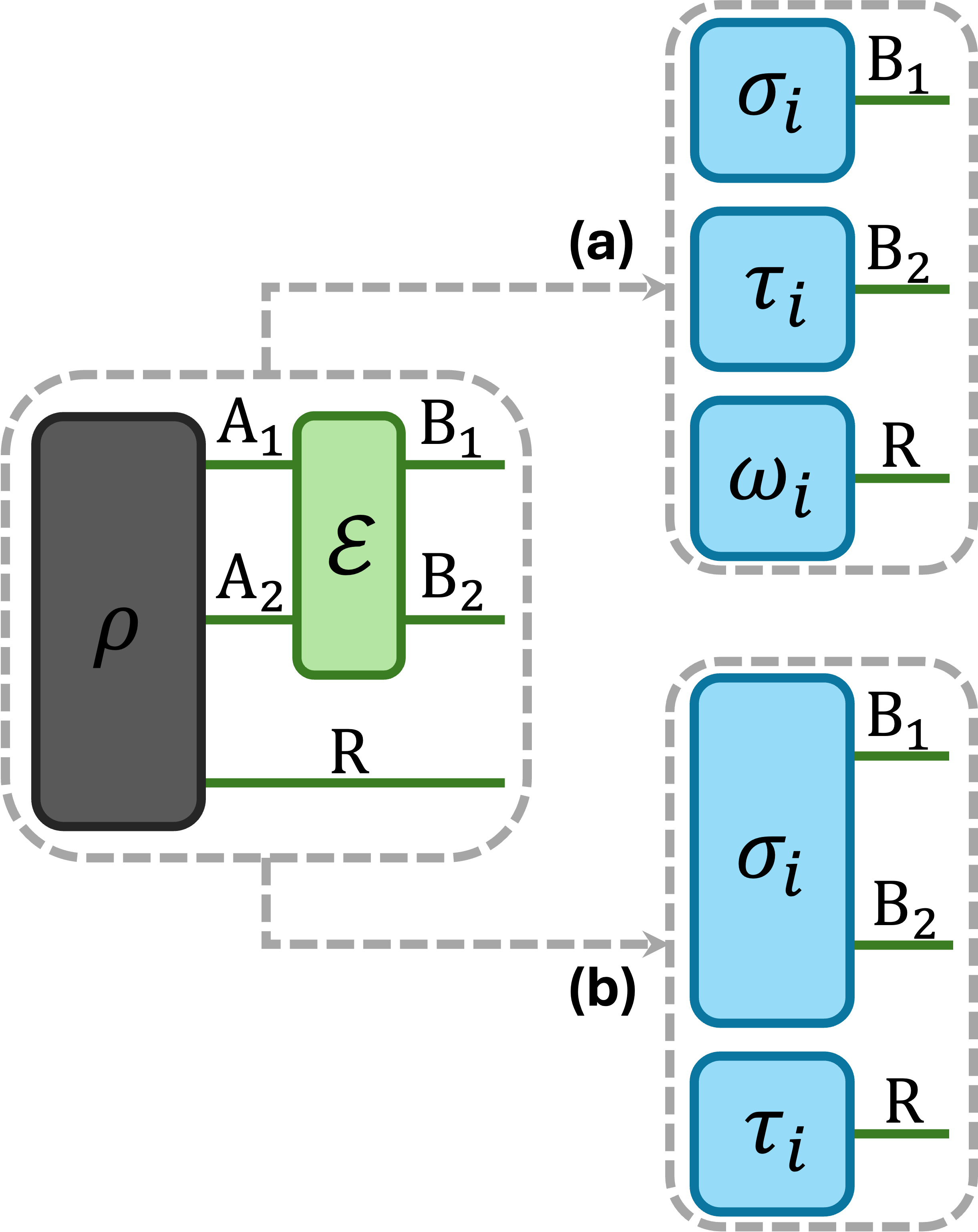}
\caption{(Color online) 
\textbf{Multipartite EB Channel}.  
A tripartite state $\rho_{RA_1A_2}$ is prepared, and the subsystems $A_1A_2$ are sent through a bipartite channel $\mE:A_1A_2\to B_1B_2$.
(a) Completely EB: the output is fully separable across all subsystems, and all quantum correlations are removed.
(b) EB: the subsystems $B_1B_2$ emerging from the channel $\mE$ is required to be separable from the reference system $R$, while correlations within $B_1B_2$ may still remain.
}
\label{fig:EBC_Multi}
\end{figure}

Consider the configuration shown in Fig.~\ref{fig:EBC_Multi}: 
a tripartite state $\rho_{RA_1A_2}$ is prepared, and the subsystem $A_1A_2$ is transmitted through a bipartite channel $\mE:A_1A_2\to B_1B_2$.
In this multipartite setting, two different notions of entanglement breaking arise.
The first -- completely entanglement breaking -- requires the output to be fully separable across all subsystems, ensuring that every quantum correlation is eliminated (see Fig.~\ref{fig:EBC_Multi}(a)). 
The second, which is operationally more meaningful and forms the basis of our development, only demands that the part of the state acted on by the channel become separable from the untouched reference system (see Fig.~\ref{fig:EBC_Multi}(b)). 
Under this weaker but practically relevant requirement, the processed subsystem may still retain internal correlations, but any entanglement with the ancillary system is necessarily destroyed.
In what follows, we take this second notion as the operative definition of entanglement breaking in multipartite settings, and use it as the basis for extending EB channels to superchannels.

Remark that, throughout this work, including Fig.~\ref{fig:EBC_Multi} and our subsequent subsections, we adopt the convention in which overall probability distributions are omitted.
For example, the separable state illustrated in Fig.~\ref{fig:EBC_Multi}(a), $\sum_i p_i\sigma_i\otimes\tau_i\otimes\omega_i$, is drawn without the explicit coefficients $p_i$, showing only the associated quantum states. 
All such diagrams implicitly correspond to properly normalized objects. 
This convention leaves all separability considerations intact, as separability is invariant under positive scalar rescaling. 
We therefore suppress these scalar factors to keep the diagrams visually clean and focused on their essential structure.

%%%%%%%%%%%%%%%%%%%%%%%%%%%%%%%%%%%%%%%%%%%%%%%%%%%%%%%%%%%%%%%%%%%%%%%
%%%%%%%%%%%%%%%%%%%%%%%%%%%%%%%%%%%%%%%%%%%%%%%%%%%%%%%%%%%%%%%%%%%%%%%

\subsection{Entanglement Breaking Superchannels}
\label{subsec:EBS}

The idea of an EB superchannel was first formulated in Ref.~\cite{Chen2020entanglement}, following the framework of Ref.~\cite{8678741} that treats a superchannel as a map taking one quantum channel to another. 
This viewpoint is mathematically precise and captures how a superchannel reshapes quantum dynamics. 
But it does not fully reflect the physical picture. 
A superchannel can also be seen as a bipartite process spread over two time steps, where an experimenter may interact with it more than once by sending different parts of a multipartite state at different rounds. 
From this operational perspective, the usual definition of an EB superchannel can behave in a counter-intuitive way: even when the superchannel is labeled EB in the channel-to-channel sense, some quantum correlations may still survive through multi-round interactions. 
This observation prompts us to revisit what EB should mean for superchannels in a physical and practically relevant way. 
The rest of this subsection develops this refined viewpoint, namely viewpoint ii in Sec.~\ref{subsec:Bi_Channels}.

We begin with the notion of an EB superchannel associated with viewpoint i in Sec.~\ref{subsec:Bi_Channels}, originally introduced in Ref.~\cite{Chen2020entanglement}.
For ease of reference, we refer to this class as the {\it Chen–Chitambar EB superchannels}.
Here the object under consideration is a bipartite channel $\mE:B_1R_1\to A_2R_2$, and the defining requirement is that the superchannel $\theta:A_1A_2\to B_1B_2$ must send any such $\mE$ to a bipartite channel that contains no dynamical entanglement~\cite{bauml2019resourcetheoryentanglementbipartite,PhysRevLett.125.180505,PhysRevA.103.062422}.
In this formulation, the output takes the form
\begin{align}\label{eq:EBS_CC}
    \theta(\mE)=\sum_i\mF_i\otimes\mG_i,
\end{align}
where the completely positive maps $\mF_i:B_1\to A_2$ and $\mG_i:R_1\to R_2$ jointly specify a valid channel $\theta(\mE)$.
Fig.~\ref{fig:CC_EBS} illustrates this picture: Chen–Chitambar EB is imposed directly on the bipartite dynamics, reflecting a channel-to-channel perspective.

\begin{figure}[t]
\centering   
\includegraphics[width=0.48\textwidth]{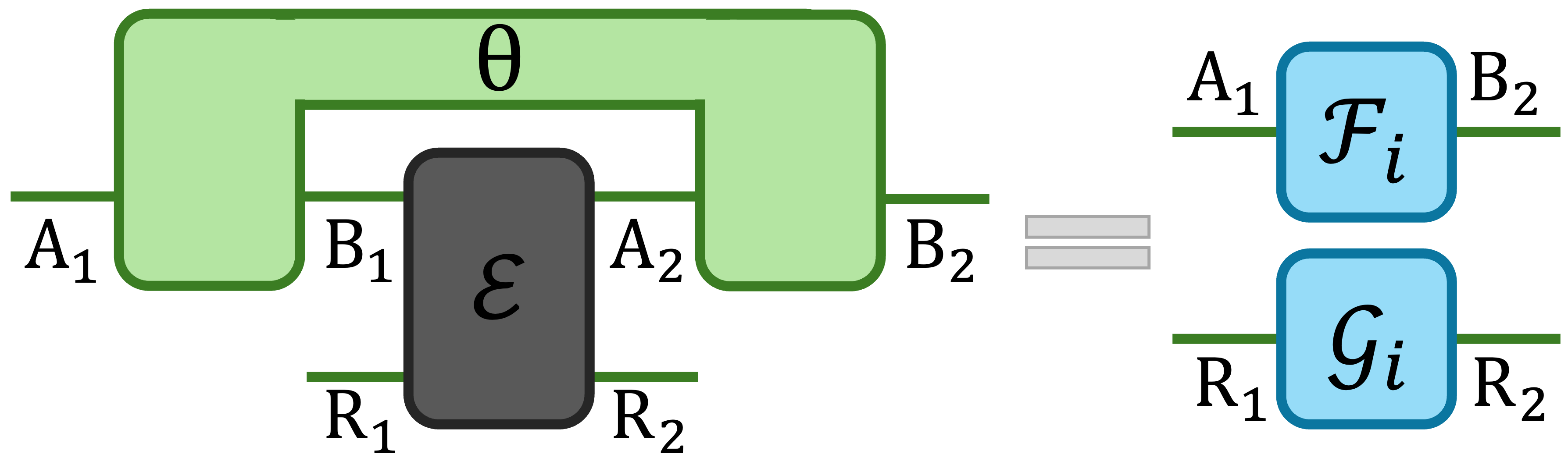}
\caption{(Color online) 
\textbf{Chen-Chitambar EB Superchannel}.  
This notion of EB superchannel ensures that any bipartite channel is mapped to a separable one. 
In particular, the image takes the form of Eq.~\eqref{eq:EBS_CC}.
In the depiction, the blue highlighting of 
$\mF_i$ and $\mG_i$ denotes a shared summation index $i$, consistent with our diagrammatic convention.
}
\label{fig:CC_EBS}
\end{figure}

By parallel with the EB channel case, a superchannel is Chen–Chitambar EB in this formulation exactly when its Choi operator is separable across the prescribed bipartition. 
Thus, Choi operator separability provides the necessary and sufficient characterization.

\begin{thm}
[\bf{Chen-Chitambar EB Superchannel}~\cite{Chen2020entanglement}]
\label{thm:EB_Superchannel_CC}
Let $\theta:A_1A_2\to B_1B_2$ be a superchannel and let $\mE:B_1R_1\to A_2R_2$ be any bipartite channel. 
The result channel $\theta(\mE)$ is separable across the bipartition $A_1B_2\,|\,R_1R_2$ for all admissible $\mE$ (see Fig.~\ref{fig:CC_EBS}) if and only if the Choi operator $J^{\theta}$ of the superchannel is separable across the partition $A_1B_2\,|\,B_1A_2$.
\end{thm}

Chen-Chitambar EB superchannels were originally characterized through the structure of the Gour operator in Eq.~\eqref{eq:Superchannel_Gour}.
Thm.~\ref{thm:ET} shows that this characterization is equivalently expressed at the level of the Choi operator, allowing us to analyze all Chen–Chitambar EB behavior entirely within the Choi representation (see Sec.~\ref{subsec:SC_Choi}). 
In this abstracted viewpoint, we classify a superchannel as {\it Type-I EB} when its Choi operator is separable across the bipartition $A_1B_2\,|\,B_1A_2$, as demonstrated in Fig.~\ref{fig:EBS_Type_1}.

\begin{figure}[t]
\centering   
\includegraphics[width=0.38\textwidth]{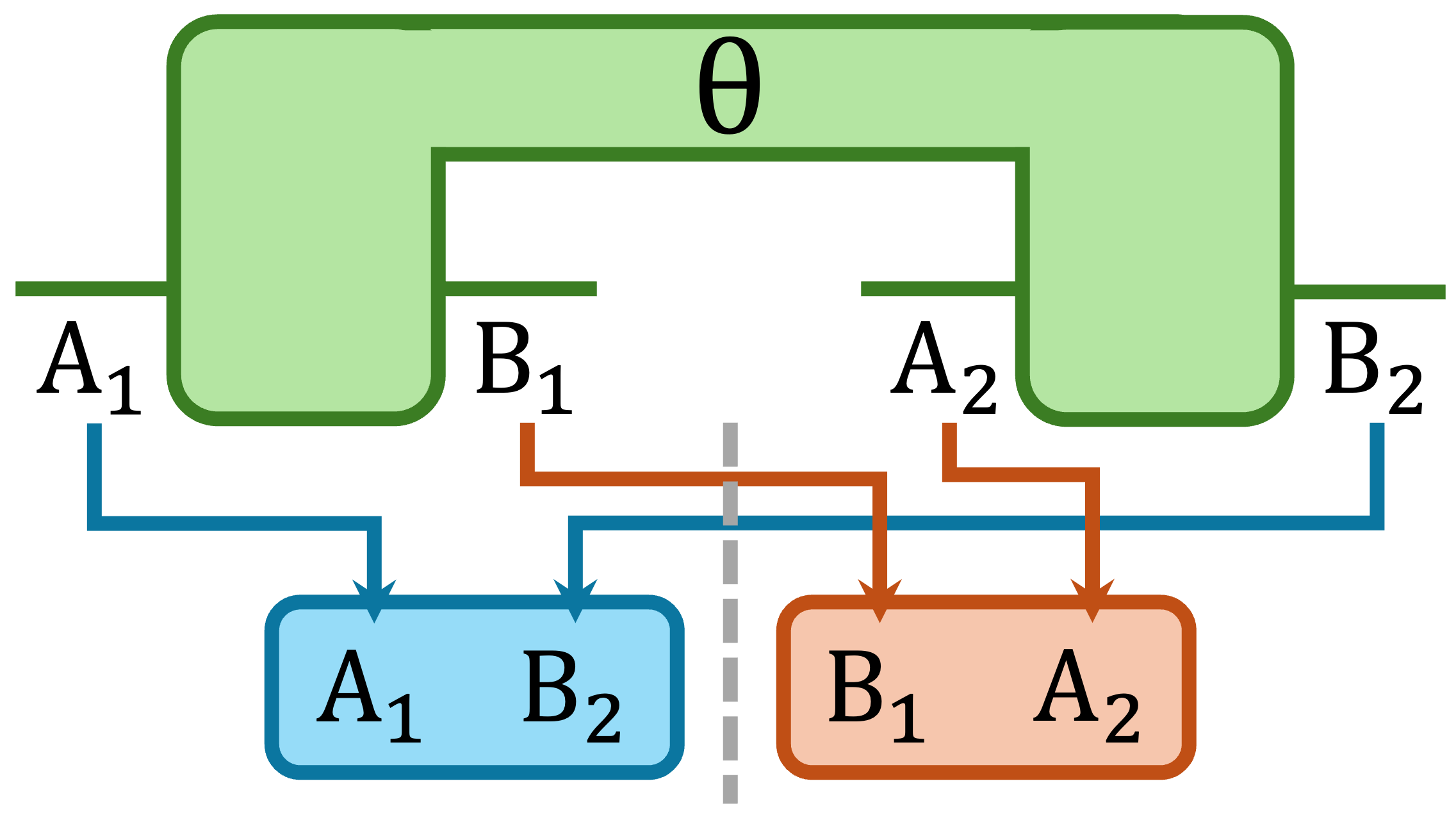}
\caption{(Color online) 
\textbf{Type-I EB Superchannel}.  
A superchannel $\theta$ is termed Type-I EB when its Choi operator is separable across the bipartition $A_1B_2\,|\,B_1A_2$.
}
\label{fig:EBS_Type_1}
\end{figure}

Let us now proceed to our notion of an EB superchannel by framing it through a scenario that closely parallels the verification of a practical quantum memory device~\cite{PhysRevX.8.021033}. 
Alice acts as the prospective buyer, while Bob, the service provider, claims that his device is capable of preserving quantum correlations across time. 
Before accepting this claim, Alice performs a two-step verification sketched in Fig.~\ref{fig:EBS}. 
She begins by preparing a multipartite state and sending subsystem $A_1$ into Bob's device. 
At a later time $t_{A_2}$, she sends another subsystem $A_2$. 
For the device to function as a genuine quantum memory, the final state held by Bob must remain entangled with her reference system $R$ after both interactions.
Equivalently, the output must be non-separable across the bipartition $R\,|\,B_1B_2$. 
Observing entanglement across this cut is precisely what certifies that Bob's device is capable of preserving quantum correlations over multiple time points.
In contrast, if every such two-step verification yields an output that is separable across $R\,|\,B_1B_2$, then the device operates as an EB superchannel -- a temporal analogue of the multipartite EB behavior shown in Fig.~\ref{fig:EBC_Multi}(b).

This operational test immediately highlights a fundamental difference between our notion of EB superchannel (see Fig.~\ref{fig:EBS}) and the Chen-Chitambar EB introduced in Ref.~\cite{Chen2020entanglement} (see Fig.~\ref{fig:CC_EBS}). 
Consider a superchannel built from a noiseless channel $\id_{A_1\to B_2}$, a fixed state preparation on $B_1$, and a trace over $A_2$. 
\begin{align}\label{TN:EBS_Ex_1}
    \raisebox{0ex}{\includegraphics[height=5em]{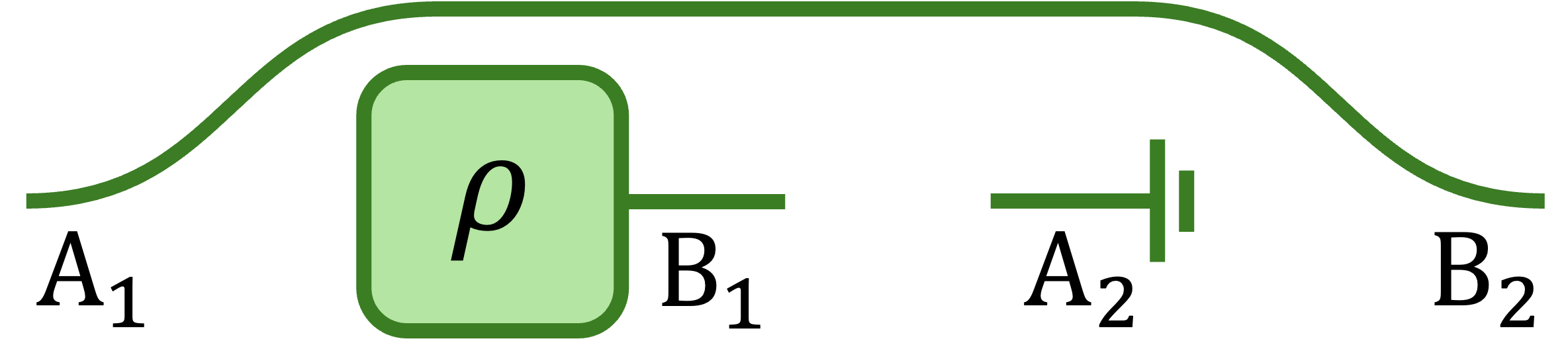}}
\end{align}
In the framework of Ref.~\cite{Chen2020entanglement}, such a construction is classified as Chen-Chitambar EB;
indeed, it satisfies the Type-I EB condition by design (see Fig.~\ref{fig:EBS_Type_1}). 
Under our operational criterion, however, the same superchannel is not EB. 
By preparing a maximally entangled state on $RA_1$ and an arbitrary state on $A_2$, Alice can ensure that the correlations between $R$ and $A_1$ reappear perfectly at the output. 
In other words, despite being Type-I EB from the channel-to-channel viewpoint (see Sec.~\ref{subsec:Bi_Channels}), this superchannel demonstrably preserves entanglement and would be certified as a functioning quantum memory. 
This discrepancy motivates our refined definition, which captures the physically relevant notion of EB superchannel in multi-round quantum interactions.

\begin{figure}[t]
\centering   
\includegraphics[width=0.48\textwidth]{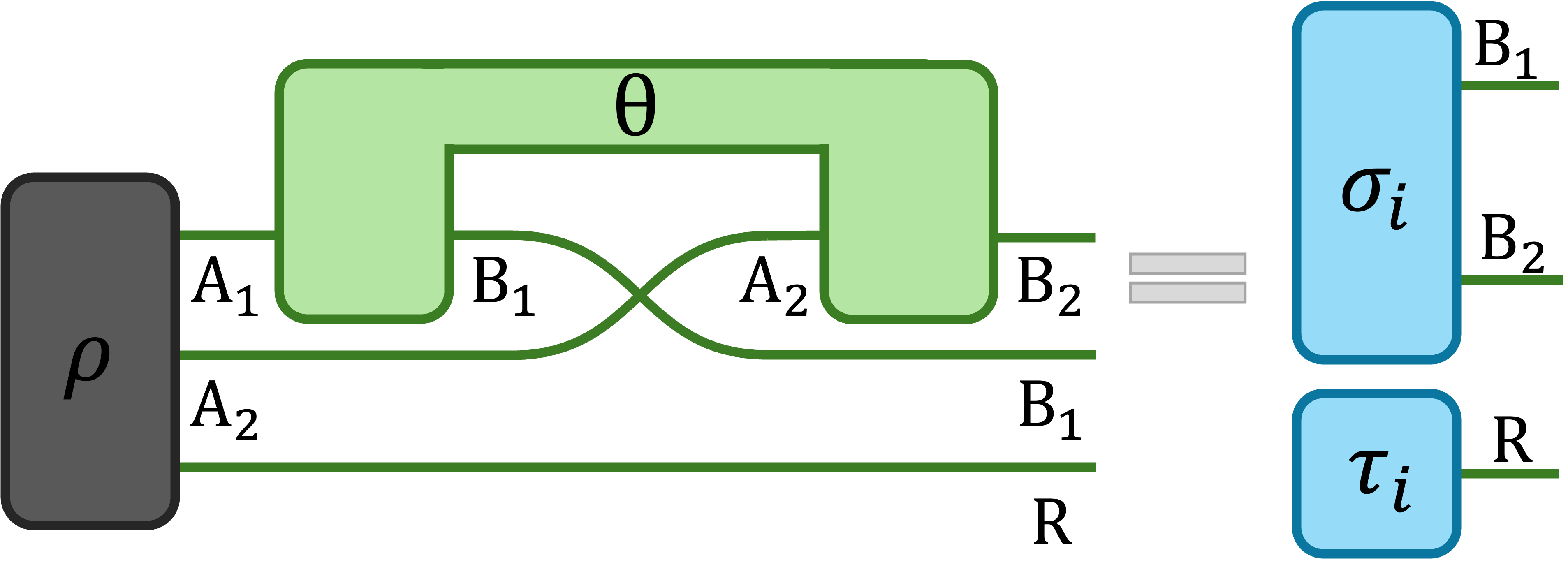}
\caption{(Color online) 
\textbf{EB Superchannel}.  
Our notion of an EB superchannel is motivated by the two-step verification of quantum-memory devices. 
Alice prepares a tripartite state in which two subsystems are used to probe the device operated by Bob, while a reference system is retained to test whether quantum correlations survive across the two rounds of interaction.
In the first round, she sends subsystem $A_1$ into Bob's device, and in the second round she transmits subsystem $A_2$. 
Bob's device, in this setting, is described by a superchannel $\theta$. 
We say that $\theta$ is EB if the resulting output state is always separable across the bipartition $R\,|\,B_1B_2$.
Otherwise, if entanglement with $R$ can be preserved for some input state, Bob's device qualifies as a genuine quantum memory.
}
\label{fig:EBS}
\end{figure}

Analogous to the formulation of Chen-Chitambar EB superchannels in Fig.~\ref{fig:CC_EBS}, our refined notion admits a complete characterization directly in terms of the Choi operator of the superchannel. 
Formally, we have

\begin{thm}
[\bf{EB Superchannel}]
\label{thm:EB_Superchannel}
Let $\theta:A_1A_2\to B_1B_2$ be a superchannel, and let $\rho_{RA_1A_2}$ be any tripartite state in which the subsystems $A_1$ and $A_2$ interact with $\theta$ in two successive rounds.
The resulting state is separable across the bipartition $R\,|\,B_1B_2$ for all inputs $\rho$ (see Fig.~\ref{fig:EBS}) if and only if the Choi operator $J^{\theta}$ is separable across the partition $A_1A_2\,|\,B_1B_2$.
\end{thm}

\begin{proof}
We begin with the sufficiency.
Assume that superchannel $\theta$ is EB. 
By definition (see Eq.~\eqref{eq:Superchannel_Choi_1}), applying $\theta$ to the maximally entangled input $\Gamma_{A_1A_1}\otimes\Gamma_{A_2A_2}$ produces its Choi operator $J^{\theta}$ that is separable across the bipartition $A_1A_2\,|\,B_1B_2$.
This establishes the ``if'' direction.
We now turn to the necessity.
Assume that the Choi operator $J^{\theta}$ can be written in a separable form $J^{\theta}_{A_1A_2B_1B_2}=\sum_iX_i\otimes Y_i$, where operators $X_i$ and $Y_i$ are acting on systems $A_1A_2$ and $B_1B_2$, respectively. 
Substituting this decomposition into $J^{\theta}\star\rho_{RA_1A_2}$ gives $\sum_iX_i\star\rho_{RA_1A_2}\otimes Y_i$ making it clear that the final state is separable across the bipartition $R\,|\,B_1B_2$.
This establishes the ``only-if'' direction and completes the proof.
\end{proof}

Comparing Thms.~\ref{thm:EB_Superchannel_CC} and~\ref{thm:EB_Superchannel} highlights the core physical distinction between the Chen-Chitambar EB superchannel defined in Ref.~\cite{Chen2020entanglement}, based on the higher-order transformation viewpoint (perspective i in Sec.~\ref{subsec:Bi_Channels}), and the definition we adopt here, rooted in the bipartite channel perspective (perspective ii in Sec.~\ref{subsec:Bi_Channels}).
The distinction stems from the fact that the Choi operator of a superchannel can be separable with respect to different bipartitions. 
In this work, we refer to those superchannels whose Choi operator is separable across the partition $A_1A_2\,|\,B_1B_2$ as {\it Type-II EB}, as illustrated in Fig.~\ref{fig:EBS_Type_2}.

More importantly, although a superchannel is formally defined as a map from channels to channels, this does not imply that an experimenter must probe it by preparing a quantum quantum channel and feeding it in as the input. 
In real experiments, especially when investigating non-Markovian dynamics, one learns about the process by sending in quantum states step by step. 
This practical mode of access naturally gives rise to alternative, and operationally distinct, notions of when a superchannel should be regarded as EB.

\begin{figure}[t]
\centering   
\includegraphics[width=0.38\textwidth]{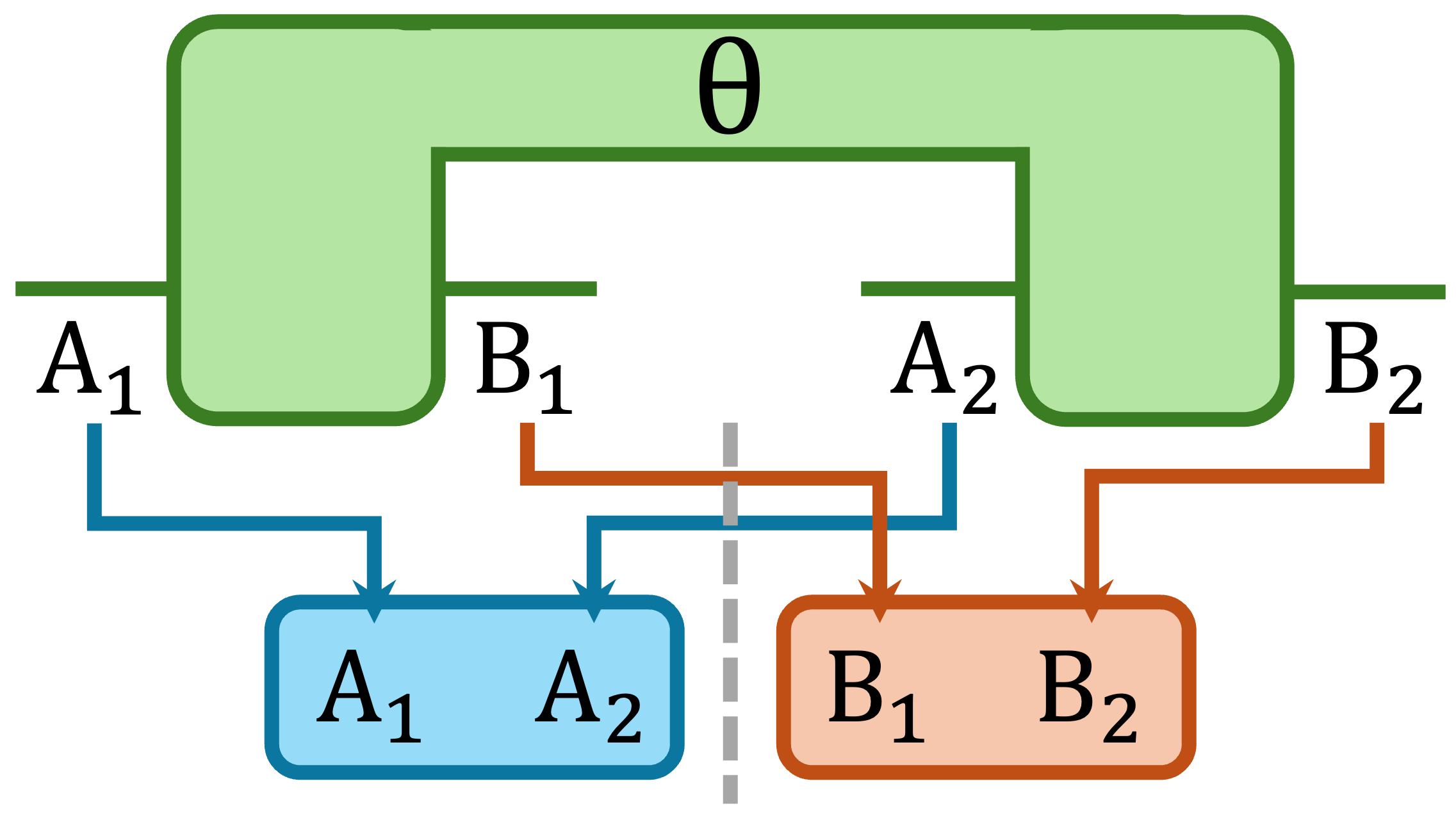}
\caption{(Color online) 
\textbf{Type-II EB Superchannel}.  
We call a superchannel $\theta$ Type-II EB when its Choi operator is separable across the bipartition $A_1A_2\,|\,B_1B_2$.
}
\label{fig:EBS_Type_2}
\end{figure}

Here we are not only concerned with introducing a new mathematical definition of an EB superchannel (see Fig.~\ref{fig:EBS}); 
we also want to understand how such a transformation can be implemented in practice. 
Given the Choi operator of the superchannel $\theta$, the framework developed in this work allows us to extract its Kraus representation and thereby determine the induced map 
$f_{\theta}$ (see Eq.~\eqref{TN:RT_1}).
This, in turn, specifies the corresponding pre-processing and post-processing channels shown in Fig.~\ref{fig:Superchannel}.

Beyond this structural route, one may also seek an explicit implementation directly from a separable decomposition of the Choi operator, much like the channel construction in Eq.~\eqref{TN:EBC_Realization}. 
Without loss of generality, suppose the EB superchannel $\theta$ has a Choi operator of the form (see Thm.~\ref{thm:EB_Superchannel})
\begin{align}
    J^{\theta}_{A_1A_2B_1B_2}=\sum_iX_i\otimes Y_i,
\end{align}
where $X_i$ acts on $A_1A_2$ and $Y_i$ acts on $B_1B_2$.
As a first step toward a physical realization, we reorganize the Choi operator into a form that makes its structure more transparent. 
For each term, we introduce a normalized state $\sigma_i$
\begin{align}
    \sigma_i:=\frac{Y_i}{\Tr[Y_i]},
\end{align} 
and define the associated operator $M_i$ accordingly:
\begin{align}
    M_i:=\Tr[Y_i]X_i.
\end{align}
With these definitions, the Choi operator becomes
\begin{align}
    J^{\theta}_{A_1A_2B_1B_2}=\sum_i M_i\otimes\sigma_i,
\end{align}
and, thanks to the trace-preserving property of $\theta$, the operators $\{M_i\}$ form a POVM acting on systems $A_1A_2$. 

Based on the measurement $\{M_i\}$ and states $\sigma_i$, we obtain the circuit that realizes the superchannel, 
\begin{align}\label{TN:EBS_Realization}
    \raisebox{0ex}{\includegraphics[height=9em]{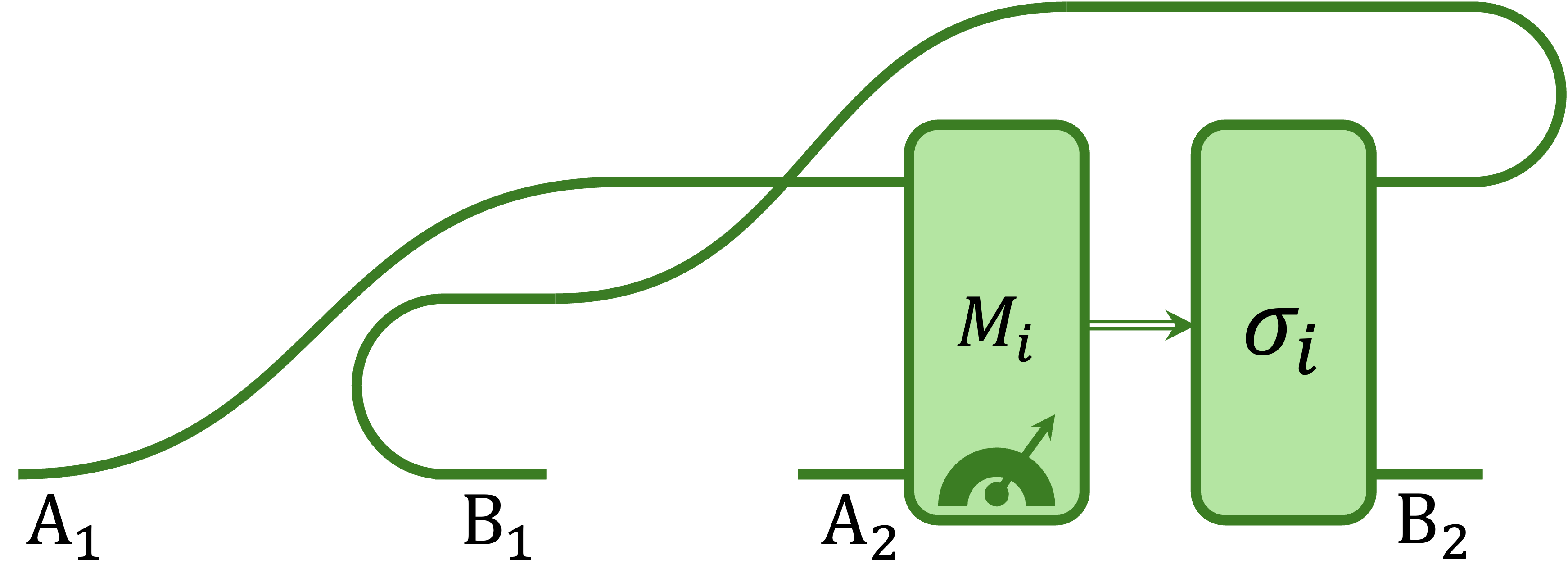}}
\end{align}
extending the channel-level construction in Eq.~\eqref{TN:EBC_Realization} to the multiple time points setting.
In this structure, the post-processing stage is implemented through Bell measurements.
Achieving the same behavior through a fully deterministic operation, i.e., a CPTP map without post-selection, would require additional ingredients that fall outside the scope of this work.
We therefore leave the development of such deterministic realizations to future investigations.

To this point, we have introduced a new way to define EB superchannels, which serves as a temporal counterpart of the multipartite EB structure in Fig.~\ref{fig:EBC_Multi}(b). The construction in Fig.~\ref{fig:EBC_Multi}(a) also suggests a stronger notion: a completely EB superchannel, one that always yields a state separable across all systems, independent of the state Alice prepares in the verification test. 
Looking beyond entanglement breaking, it is natural to consider {\it entanglement annihilation}~\cite{Moravcikova_2010,PhysRevA.88.032316,Aubrun2023} in processes that unfold over multiple time steps, pointing to a richer structure of temporal quantum correlations that calls for deeper examination.

%%%%%%%%%%%%%%%%%%%%%%%%%%%%%%%%%%%%%%%%%%%%%%%%%%%%%%%%%%%%%%%%%%%%%%%
%%%%%%%%%%%%%%%%%%%%%%%%%%%%%%%%%%%%%%%%%%%%%%%%%%%%%%%%%%%%%%%%%%%%%%%

\subsection{Common Cause Breaking Superchannels}
\label{subsec:CCBS}

``Correlation does not imply causation'' ~\cite{Pearl_2009}. 
In quantum theory, correlations can exhibit phenomena with no classical counterpart, most notably entanglement~\cite{RevModPhys.81.865}, Einstein-Podolsky-Rosen steering~\cite{RevModPhys.92.015001}, and Bell nonlocality~\cite{RevModPhys.86.419}. 
Causality itself is also richer in the quantum domain~\cite{Ried2015,Fitzsimons2015}. 
For instance, whereas direct causes and common causes can only be combined probabilistically in classical systems, quantum mechanics allows them to be superposed coherently~\cite{MacLean2017}. 
Quantum theory further provides tools that have no classical analogue for probing causal structure, such as uncertainty relations tailored to causal inference~\cite{PhysRevLett.130.240201}. 
Earlier studies have largely focused on how to break quantum correlations; a natural question is whether similar manipulations can break quantum causal structure.
In this section, paralleling our earlier discussions in Secs.~\ref{subsec:EBC},~\ref{subsec:MEBC}, and~\ref{subsec:EBS}, we introduce the notion of a common cause breaking superchannel.

\begin{figure}[t]
\centering   
\includegraphics[width=0.48\textwidth]{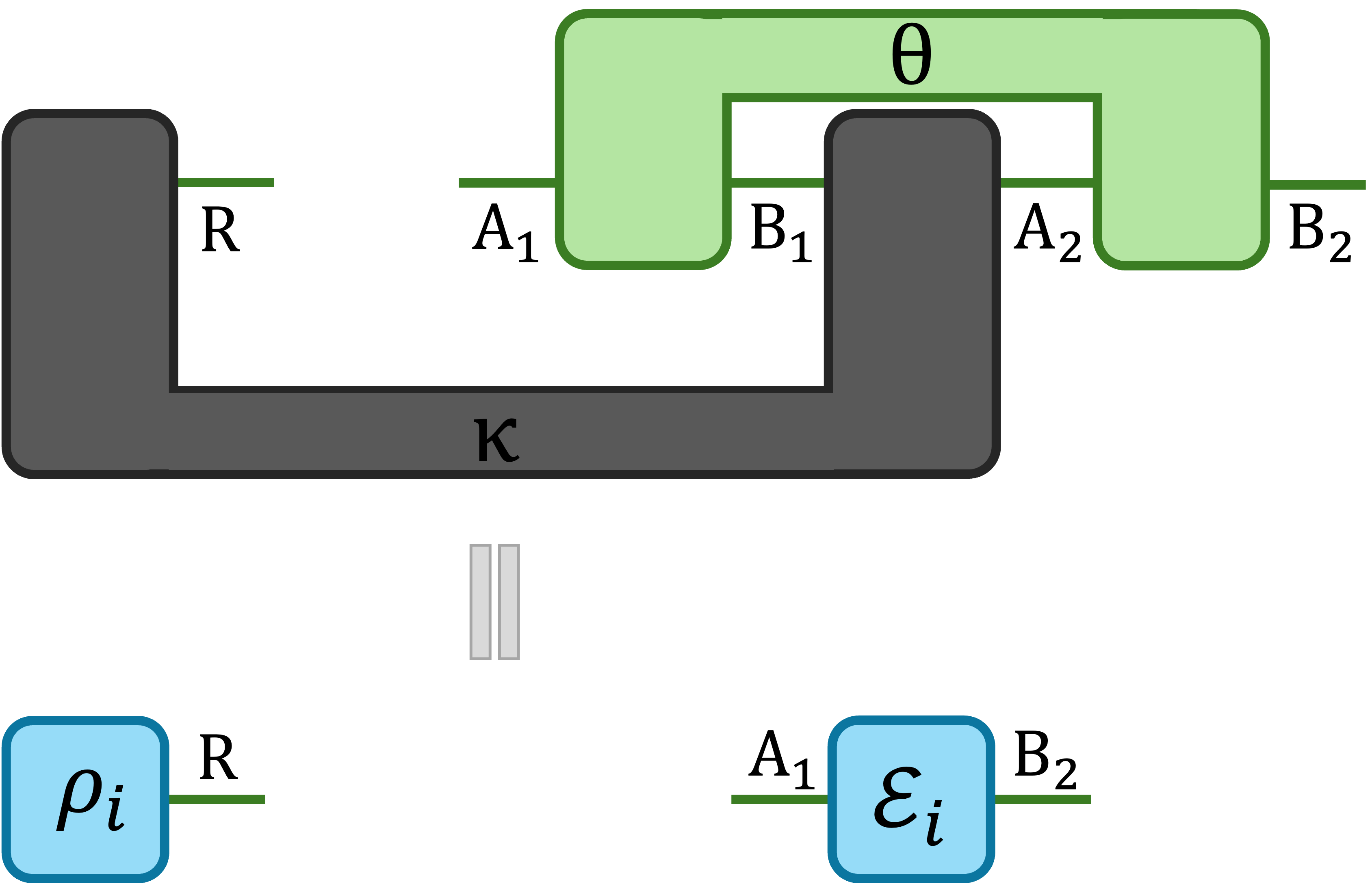}
\caption{(Color online) 
\textbf{Common Cause Breaking Superchannel}.  
A common cause breaking superchannel $\theta$ converts any causal map $\kappa$ into a state preparation on $R$ followed by a channel from $A_1$ to $B_2$, with the two stages linked solely by classical memory. 
The resulting structure is precisely the one shown in Eq.~\eqref{eq:EBS_CC}. In the diagram, the blue marking of 
$\rho_i$ and $\mE_i$ indicates that they share the same classical index $i$, following our diagrammatic convention.
}
\label{fig:CCBS}
\end{figure}

The object of interest is the causal map $\kappa$, depicted in black on the left side of Fig.~\ref{fig:CCBS}. 
Mathematically, $\kappa$ is a special instance of a superchannel in which the pre-processing stage has a trivial input.
We may further process $\kappa$ with another superchannel $\theta$ acting on its second quantum process. 
When $\theta$ is common cause breaking, it removes the quantum linkage between the two stages of $\kappa$: after its action, the global transformation is separable across 
$R\,|\,A_1B_2$ and effectively reduces to preparing a state $\rho_i$ on $R$ followed by a channel $\mE_i$ from 
$A_1$ to $B_2$, correlated only through a classical label; that is
\begin{align}
    \theta(\kappa)=\sum_ip_i\rho_i\otimes\mE_i.
\end{align}
In this work we restrict attention to breaking the quantum component of the common cause; 
the remaining classical index connecting the prepared state and the subsequent channel reflects a purely classical common cause structure.

The common cause breaking condition admits a direct characterization in terms of the superchannel's Choi operator $J^{\theta}$.
We state the resulting theorem below; its proof follows the same structural pattern as that of Thm.~\ref{thm:EB_Superchannel}.

\begin{thm}
[\bf{Common Cause Breaking Superchannel}]
\label{thm:CCBS}
Consider a superchannel $\theta:A_1A_2\to B_1B_2$ and an arbitrary causal map $\kappa:B_1\to RA_2$ through which the subsystems $B_1$ and $A_2$ interact with $\theta$.
The superchannel $\theta$ is said to be common cause breaking, meaning that the resulting dynamics is always separable across the bipartition $R\,|\,A_1B_2$ for every input $\kappa$ (see Fig.~\ref{fig:CCBS}) if and only if its Choi operator $J^{\theta}$ is separable across the partition $A_1B_2\,|\,B_1A_2$.
\end{thm}

These observations place common cause breaking superchannels within our Type-I EB framework (see Fig.~\ref{fig:EBS_Type_1}), providing a direct bridge between quantum causal inference and the dynamical entanglement viewpoint developed here.
We continue to use the term ``causal map'' to follow the original terminology introduced in Ref.~\cite{Ried2015}, where this form of quantum dynamics was first studied. 
For readers seeking a fuller account of causal maps and their relation to quantum entanglement, Ref.~\cite{Milz2021} provides an in-depth discussion.

%%%%%%%%%%%%%%%%%%%%%%%%%%%%%%%%%%%%%%%%%%%%%%%%%%%%%%%%%%%%%%%%%%%%%%%
%%%%%%%%%%%%%%%%%%%%%%%%%%%%%%%%%%%%%%%%%%%%%%%%%%%%%%%%%%%%%%%%%%%%%%%

\section{Discussions}
\label{sec:Discussions}

Quantum channels offer a unified framework for describing state preparation, gate implementation, and measurement readout, supported by a well-established suite of analytical and computational tools. 
The Choi–Jamio\l kowski isomorphism (see Sec.~\ref{subsec:Choi}) streamlines optimization tasks, Kraus operators (see Sec.~\ref{subsec:Kraus}) capture physical noise across hardware platforms, and Stinespring dilation (see Sec.~\ref{subsec:Stinespring}) underlie master equations and non-Markovian dynamics. 
The Liouville superoperator (see Sec.~\ref{subsec:Liouville}) is equally essential for simulating noisy circuits, and extracting decoherence rates.
Together, these representations form a powerful toolkit for quantum information processing. 
A qualitative shift arises, however, when we move from analyzing channels to manipulating them: superchannels operate at a higher tier of dynamics, where the familiar tools must be extended and reinterpreted.

In this work, we resolve the inconsistency between competing Choi representations of superchannels (see Thm.~\ref{thm:ET}) and use this resolution as the foundation for a complete and operationally grounded toolkit for higher-order quantum dynamics, encompassing the Kraus decomposition (see Thm.~\ref{thm:Superchannel_Krasu}), Stinespring dilation (see Thm.~\ref{thm:Superchannel_Stinespring}), and Liouville superoperator (see Thm.~\ref{thm:Superchannel_Liouville}) for superchannels. 
Building on this framework, we develop new approaches to characterizing entanglement breaking and common cause breaking superchannels, exposing structural features that were previously inaccessible. 
Our guiding principle is the generalized Occam's razor introduced in Sec.~\ref{subsec:Bi_Channels}, which favors formulations that remain minimal while fully consistent with physical constraints and established theories. 
These advances place superchannels on a coherent footing and open the way to a systematic study of adaptive and multi-time quantum dynamics. 
As our framework includes the standard toolkit for quantum channels as a special case, we expect it to have broad implications for understanding memory effects and other non-Markovian phenomena in quantum systems.

Technically, our results are grounded in a tensor-network formulation that not only yields new structural insights into superchannels, but also provides an alternative and markedly simplified proof of the realization theorem. 
This framework further offers a clear and intuitive way to determine the minimal memory required to simulate a given superchannel. 
Because the essential object in our construction is a linear transformation from CP maps to CP maps, the same techniques extend naturally to settings with multiple time steps and higher-order transformation hierarchies, offering a coherent path toward understanding general non-Markovian dynamics~\cite{PhysRevLett.101.150402,RevModPhys.89.015001,PhysRevA.97.012127,LI20181,PRXQuantum.2.030201} and even processes with indefinite causal order~\cite{Oreshkov2012,PhysRevA.88.022318}. 
At the same time, tensor networks occupy a central place in quantum many-body physics~\cite{Verstraete01032008,ORUS2014117,Orus2019,RevModPhys.93.045003}, numerical simulation~\cite{Markov2008,PhysRevLett.128.030501,PhysRevLett.129.090502}, and quantum error correction~\cite{PhysRevLett.113.030501,PhysRevLett.119.040502,PhysRevLett.127.040507,PRXQuantum.3.020332}; our formulation points to dynamical analogues in these areas, where memory effects become a defining structural feature rather than an incidental complication. 
Taken together, these observations reveal a broader landscape that is only beginning to be explored and that our results help to open.

%%%%%%%%%%%%%%%%%%%%%%%%%%%%%%%%%%%%%%%%%%%%%%%%%%%%%%%%%%%%%%%%%%%%%
% \noindent \textbf{Acknowledgments}--
\section*{Acknowledgments}

Yunlong Xiao acknowledges support from A*STAR under its Career Development Fund (C243512002).

%%%%%%%%%%%%%%%%%%%%%%%%%%%%%%%%%%%%%%%%%%%%%%%%%%%%%%%%%%%%%%%%%%%%%

\bibliography{Bib}
\end{document}